\documentclass[11pt,letterpaper]{article}

\usepackage[margin=1in]{geometry}

\usepackage{amsfonts, amsmath, amssymb, amsthm}
\usepackage{mathrsfs} % \mathscr{P} or \mathscr{B}

\usepackage{hyperref}
\hypersetup{
    colorlinks=true,
    linkcolor=blue,
    citecolor=blue,
    urlcolor=blue,
}
\usepackage[numbers]{natbib}
\usepackage{graphicx}           
\usepackage{xcolor}
\usepackage{subfig}
\usepackage{authblk}
\usepackage{algorithm}
\usepackage{algorithmic}
\usepackage{times}

\newtheorem{theorem}{Theorem} % \newtheorem{theorem}{Theorem}[section] for numbering by sections		
\newtheorem{lemma}{Lemma} % \newtheorem{lemma}[theorem]{Lemma} to follow Theorem numbering
\newtheorem{proposition}{Proposition}

\theoremstyle{definition}
\newtheorem{definition}{Definition}	
\newtheorem{remark}{Remark}
\newtheorem{example}{Example}
\newtheorem{assumption}{Assumption}

\numberwithin{equation}{section}

\newcommand{\cA}{\mathcal{A}}

\newcommand{\cF}{\mathcal{F}}

\newcommand{\cN}{\mathcal{N}}

\newcommand{\cQ}{\mathcal{Q}}
\newcommand{\cS}{\mathcal{S}}

\newcommand{\cX}{\mathcal{X}}
\newcommand{\cY}{\mathcal{Y}}

\newcommand{\N}{\mathbb{N}}
\newcommand{\R}{\mathbb{R}}

\newcommand{\E}{\mathbb{E}} % no spacing
\renewcommand{\P}{\mathbb{P}}

\DeclareMathOperator*{\argmin}{arg\,min} % \, for spacing

\renewcommand{\epsilon}{\varepsilon}

\title{Robust Point Matching with Distance Profiles}
\author[1]{YoonHaeng Hur}
\author[1]{Yuehaw Khoo}
\affil[1]{Department of Statistics, University of Chicago}

\begin{document}
\maketitle
\begin{abstract}
    Computational difficulty of quadratic matching and the Gromov-Wasserstein distance has led to various approximation and relaxation schemes. One of such methods, relying on the notion of distance profiles, has been widely used in practice, but its theoretical understanding is limited. By delving into the statistical complexity of the previously proposed method based on distance profiles, we show that it suffers from the curse of dimensionality unless we make certain assumptions on the underlying metric measure spaces. Building on this insight, we propose and analyze a modified matching procedure that can be used to robustly match points under a certain probabilistic setting. We demonstrate the performance of the proposed methods using simulations and real data applications to complement the theoretical findings. As a result, we contribute to the literature by providing theoretical underpinnings of the matching procedures based on distance invariants like distance profiles, which have been widely used in practice but rarely analyzed theoretically.
\end{abstract}

\section{Introduction}
\label{sec:intro}

Matching is one of the most fundamental tasks of modern data science, which can be found in numerous applications. For instance, the goal of shape correspondence problems \citep{van2011survey,sahilliouglu2020recent} in computer graphics is to match corresponding parts of different geometric shapes, often obtained by 3D scanners or imaging equipment used for Magnetic Resonance Imaging (MRI) or Computed Tomography (CT). Graph matching aims to find corresponding nodes between two or more graphs representing structured data, such as neuronal connectivity in human brains \citep{sporns2005human} or social networks \citep{borgatti2009network}. 

One formal mathematical framework for matching \citep{memoli2011gromov} is to conceptualize objects as metric measure spaces, which are essentially probability measures defined on suitable metric spaces. In practice, of course, objects are not fully available and thus approximated by suitable point clouds, say, $\{X_1, \ldots, X_n\}$ and $\{Y_1, \ldots, Y_m\}$ on metric spaces $(\cX, d_\cX)$ and $(\cY, d_\cY)$, respectively, and the practical goal is to match these point clouds. Quadratic matching is one of the most standard approaches to this problem, which formulates and solves a quadratic program based on pairwise distances $\{d_\cX(X_i, X_\ell)\}_{i, \ell = 1}^{n}$ and $\{d_\cY(Y_j, Y_\ell)\}_{j, \ell = 1}^{m}$. When $n = m$, quadratic matching finds a one-to-one correspondence that aligns the pairwise distances as closely as possible, which is formulated as a Quadratic Assignment Problem (QAP) \citep{koopmans1957assignment,loiola2007survey}. When $n \neq m$, the notion of correspondence is relaxed to soft assignments or couplings, leading to the formulation of the Gromov-Wasserstein distance \citep{memoli2011gromov}, which has been increasingly used in applications, such as word alignment \citep{alvarez2018gromov}, statistical inference \citep{brecheteau2019statistical,weitkamp2022distribution}, and generative modeling \citep{bunne2019learning,hur2024reversible}. On the practical side, however, quadratic matching requires solving a quadratic program that is NP-hard in general, and thus is approximated by several methods, such as convex relaxation \citep{aflalo2015convex}, semidefinite relaxation \citep{zhao1998semidefinite,kezurer2015tight}, spectral relaxation \citep{leordeanu2005spectral}, or entropic regularization \citep{rangarajan1999convergence,solomon2016entropic,peyre2016gromov}.

Matching procedures based on pairwise distances $\{d_\cX(X_i, X_\ell)\}_{i, \ell = 1}^{n}$ and $\{d_\cY(Y_j, Y_\ell)\}_{j, \ell = 1}^{m}$ are invariant to isometries, which is desirable in practice. For instance, in imaging problems, the images of the objects of interest are often obtained from different angles and thus need to be matched up to rigid transformations. However, such procedures are difficult to analyze theoretically and the literature has been mostly studying theoretical guarantees for other matching procedures based on $\|X_i - Y_j\|_2$'s when $(\cX, d_\cX) = (\cY, d_\cY)$ is the standard Euclidean space. In this setting, the point clouds lie in the same space and are often viewed as perturbations of one another, where one aims to match the points based on the proximity. \cite{collier2016minimax,kunisky2022strong,wang2022random} establish theoretical results on recovering the ground truth matching for the procedures based on $\|X_i - Y_j\|_2$'s under some noise, and \cite{galstyan2022optimal,minasyan2023matching,galstyan2023optimality} show matching guarantees in the presence of outliers. Although Theorem 2 of \cite{wang2022random} concerns a procedure using the pairwise distances in the Euclidean setting, it requires solving a nontrivial minimization over permutations, which is computationally challenging, and the result does not consider the presence of outliers.

Motivated by the above challenges, this paper theoretically analyzes matching procedures based on the distances $\{d_\cX(X_i, X_\ell)\}_{i, \ell = 1}^{n}$ and $\{d_\cY(Y_j, Y_\ell)\}_{j, \ell = 1}^{m}$, which are easily implementable in practice. The idea is to associate each point with the set of distances from it to other points, namely, $X_i$ and $Y_j$ are associated with the sets of distances to the other points $\{d_\cX(X_i, X_\ell)\}_{\ell = 1}^{n}$ and $\{d_\cY(Y_j, Y_\ell)\}_{\ell = 1}^{m}$, which are the distance profiles of $X_i$ and $Y_j$, respectively. The matching procedure is to find for each $X_i$, the point $Y_j$ having the most similar distance profile to that of $X_i$, where the similarity is measured by a suitable quantity between the sets $\{d_\cX(X_i, X_\ell)\}_{\ell = 1}^{n}$ and $\{d_\cY(Y_j, Y_\ell)\}_{\ell = 1}^{m}$. The idea of associating a point with suitable invariants based on the distances to other points from it has a long history in the object matching literature, which includes shape contexts \citep{belongie2002shape,ruggeri2008isometry}, integral invariants \citep{manay2006integral}, invariant histograms \citep{brinkman2012invariant}, the local distribution of distances \citep{memoli2011gromov}, and other related concepts \citep{gortler2019generic,norelli2023asif}. However, theoretical results on robustness of such methods are rarely found in the literature. In the context of random graph matching, \cite{ding2021efficient,ding2023efficiently} propose and analyze a degree profile matching procedure based on the adjacency matrices $(A_{i \ell})_{i, \ell = 1}^{n}$ and $(B_{j \ell})_{j, \ell = 1}^{n}$ of certain correlated random graphs. However, their analysis focuses on random graph models, such as Erd{\"o}s-R{\'e}nyi and correlated Wigner models, where the edges are independently generated, namely, $(A_{i \ell})_{i, \ell = 1}^{n}$ are mutually independent, which is substantially different from our setting where the distances $\{d_\cX(X_i, X_\ell)\}_{i, \ell = 1}^{n}$ are not mutually independent.

\paragraph{Organization} Section \ref{sec:foundations} presents mathematical background on the Gromov-Wasserstein distance and its computational alternative based on the distance profile, together with the new sample complexity result in Section \ref{sec:connection_GW_TLB}. Motivated by the insights built in Section \ref{sec:foundations}, Section \ref{sec:main} proposes and studies the partial matching algorithm based on distance profiles, where we formally show its robustness to noise and outlier components, together with the noise stability analysis for the one-to-one matching procedure. Section \ref{sec:simulations} demonstrates the performance of the proposed methods using simulations. All technical proofs, lemmas, extensions are deferred to Appendices \ref{sec:proofs}, \ref{sec:lemmas}, \ref{sec:extensions}, respectively, to streamline the main text.

\paragraph{Notation} For $k \in \N$, let $[k]$ denote $\{1, \ldots, k\}$, $\Delta_k = \{(p_1, \ldots, p_k) \in \R_+^k : \sum_{i = 1}^{k} p_i = 1\}$ denote the probability simplex, and $\cS_k$ denote the set of all permutations on $[k]$. For $a, b \in \R$, let $a \vee b = \max\{a, b\}$ and $a \wedge b = \min\{a, b\}$. Let $\mathscr{P}(\R^d)$ denote the set of all Borel probability measures defined on $\R^d$ and for any $z \in \R^d$, let $\delta_z \in \mathscr{P}(\R^d)$ denote the Dirac measure at $z$. We denote by $W_p$ the Wasserstein-$p$ distance for $p \ge 1$. For any measure $\mu$ on a measurable space $\cX$ and a measurable function $f \colon \cX \to \R$, let $\mu f = \int_{\cX} f \, \mathrm{d} \mu$.

\section{Foundations: Distance Profile Distance and Sample Complexity}
\label{sec:foundations}
As mentioned in Section \ref{sec:intro}, computational difficulty of the Gromov-Wasserstein (GW) distance serves as a major obstacle in practical applications, which has led to the development of various approximation and relaxation schemes. One of such methods, proposed as a tractable lower bound to the GW distance by \cite{memoli2011gromov}, utilizes the notion of distance profiles. We delineate the mathematical foundations of the distance profile in connection with the GW distance and see how the distance based on it, which we call the distance profile distance (DPD), avoids the computational difficulty of the GW distance. Then, we derive new results on the sample complexity of the DPD and discuss the limitations and potentials of the DPD as a matching method under certain probabilistic settings, which motivates the distance profile matching procedure in the next section.

\subsection{Background: Gromov-Wasserstein Distance and Computational Difficulty}
Following \cite{memoli2011gromov}, we build on the mathematical framework of object matching based on metric measure spaces. The Gromov-Wasserstein (GW) distance is a generalization of quadratic matching between point clouds to abstract metric measure spaces. The GW distance is defined based on the notion of couplings, which are probability measures on the product space of two metric measure spaces that preserve the marginal distributions of the two spaces. Couplings represent the correspondences between the points in the two metric measure spaces, and the GW distance finds the optimal coupling that best aligns the pairwise distances in the two spaces. Formal definitions are as follows.

\begin{definition}
    We call $(\cX, d_\cX, \mu)$ a metric measure space if $\mu$ is a Borel probability measure on a complete and separable metric space $(\cX, d_\cX)$.
\end{definition}

\begin{definition}
    Let $(\cX, d_\cX, \mu)$ and $(\cY, d_\cY, \nu)$ be two metric measure spaces. We call a Borel probability measure $\gamma$ on $\cX \times \cY$ a coupling if $\gamma(A \times \cY) = \mu(A)$ and $\gamma(\cX \times B) = \nu(B)$ for any Borel sets $A \subset \cX$ and $B \subset \cY$. We denote the set of all couplings by $\Pi(\mu, \nu)$.
\end{definition}

\begin{definition}
    The Gromov-Wasserstein distance of order $p \ge 1$ between two metric measure spaces $(\cX, d_\cX, \mu)$ and $(\cY, d_\cY, \nu)$ is defined as follows:
    \begin{equation}
        \label{eq:GW}
        \mathrm{GW}_p(\mu, \nu) := \left(\inf_{\gamma \in \Pi(\mu, \nu)} \int_{(\cX \times \cY)^2} |d_\cX(x, x') - d_\cY(y, y')|^p \, \mathrm{d} \gamma(x, y) \mathrm{d} \gamma(x', y')\right)^{1 / p}.
    \end{equation}
\end{definition}

The objective function on the right-hand side of \eqref{eq:GW} is quadratic in the variable $\gamma$ and is not convex in general, which makes the minimization computationally challenging. The following discrete case illustrates the computational difficulty of the GW distance in the case of point clouds, which is a common setting in practice.

\begin{example}[Discrete Case]
    \label{ex:GW_discrete}
    Consider the point clouds $\{X_1, \ldots, X_n\}$ and $\{Y_1, \ldots, Y_m\}$ on complete and separable metric spaces $(\cX, d_\cX)$ and $(\cY, d_\cY)$, respectively. We can define the GW distance between these point clouds as the GW distance between the metric measure spaces $(\cX, d_\cX, \hat{\mu})$ and $(\cY, d_\cY, \hat{\nu})$, where $\hat{\mu} = \frac{1}{n} \sum_{i = 1}^{n} \delta_{X_i}$ and $\hat{\nu} = \frac{1}{m} \sum_{j = 1}^{m} \delta_{Y_j}$ are the corresponding empirical measures. In this case, 
    \begin{equation}
        \label{eq:GW_discrete}
        \mathrm{GW}_p^p(\hat{\mu}, \hat{\nu}) = \min_{\gamma \in \Pi_{n, m}} \sum_{i, k = 1}^{n} \sum_{j, \ell = 1}^{m} |d_\cX(X_i, X_k) - d_\cY(Y_j, Y_\ell)|^p \gamma_{i j} \gamma_{k \ell},
    \end{equation}
    where $\Pi_{n, m} = \{\gamma \in \R_+^{n \times m} : \sum_{j = 1}^{m} \gamma_{i j} = \frac{1}{n} ~ \forall i \in [n] ~ \text{and} ~ \sum_{i = 1}^{n} \gamma_{i j} = \frac{1}{m} ~ \forall j \in [m]\} $ represents the set of all couplings between $\hat{\mu}$ and $\hat{\nu}$ as nonnegative matrices---called coupling matrices---of which the sum of each row (resp.\ column) is $\frac{1}{n}$ (resp.\ $\frac{1}{m}$). We again observe that the objective function of \eqref{eq:GW_discrete} is quadratic in the variable $\gamma$, which is NP-hard to solve in general.
\end{example}

\subsection{Distance Profile Distance}
Due to the aforementioned computational difficulty of the GW distance, various approximation and relaxation schemes have been proposed in the literature. \cite{memoli2011gromov} proposes several relaxations to the GW distance, which provide lower bounds to the GW distance. One of the three lower bounds proposed in \cite{memoli2011gromov}, which he named the third lower bound based on the ordering, is based on the idea of local distribution of distances, which we generalize as distance profiles below. 

\begin{definition}
    \label{def:distance_profile}
    Let $(\cX, d_\cX, \mu)$ be a metric measure space. For $x \in \cX$, let $\mu_x$ be the probability measure on $\R$ that represents the distribution of $d_\cX(x, X)$, where $X$ is a $\cX$-valued random variable whose law is $\mu$. We call $\mu_x$ the distance profile of $\mu$ at point $x$.
\end{definition}

In a nutshell, for each $x \in \cX$, the distance profile $\mu_x$ is the distribution of the distances from $x$ to the points in $\cX$ with respect to the measure $\mu$. This leads to a natural matching procedure as follows: given two metric measure spaces $(\cX, d_\cX, \mu)$ and $(\cY, d_\cY, \nu)$, we find a coupling that best aligns the distance profiles, namely, associate $x \in \cX$ and $y \in \cY$ if the distance profiles $\mu_x, \nu_y \in \mathscr{P}(\R)$ are similar. Here, similarity is measured by any suitable discrepancy on $\mathscr{P}(\R)$; we use the Wasserstein distance as a natural choice, which leads to the following definition.

\begin{definition}
    \label{def:distance_profile_distance}
    The distance profile distance (DPD) of order $p \ge 1$ between metric measure spaces $(\cX, d_\cX, \mu)$ and $(\cY, d_\cY, \nu)$ is defined as follows:
    \begin{equation}
        \label{eq:TLB}
        \mathfrak{T}_p(\mu, \nu) := \left(\inf_{\gamma \in \Pi(\mu, \nu)} \int_{\cX \times \cY} W_p^p(\mu_x, \nu_y) \, \mathrm{d} \gamma(x, y)\right)^{1 / p}.
    \end{equation}
\end{definition}

\begin{remark}
    \label{rmk:TLB}
    In Definition \ref{def:distance_profile_distance}, while we use $W_p^p$ to measure the discrepancy between the distance profiles $\mu_x, \nu_y \in \mathscr{P}(\R)$, any suitable discrepancy measure on $\mathscr{P}(\R)$ can be used in place of $W_p^p$. The choice $W_p^p$ leads to a nice metric property, namely, $\mathfrak{T}_p$ is a pseudometric on the set of metric measure spaces, and also naturally connects to the Gromov-Wasserstein distance. Indeed, one can show that $\mathfrak{T}_p \le \mathrm{GW}_p$. See \cite{memoli2011gromov} and \cite{memoli2022distance} for technical details.
\end{remark}

Unlike the GW distance, the objective function on the right-hand side of \eqref{eq:TLB} is linear in the variable $\gamma$, which makes the minimization computationally tractable. The following discrete case illustrates the computational advantage of the distance profile distance in the case of point clouds.

\begin{example}
    \label{ex:TLB_discrete}
    Consider the setting of Example \ref{ex:GW_discrete}. Again, we can define the distance profile distance between the point clouds $\{X_1, \ldots, X_n\}$ and $\{Y_1, \ldots, Y_m\}$ as the distance profile distance $\mathfrak{T}_p(\hat{\mu}, \hat{\nu})$. In this case, the distance profile of $\hat{\mu}$ at $x \in \cX$ is $\hat{\mu}_x = \frac{1}{n} \sum_{\ell = 1}^{n} \delta_{d_\cX(x, X_\ell)}$. Distance profiles of $\hat{\nu}$ are defined analogously. Then, one can see that the DPD is given as
    \begin{equation}
        \label{eq:TLB_discrete}
        \mathfrak{T}_p^p(\hat{\mu}, \hat{\nu}) = \min_{\gamma \in \Pi_{n, m}} \sum_{i = 1}^{n} \sum_{j = 1}^{m} W_p^p(\hat{\mu}_{X_i}, \hat{\nu}_{Y_j}) \gamma_{i j}.
    \end{equation}
    where $\Pi_{n, m}$ is as defined in Example \ref{ex:GW_discrete}. Computational advantage of $\mathfrak{T}_p$ over $\mathrm{GW}_p$ is clear: $\mathfrak{T}_p$ can be computed by solving a linear program, while $\mathrm{GW}_p$ requires solving a non-convex quadratic program. Unlike \eqref{eq:GW_discrete}, the minimization \eqref{eq:TLB_discrete} is a linear program that is solvable by polynomial-time algorithms.     
\end{example}

\subsection{Sample Complexity: Challenges and Opportunities}
\label{sec:connection_GW_TLB}
In practice, we have access to point clouds $\{X_1, \ldots, X_n\}$ and $\{Y_1, \ldots, Y_m\}$ that approximate some underlying objects represented as metric measure spaces $(\cX, d_\cX, \mu)$ and $(\cY, d_\cY, \nu)$, respectively. Hence, the quantity $\mathfrak{T}_p(\hat{\mu}, \hat{\nu})$ in Example \ref{ex:TLB_discrete} is an approximation of $\mathfrak{T}_p(\mu, \nu)$. An important statistical question is the rate of convergence of $\mathfrak{T}_p(\hat{\mu}, \hat{\nu})$ to $\mathfrak{T}_p(\mu, \nu)$ in terms of the sample sizes $n$ and $m$ when the point clouds are independent samples from $\mu$ and $\nu$. Faster convergence rates are desirable as they indicate that the distance profile distance is a good approximation of the population version $\mathfrak{T}_p(\mu, \nu)$ with a small sample size, which can further expedite the computation of the distance profile distance in practice.

It turns out, however, that the convergence rate of $\mathfrak{T}_p(\hat{\mu}, \hat{\nu})$ to $\mathfrak{T}_p(\mu, \nu)$ can be slow in the high-dimensional setting, just as the standard Wasserstein distance $W_p$ suffers from the curse of dimensionality when estimated from samples \citep{dudley1969speed,dereich2013constructive,boissard2014mean,fournier2015rate,weed2019sharp}. We derive this result for $p = 1$ as the leading case, focusing on $\E |\mathfrak{T}_1(\mu, \nu) - \mathfrak{T}_1(\hat{\mu}, \hat{\nu})|$. While the result is pessimistic in its full generality, we find that a better rate of convergence can be achieved when the underlying metric spaces have certain structures, controlled by the effective dimension of the metric measure space that can be much smaller than the ambient dimension.

\paragraph{Setting} Throughout this section, we mainly consider the setting of Example \ref{ex:TLB_discrete} where the point clouds are viewed as independent samples from the underlying metric measure spaces. It suffices to analyze the rate of convergence of $\E \mathfrak{T}_1(\mu, \hat{\mu})$ due to the triangle inequality: $\E |\mathfrak{T}_1(\mu, \nu) - \mathfrak{T}_1(\hat{\mu}, \hat{\nu})| \le \E \mathfrak{T}_1(\mu, \hat{\mu}) + \E \mathfrak{T}_1(\nu, \hat{\nu})$. Hence, we focus on $\E \mathfrak{T}_1(\mu, \hat{\mu})$.

\subsubsection{General Case: Curse of Dimensionality} 
We first observe that
\begin{equation}
    \label{eq:TLB_decomposition}
    \begin{split}
        \mathfrak{T}_1(\mu, \hat{\mu}) 
        & = \inf_{\gamma \in \Pi(\mu, \hat{\mu})} \int_{\cX \times \cX} W_1(\mu_x, \hat{\mu}_{x'}) \, \mathrm{d} \gamma(x, x') \\
        & \le \inf_{\gamma \in \Pi(\mu, \hat{\mu})} \int_{\cX \times \cX} \left(W_1(\mu_x, \mu_{x'}) + W_1(\mu_{x'}, \hat{\mu}_{x'})\right) \, \mathrm{d} \gamma(x, x') \\
        & \le \underbrace{\inf_{\gamma \in \Pi(\mu, \hat{\mu})} \int_{\cX \times \cX} W_1(\mu_x, \mu_{x'}) \, \mathrm{d} \gamma(x, x')}_{(*)} + \underbrace{\int_{\cX} W_1(\mu_{x'}, \hat{\mu}_{x'}) \, \mathrm{d} \hat{\mu}(x')}_{(**)}.
    \end{split}
\end{equation}
It turns out that the expectation of (**) enjoys a fast convergence rate. The idea is that
\begin{equation}
    \label{eq:TLB_decomposition_fast_part}
    \E(**) = \E\left(\frac{1}{n} \sum_{i = 1}^{n} W_1(\mu_{X_i}, \hat{\mu}_{X_i})\right) \le \E\left(\sup_{x \in \cX} W_1(\mu_x, \hat{\mu}_x)\right),
\end{equation}
where the last term can be shown to converge at the rate of $O(n^{-1/2})$ as stated below. 
\begin{theorem}
    \label{thm:W1_bound_sup_x_euclidean}
    Let $(\cX, d_\cX, \mu)$ be a metric measure space and $\hat{\mu} = \frac{1}{n} \sum_{i = 1}^{n} \delta_{X_i}$ be the empirical measure based on $X_1, \ldots, X_n$ that are i.i.d.\ from $\mu$. Suppose $\cX$ is a closed subset of $\R^d$ and $d_\cX$ is the standard Euclidean distance. Then, 
    \begin{equation}
        \label{eq:sup_W1_euclidean}
        \E\left[\sup_{x \in \cX} W_1(\mu_x, \hat{\mu}_x)\right] \le C \sqrt{\frac{d + 1}{n}},
    \end{equation}
    where $C$ is an absolute constant that is independent of $n$, $d$, and $\mu$.
\end{theorem}

\begin{remark}
    For each fixed $x \in \cX$, as $\mu_x$ is a probability measure on $\R$, one can deduce from the standard results \citep{dudley1969speed,fournier2015rate} that $W_1(\mu_x, \hat{\mu}_x)$ enjoys a fast rate of convergence in terms of the sample size $n$, which is $O(n^{-1 / 2})$. Theorem \ref{thm:W1_bound_sup_x_euclidean} shows that taking the supremum over $x \in \cX$ does not slow down the rate of convergence. We present the extension of Theorem \ref{thm:W1_bound_sup_x_euclidean} to general metric measure spaces in Appendix \ref{sec:W1_bound_sup_x_general}.
\end{remark}

While Theorem \ref{thm:W1_bound_sup_x_euclidean} shows fast convergence of (**), the term (*) in \eqref{eq:TLB_decomposition} is the main bottleneck for the rate of convergence of $\E \mathfrak{T}_1(\mu, \hat{\mu})$. Unfortunately, the rate of convergence of the expectation of (*) can be as slow as that of $\E W_1(\mu, \hat{\mu})$, which is known to be $O(n^{-1 / d})$ for $\cX \subset \R^d$. To see this, deduce (*) $\le W_1(\mu, \hat{\mu})$ from Proposition \ref{prop:analytic_properties}.

\begin{proposition}
    \label{prop:analytic_properties}
    Let $(\cX, d_\cX, \mu)$ be a metric measure space. Then, $W_1(\mu_x, \mu_{x'}) \le d_\cX(x, x')$ for $x, x' \in \cX$. 
\end{proposition}

In summary, despite the fast convergence rate of the term (**) in \eqref{eq:TLB_decomposition}, the term (*) in \eqref{eq:TLB_decomposition} can lead to a slow rate of convergence of $\E \mathfrak{T}_1(\mu, \hat{\mu})$ as $O(n^{-1 / d})$, thereby suffering from the curse of dimensionality.

\subsubsection{Faster Rate by Invariance and Effective Dimension}
It turns out that when $\mu$ has a certain structure, we can derive a faster rate of convergence for (*). Namely, if $\mu$ is a mixture of rotationally invariant distributions like a Gaussian mixture, then the effective dimension of $\mu$ is the number of the components, which can be much smaller than the ambient dimension $d$. The key is that $W_1(\mu_x, \mu_{x'})$, the distance between distance profiles, captures this structure and thus is determined by the effective dimension, rather than the ambient dimension.

\begin{proposition}
    \label{prop:W1_distance_profiles_mixture}
    Suppose $\mu$ is a probability measure on $\R^d$ given as a mixture of rotationally invariant distributions, namely, $\mu = \sum_{k = 1}^{t} p_k \mu_k$, where $(p_1, \ldots, p_t) \in \Delta_t$, and $\mu_1, \ldots, \mu_t \in \mathscr{P}(\R^d)$ are rotationally invariant with respect to the centers $\theta_1, \ldots, \theta_t \in \R^d$, respectively. Then, for any $x, x' \in \R^d$,
    \begin{equation}
        \label{eq:W1_distance_profiles_mixture}
        W_1(\mu_x, \mu_{x'}) \le \sum_{k = 1}^{t} p_k \left|\|x - \theta_k\|_2 - \|x' - \theta_k\|_2\right|.
    \end{equation}
\end{proposition}

Combining Proposition \ref{prop:W1_distance_profiles_mixture} with Theorem \ref{thm:W1_bound_sup_x_euclidean}, we can derive the following convergence rate of $\E \mathfrak{T}_1(\mu, \hat{\mu})$ depending on the effective dimension $t$ of $\mu$.

\begin{theorem}
    \label{thm:convergence_TLB_mixture}
    Suppose $\mu$ is a probability measure on $\R^d$ given as a mixture of rotationally invariant distributions, namely, $\mu = \sum_{k = 1}^{t} p_k \mu_k$, where $(p_1, \ldots, p_t) \in \Delta_t$, and $\mu_1, \ldots, \mu_t \in \mathscr{P}(\R^d)$ are rotationally invariant with respect to the centers $\theta_1, \ldots, \theta_t \in \R^d$, respectively. Let $\hat{\mu} = \frac{1}{n} \sum_{i = 1}^{n} \delta_{X_i}$ be the empirical measure based on $X_1, \ldots, X_n$ that are i.i.d.\ from $\mu$. Then, 
    \begin{equation*}
        \E \mathfrak{T}_1(\mu, \hat{\mu}) \le O(n^{-1 / t}) + C \sqrt{\frac{d + 1}{n}},
    \end{equation*}
    where $C$ is an absolute constant that is independent of $n$, $d$, and $\mu$.
\end{theorem}

\paragraph{Summary and Insights}
The above results provide insights into how to use distance profiles for matching. Due to the curse of dimensionality, it is impractical to estimate the distance profile distance $\mathfrak{T}_1$ between high-dimensional metric measure spaces by random samples. Hence, it is more suitable to use and study $\mathfrak{T}_1$ under a different setting than the sampling scenario; the next section conducts such a study based on the model with fixed high-dimensional locations corrupted by noise. Meanwhile, we have seen that a better rate can be achieved under a mixture model with rotationally invariant components. The high-level idea is that when such a structure is present, points from the same component are treated as the same point so that the effective dimension is essentially the number of components. This is pertinent to partial matching, where we want to allow similar points from one object to be matched to the same point in another object, not necessarily requiring a one-to-one correspondence or a coupling. In the next section, we study the distance profile matching procedure that allows for such flexibility, with a robustness property.

\section{Distance Profile Matching and Its Theoretical Guarantees}
\label{sec:main}
This section introduces and analyzes the distance profile matching algorithms, inspired by the insights on the distance profile distance $\mathfrak{T}_1$ in the previous section. We have learned that using $\mathfrak{T}_1$ directly for matching in the random sampling setting can be impractical. Hence, we propose a flexible matching procedure---Algorithm \ref{alg:distance_profile_matching}---modified from $\mathfrak{T}_1$ that allows for partial matching and derive its robustness property. Meanwhile, Algorithm \ref{alg:distance_profile_matching_assignment} is a special case of $\mathfrak{T}_1$ to find a one-to-one correspondence, which we analyze under a fixed-location model with noise, deviating from the random sampling setting where $\mathfrak{T}_1$ is impractical.

\subsection{Distance Profile Matching Algorithms}

\subsubsection{Distance Profile Matching Algorithm for Partial Matching}
\label{sec:distance_profile_matching}

\begin{algorithm}[!b]
    \caption{Distance Profile Matching}
    \label{alg:distance_profile_matching}
    \begin{algorithmic}[1]
    \REQUIRE Pairwise distances $\{d_\cX(X_i, X_\ell)\}_{i, \ell = 1}^{n}$ and $\{d_\cY(Y_j, Y_\ell)\}_{j, \ell = 1}^{m}$ from two point clouds $\{X_1, \ldots, X_n\}$ and $\{Y_1, \ldots, Y_m\}$ on metric spaces $(\cX, d_\cX)$ and $(\cY, d_\cY)$, respectively.
    \REQUIRE (Optional) A threshold $\rho > 0$. Set $\rho = \infty$ if not provided.
    \STATE For each pair $(i, j) \in [n] \times [m]$, define 
    \begin{equation}
        \label{eq:W(i, j)}
        W(i, j) := W_1\left(\frac{1}{n} \sum_{\ell = 1}^{n} \delta_{d_\cX(X_i, X_\ell)}, \frac{1}{m} \sum_{\ell = 1}^{m} \delta_{d_\cY(Y_j, Y_\ell)}\right).
    \end{equation}
    \STATE Define a map $\pi \colon [n] \to [m]$ by $\pi(i) = \argmin_{j \in [m]} W(i, j)$.
    \STATE Define $I \subset [n]$ by $I = \{i \in [n] : W(i, \pi(i)) < \rho\}$.
    \RETURN $\pi$ and $I$.
\end{algorithmic}
\end{algorithm}

Algorithm \ref{alg:distance_profile_matching} follows the idea of $\mathfrak{T}_1$ in that it compares the distance profiles of two point clouds, but it does not require finding an optimal coupling which was the main bottleneck of the sample complexity of $\mathfrak{T}_1$ in the random sampling setting. Instead, for each point $X_i$, we find $Y_j$ whose distance profile is closest to the distance profile of $X_i$ under $W_1$ as shown in Step 2 of Algorithm \ref{alg:distance_profile_matching}. Here, note that $W(i, j) = W_1(\hat{\mu}_{X_i}, \hat{\nu}_{Y_j})$, where $\hat{\mu}_{X_i}$ and $\hat{\nu}_{Y_j}$ are the distance profiles based on the empirical measures introduced in Example \ref{ex:TLB_discrete}. Therefore, the key difference from $\mathfrak{T}_1$ is that instead of finding a coupling $\gamma$ as in \eqref{eq:TLB_discrete}, we find $\pi$ that matches each $X_i$ to a point $Y_{\pi(i)}$ closes to $X_i$ in terms of the distance profiles. 

Note that $\pi$ may not be injective or surjective, meaning that some points may not be matched or some points may be matched to multiple points. The motivation behind this is simple: in partial matching, (1) we want some points to be ignored or matched arbitrarily, and (2) we may allow similar points to be matched to the same point. The underlying objects that the points $X_i$'s and $Y_j$'s approximate, say, probability measures $\mu$ and $\nu$ defined on $\R^d$, respectively, are not identical in practice. Usually, they coincide up to several outlier components, which we can formally state as $\mu = \lambda \rho + (1 - \lambda) \mu_o$ and $\nu = \lambda \rho + (1 - \lambda) \nu_o$ for some $\lambda \in (0, 1)$ and $\rho, \mu_o, \nu_o \in \mathscr{P}(\R^d)$. Here, $\lambda \rho$ represents the common part of $\mu$ and $\nu$ that we want to match, while $(1 - \lambda) \mu_o$ and $(1 - \lambda) \nu_o$ are the outlier components of $\mu$ and $\nu$, respectively. We are mainly interested in matching $X_i$'s and $Y_j$'s from the common part, while ignoring or arbitrarily matching the points from the outlier components.

The threshold $\rho$ is an optional parameter that can be used to obtain the subset $I \subset [n]$ of indices $i$'s such that the discrepancy $W(i, \pi(i))$ is below $\rho$. Small $\rho$ leads to a more conservative matching that only focuses on the points that are closely matched, which can be used for other downstream tasks. For instance, one can solve the orthogonal Procrustes problem \citep{wahba1965least,yang2019quaternion} using the matched pairs $\{(X_i, Y_{\pi(i)})\}_{i \in I}$ to find the optimal rotation and translation aligning them as closely as possible, which theoretically works with just $|I| = d$ pairs in $\R^d$, or $|I| = O(d)$ due to potential noise or degeneracy in practice. The obtained transformation can then be used as an initialization of the registration procedures, such as the Iterative Closest Point (ICP) algorithm \citep{besl1992method}, to find the optimal transformation that aligns the whole point clouds. We will utilize this idea in the experiments in Section \ref{sec:cryo-em}.

\subsubsection{Distance Profile Matching by Assignment}
Algorithm \ref{alg:distance_profile_matching_assignment} is a special case of $\mathfrak{T}_1$ that finds a one-to-one correspondence, namely, assignment, when the two point clouds have the same number of points. In Step 2 of Algorithm \ref{alg:distance_profile_matching_assignment}, we find a permutation $\pi \in \cS_n$ that minimizes the total discrepancy between the distance profiles of $X_i$ and $Y_{\pi(i)}$'s, which is formulated as the standard linear assignment problem. By the standard linear programming theory, one can deduce that the minimum of this linear assignment problem is equivalent to the minimum of the linear program over the Birkhoff polytope, which is $\mathfrak{T}_1$ defined in \eqref{eq:TLB_discrete} when $n = m$. Hence, Algorithm \ref{alg:distance_profile_matching_assignment} is essentially $\mathfrak{T}_1$, which, as discussed, can be impractical in the random sampling setting. Instead, it can be useful in the fixed-location model with noise, which we will analyze in Section \ref{sec:noise_stability}.

\begin{algorithm}[!h]
    \caption{Distance Profile Matching by Assignment}
    \label{alg:distance_profile_matching_assignment}
    \begin{algorithmic}[1]
    \REQUIRE Pairwise distances $\{d_\cX(X_i, X_\ell)\}_{i, \ell = 1}^{n}$ and $\{d_\cY(Y_j, Y_\ell)\}_{j, \ell = 1}^{n}$ from two point clouds $\{X_1, \ldots, X_n\}$ and $\{Y_1, \ldots, Y_n\}$ on metric spaces $(\cX, d_\cX)$ and $(\cY, d_\cY)$, respectively.
    \STATE  For each pair $(i, j) \in [n] \times [n]$, define $W(i, j)$ as in \eqref{eq:W(i, j)}.
    \STATE Solve $\hat{\pi} = \argmin_{\pi \in \cS_n} \sum_{i = 1}^{n} W(i, \pi(i))$.
    \RETURN $\hat{\pi}$.
\end{algorithmic}
\end{algorithm}

\subsubsection{Computational Complexity}
The Wasserstein-1 distance between two sets of points in $\R$ admits a closed-form expression based on the corresponding cumulative distribution functions, which can be computed efficiently \citep{peyre2019computational}. In this case, the computational complexity of Algorithms \ref{alg:distance_profile_matching} and \ref{alg:distance_profile_matching_assignment} is computed as follows. For Step 1, assuming $n \ge m$, computing the Wasserstein-1 distances between two sets of points in $\R$ for $n m$ pairs takes $O(n^3 \log(n))$ in total; when $n = m$ as in Algorithm \ref{alg:distance_profile_matching_assignment}, the exact computation is possible solely based on sorting. Meanwhile, Step 2 takes $O(n m \log(m))$ as it requires sorting $m$ numbers $n$ times; for Algorithm \ref{alg:distance_profile_matching_assignment}, it takes $O(n^3)$ to solve the linear assignment problem. Hence, the overall complexity of Algorithm \ref{alg:distance_profile_matching} or \ref{alg:distance_profile_matching_assignment} is $O(n^3 \log(n))$ assuming $n \ge m$, which is nearly cubic in the number of points $n$; however, as the input dimension of the algorithm is $n^2$, we can say that the complexity scales nearly as the input dimension $n^2$ raised to the power of $\frac{3}{2}$. 

\begin{remark}
    While we mainly use the Wasserstein-1 distance $W_1$ to compare the distance profiles, one may use any well-defined discrepancy measure in the place of $W_1$ in \eqref{eq:W(i, j)}. The choice $W_1$ is natural in connection with the previous section (Remark \ref{rmk:TLB}), which is theoretically appealing and practically useful. Lastly, we comment that Algorithms \ref{alg:distance_profile_matching} and \ref{alg:distance_profile_matching_assignment} are implementable without knowing the points $X_i$'s and $Y_j$'s as long as the pairwise distances $\{d_\cX(X_i, X_\ell)\}_{i, \ell = 1}^{n}$ and $\{d_\cY(Y_j, Y_\ell)\}_{j, \ell = 1}^{m}$ are provided. This is because the distance profiles are solely determined by the pairwise distances.
\end{remark}

\subsection{Partial Matching Analysis: Robustness to Outlier Components}
\label{sec:outlier_robustness}
This section analyzes the robustness of Algorithm \ref{alg:distance_profile_matching} in the presence of outlier components. Algorithm \ref{alg:distance_profile_matching} is designed to match the points corresponding to the same region of the objects, while ignoring or arbitrarily matching the points from the outlier components. In order to theoretically analyze robustness to such components, we consider appropriate model assumptions on the underlying objects $\mu, \nu \in \mathscr{P}(\R^d)$ that capture the nature of the outlier components that we want to ignore in the partial matching setting. Motivated by several practical considerations, we model $\mu, \nu$ as follows.

\begin{assumption}
    \label{a:mixture}
    Suppose that $\mu$ and $\nu$ are mixture distributions consisting of $t$ and $s$ components, respectively, defined as follows: for some $(p_1, \ldots, p_t) \in \Delta_t$ and $(q_1, \ldots, q_s) \in \Delta_s$, 
    \begin{equation*}
        \mu = \sum_{k = 1}^{t} p_k \mu_k \quad \text{and} \quad \nu = \sum_{k = 1}^{s} q_k \nu_k,
    \end{equation*}
    where $\mu_k, \nu_k \in \mathscr{P}(\R^d)$ are sub-Gaussian\footnote{See Lemma \ref{lem:sub_Gaussian} for the precise definition.} with variance proxies $\sigma_k^2, \tau_k^2$ and means $\theta_k, \eta_k \in \R^d$, respectively. For some $K \le \min\{t, s\}$, we assume that the first $K$ components of $\mu$ and $\nu$ are identical, that is, $p_k = q_k$ and $\mu_k = \nu_k$ for $k \in [K]$.
\end{assumption}

Assumption \ref{a:mixture} is equivalent to letting $\mu = \lambda \rho + (1 - \lambda) \mu_o$ and $\nu = \lambda \rho + (1 - \lambda) \nu_o$, where the common part $\lambda \rho$ is the mixture consisting of the first $K$ components and the other parts, namely, the outlier parts $(1 - \lambda) \mu_o, (1 - \lambda) \nu_o$ are the mixtures of the remaining components. The motivation behind Assumption \ref{a:mixture} is that the objects $\mu, \nu$ are composed of several regions, where each region is represented as a sub-Gaussian distribution so that the mean corresponds to the center of the region and the variance proxy roughly represents the size of the region. The points corresponding to the same region are to be matched, while we ignore matching the points from the outlier components. 

\begin{remark}[Scope of Assumption \ref{a:mixture}]
    \label{rmk:mixture}
    Gaussian mixture models are the most common instance of Assumption \ref{a:mixture}, where $\mu_k$'s and $\nu_k$'s are Gaussian. By allowing these components to be sub-Gaussian, Assumption \ref{a:mixture} encompasses a wider class of distributions. For instance, stochastic ball models \citep{nellore2015recovery,iguchi2017probably}---widely used to evaluate clustering algorithms---are special cases of Assumption \ref{a:mixture} where $\mu_k$'s and $\nu_k$'s are supported on balls. Also, Assumption \ref{a:mixture} allows the components to be low-dimensional, which can model 3D objects given as smooth densities on the surface of the objects.
\end{remark}

Now, we apply Algorithm \ref{alg:distance_profile_matching} to the point clouds $\{X_1, \ldots, X_n\}$ and $\{Y_1, \ldots, Y_m\}$ that approximate the objects $\mu$ and $\nu$, respectively. We take a viewpoint that the points $X_1, \ldots, X_n$ and $Y_1, \ldots, Y_m$ are independent samples from $\mu$ and $\nu$, respectively. For $i \in [n]$ and $j \in [m]$, let $t(i) \in [t]$ and $s(j) \in [s]$ denote the components where $X_i$ and $Y_j$ are from, respectively. Our goal is to match $X_i$ and $Y_j$ if $t(i) = s(j) \in [K]$, namely, they are from the same region. Under this probabilistic setting, we show that Algorithm \ref{alg:distance_profile_matching} outputs the matching $\pi$ that has the following guarantee: for $X_i$ such that $t(i) \in [K]$, we have $s(\pi(i)) = t(i)$, namely, its match $Y_{\pi(i)}$ is from the corresponding component. More precisely, we show that the event
\begin{equation}
    \label{eq:matching_success_event}
    \bigcap_{i \in [n] : t(i) \in [K]} (s(\pi(i)) = t(i)),
\end{equation}
which represents the successful matching of the points from the common components, holds with high probability. 

\begin{theorem}
    \label{thm:outlier_robustness}
    Suppose $\mu$ and $\nu$ satisfy Assumption \ref{a:mixture}. Let $X_1, \ldots, X_n$ and $Y_1, \ldots, Y_m$ be independent samples from $\mu$ and $\nu$, respectively. For $i \in [n]$ and $j \in [m]$, let $t(i) \in [t]$ and $s(j) \in [s]$ denote the components where $X_i$ and $Y_j$ are from, respectively. Let $\pi \colon [n] \to [m]$ be the output of Algorithm \ref{alg:distance_profile_matching}. Then, for any $\delta \in (0, 1)$, we have
    \begin{equation}
        \label{eq:matching_guarantee}
        \P\left(\bigcap_{i \in [n] : t(i) \in [K]} (s(\pi(i)) = t(i))\right) \ge 1 - \delta
    \end{equation}
    if
    \begin{equation}
        \label{eq:separation_condition}
        \begin{split}
            & \min_{\substack{\alpha \in [K] \\ \beta \in [s] \backslash \{\alpha\}}} \overline{W}(\alpha, \beta) - R \cdot \bigg(1 - \sum_{k \in [K]} p_k \bigg) \\
            & \ge \max\bigg(8 R \sqrt{\frac{\log K +\log (t + s - 2) + \log \tfrac{4}{\delta}}{2 \min(n, m)}}, \, 32 \Gamma \sqrt{2 d + \log \frac{4 (\max(n, m))^2}{\delta}}\bigg),
        \end{split}
    \end{equation}
    where $R = \max\{\|x - y\|_2 : x, y \in \{\theta_1, \ldots, \theta_t, \eta_1, \ldots, \eta_s\}\}$, $\Gamma = \max\{\sigma_1, \ldots, \sigma_t, \tau_1, \ldots, \tau_s\}$, and
    \begin{equation}
        \label{eq:W1_distance_profiles_centers}
        \overline{W}(\alpha, \beta) := W_1\left(\sum_{k = 1}^{t} p_k \delta_{\|\theta_{\alpha} - \theta_k\|_2}, \sum_{k = 1}^{s} q_k \delta_{\|\eta_{\beta} - \eta_k\|_2}\right) \quad \forall (\alpha, \beta) \in [t] \times [s].
    \end{equation}
\end{theorem}

The main takeaway of Theorem \ref{thm:outlier_robustness} is that we have the desired guarantee \eqref{eq:matching_guarantee} if
\begin{equation*}
    \begin{split}
        & \text{\textit{separation among the components compared to the proportion of outlier components}} \\
        & \ge \text{\textit{some threshold depending on the noise and sampling error}},
    \end{split}
\end{equation*}
which is exactly what \eqref{eq:separation_condition} states. Here, the left-hand side of \eqref{eq:separation_condition} captures the separation in terms of \eqref{eq:W1_distance_profiles_centers} in comparison to the proportion $1 - \sum_{k \in [K]} p_k$ of the outlier components, while the right-hand side is a threshold 
determined by the noise and sampling error. 

It is worth nothing that the right-hand side of \eqref{eq:separation_condition} contains the term $\Gamma \sqrt{d}$, which is sensible as a bound on the typical size of the noise displacement from the centers of the components. This is a minimal condition given that both sides of \eqref{eq:separation_condition} scale with typical Euclidean distances in $\R^d$. See Figure \ref{fig:centers}. Each sub-Gaussian component corresponds---in the worst case sense---to a ball of radius $\Gamma \sqrt{d}$, which has to be well-separated from the other balls to ensure that it is a distinct component. 

\begin{figure}[!htb]
    \centering
    \includegraphics[trim=0.7cm 0.7cm 0.7cm 0.7cm, clip=true, width=0.375\textwidth]{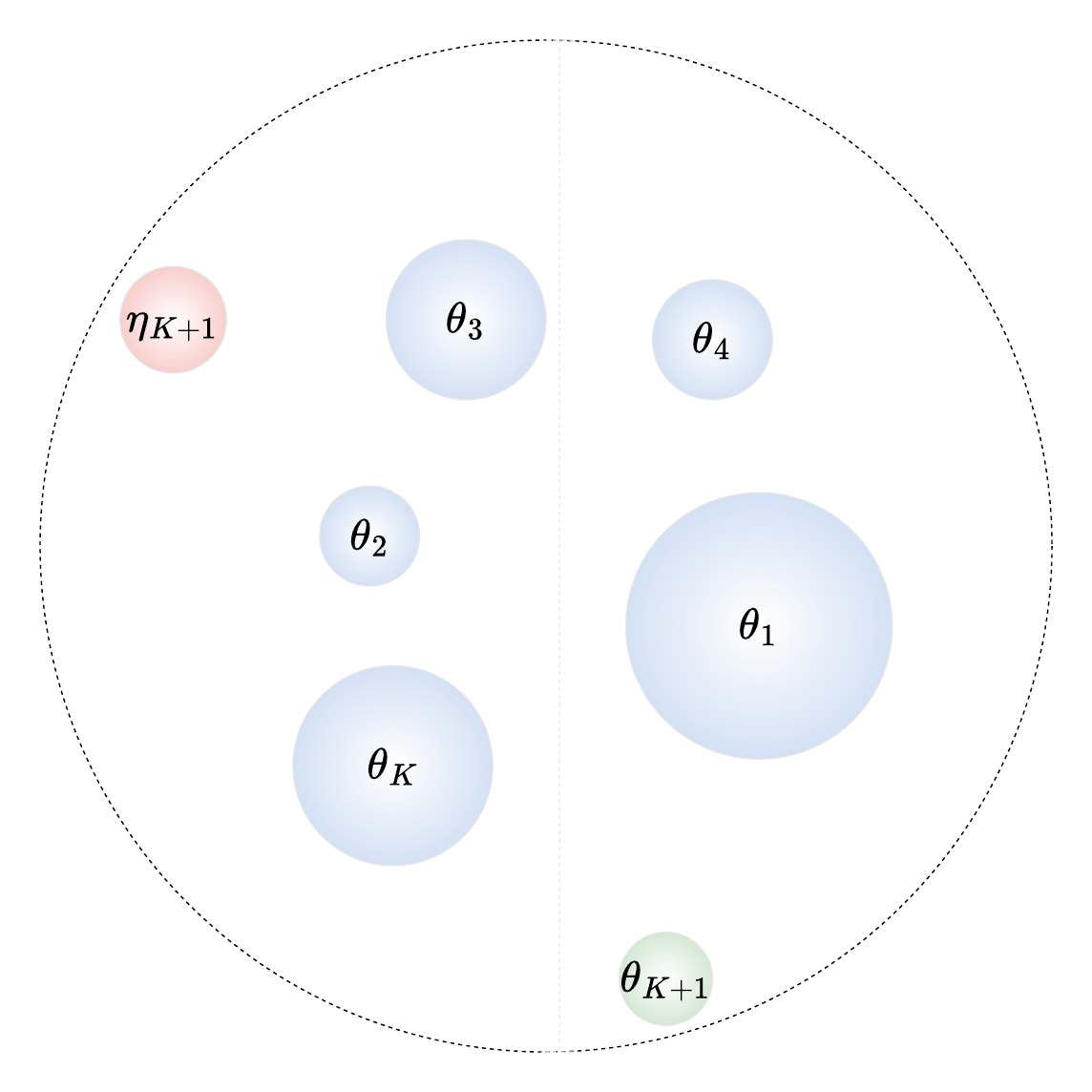}
    \caption{Illustration of the sub-Gaussian mixtures. For simplicity, each component is represented as a ball centered at the corresponding center, where the radius of the ball is roughly $\Gamma \sqrt{d}$. All of them are contained in a larger ball of diameter $R$.}
    \label{fig:centers}
\end{figure}

\subsection{One-to-One Matching Analysis: Noise Stability}
\label{sec:noise_stability}
We now analyze the noise stability of Algorithm \ref{alg:distance_profile_matching_assignment}. Again, remember that Algorithm \ref{alg:distance_profile_matching_assignment} is a special case of $\mathfrak{T}_1$ that finds a one-to-one correspondence, namely, assignment, which is not recommended in the random sampling setting of the previous section. Instead, this algorithm can be useful in the fixed-location model with noise; as we analyze below, we find that the matching procedure based on distance profiles is robust to the additive noise.

\paragraph{Fixed-Location Model with Noise}
We consider two sets of locations $\{\theta_1, \ldots, \theta_n\}$ and $\{\eta_1, \ldots, \eta_n\}$ in $\R^d$ and assume $\eta_{\pi^\ast(i)} = T(\theta_i)$ for $i \in [n]$, where $T \colon \R^d \to \R^d$ is an unknown rigid transformation and $\pi^\ast$ is an unknown permutation on $[n]$ representing the correspondence, say, $\theta_i$ and $\eta_{\pi^\ast(i)}$ are to be matched. Unlike the previous model in Section \ref{sec:outlier_robustness}, we assume there is a ground truth one-to-one correspondence $\pi^\ast$ between two sets of locations, which we want to recover from the following perturbed locations $X_i = \theta_i + \xi_i$ and $Y_i = \eta_i + \zeta_i$, where $\xi_i$ and $\zeta_i$ are independent mean-zero random vectors following sub-Gaussian distributions with variance proxies $\sigma_i^2$ and $\tau_i^2$, respectively.\footnote{Technically, one may think of this model as the previous model in Section \ref{sec:outlier_robustness} with $t = s = K$, namely, no outlier components. However, the key difference is that we no longer sample randomly from $\mu$ and $\nu$, which is not the setting we want to use Algorithm \ref{alg:distance_profile_matching_assignment} for as we previously discussed.} For instance, for $d = 3$, we can think of $X_i$'s and $Y_i$'s as the coordinates obtained by scanning some objects of interest but from different angles, say, frontal and lateral views, where the noise is added to the ground truth coordinates $\theta_i$'s and $\eta_i$'s during the scanning process. 

The goal is to recover the true permutation $\pi^\ast$ from the perturbed locations $X_i$'s and $Y_i$'s. We show that the output $\hat{\pi}$ of Algorithm \ref{alg:distance_profile_matching_assignment} recovers $\pi^\ast$ with high probability. As discussed in Section \ref{sec:outlier_robustness}, several conditions determine the success of the matching. Again, the configuration of the locations $\theta_1, \ldots, \theta_n, \eta_1, \ldots, \eta_n$ is crucial for the matching to be well-defined. To this end, we suppose the following configuration condition:
\begin{equation}
    \label{eq:separation_locations}
    \Phi := \min_{i \neq j} W_1\left(\frac{1}{n} \sum_{k = 1}^{n} \delta_{\|\theta_i - \theta_k\|_2}, \frac{1}{n} \sum_{k = 1}^{n} \delta_{\|\theta_j - \theta_k\|_2}\right) > 0.
\end{equation}
The condition \eqref{eq:separation_locations} ensures that $\pi^\ast$ is the unique optimal matching in the following sense:
\begin{equation*}
    \pi^\ast = \argmin_{\pi \in \cS_n} \sum_{i = 1}^{n} W_1\left(\frac{1}{n} \sum_{k = 1}^{n} \delta_{\|\theta_i - \theta_k\|_2}, \frac{1}{n} \sum_{k = 1}^{n} \delta_{\|\eta_{\pi(i)} - \eta_k\|_2}\right) =: \argmin_{\pi \in \cS_n} \cQ(\pi).
\end{equation*}
To see this, note that $\cQ(\pi^\ast) = 0$, while $\cQ(\pi) > 0$ for $\pi \neq \pi^\ast$ as $\Phi > 0$. The following theorem shows that the output $\hat{\pi}$ of Algorithm \ref{alg:distance_profile_matching_assignment} recovers $\pi^\ast$ with high probability given this condition \eqref{eq:separation_locations} and the variances $\sigma_i$'s and $\tau_i$'s are sufficiently small.

\begin{theorem}
    \label{thm:noise_stability}
    Suppose that the locations $\theta_1, \ldots, \theta_n, \eta_1, \ldots, \eta_n \in \R^d$ satisfy
    \begin{equation}
        \label{eq:locations_eta}
        \eta_{\pi^\ast(i)} = T(\theta_i) \quad \forall i \in [n],
    \end{equation}
    for some unknown rigid transformation $T \colon \R^d \to \R^d$ and some unknown permutation $\pi^\ast$ on $[n]$. Now, suppose we have the following noisy locations:
    \begin{equation}
        \label{eq:noise_stability_X_Y}
        X_i = \theta_i + \xi_i \quad \text{and} \quad Y_i = \eta_i + \zeta_i,
    \end{equation}
    where $\xi_i$ and $\zeta_i$ are independent mean-zero random vectors following sub-Gaussian distributions with variance proxies $\sigma_i^2$ and $\tau_i^2$, respectively. Let $\hat{\pi}$ be the output of Algorithm \ref{alg:distance_profile_matching_assignment}. Then, for any $\delta \in (0, 1)$, we have
    \begin{equation*}
        \P(\hat{\pi} = \pi^\ast) \ge 1 - \delta
    \end{equation*}
    if for $\Phi$ defined in \eqref{eq:separation_locations}, the following condition holds:
    \begin{equation}
        \label{eq:noise_level_condition}        
        \max\{\sigma_1, \ldots, \sigma_n, \tau_1, \ldots, \tau_n\} \le \frac{\Phi}{16 \sqrt{2 d + \log(2 n^2 / \delta)}}.
    \end{equation}
\end{theorem}

For the case where there is no rigid transformation, namely, $T$ is the identity map, the noise stability guarantee in the above model is well-studied in the literature, known as the feature matching problem. In this case, the simplest approach is to solve the following discrete optimal transport problem \citep{peyre2019computational}:, 
\begin{equation}
    \label{eq:feature_matching_linear_assignment}
    \hat{\pi}^{\mathrm{OT}} = \argmin_{\pi \in \cS_n} \sum_{i = 1}^{n} \|X_i - Y_{\pi(i)}\|_2^2.
\end{equation} 
It is well-studied \citep{collier2016minimax,galstyan2022optimal} that such a procedure is indeed guaranteed to recover $\pi^\ast$ with high probability if the separation among $\theta_i$'s is large enough compared to the noise level. However, in the presence of an unknown rigid transformation $T$, such a procedure is not guaranteed to succeed, while $\hat{\pi}$ in Theorem \ref{thm:noise_stability} is invariant to rigid transformations. We demonstrate this point in Section \ref{sec:simulations}.

\section{Simulations}
\label{sec:simulations}

\subsection{Synthetic Data}

\subsubsection{Gaussian Mixture Matching}
We first provide a simulation study to demonstrate the performance of the matching procedure discussed in Section \ref{sec:outlier_robustness}. Figure \ref{fig:gmm_matching}(a) shows the probability of the perfect matching, namely, the left-hand side of \eqref{eq:matching_guarantee}. Here, $d = 3$ and $K = 10$, and we consider the clean case $t = s = K$ and the case with one outlier component $t = s = K + 1$ for the sample size $n = m \in \{500, 1000\}$; also, we let $p_k = q_k = \frac{1}{t}$ and consider Gaussian mixtures, namely, $\mu_k = N(\theta_k, \sigma^2 I_d)$ and $\nu_k = N(\eta_k, \sigma^2 I_d)$, where we vary $\sigma$. We can observe that the probability of the perfect matching decreases as $\sigma$ grows for both cases. Though the presence of an outlier component deteriorates the matching performance as expected, we can verify that Algorithm \ref{alg:distance_profile_matching} is robust to the outlier components in that the recovery probability is sufficiently high when $\sigma$ is small. Comparing the plots for $n = m = 500$ and $n = m = 1000$, we can see that the matching performance is better for the larger sample size as it better approximates the population-level distance profiles.

\begin{figure}[!ht]
    \centering
    \subfloat[Perfect matching probability]{\includegraphics[trim=0.1cm 0.45cm 0.1cm 0.35cm, clip=true, width=0.42\textwidth]{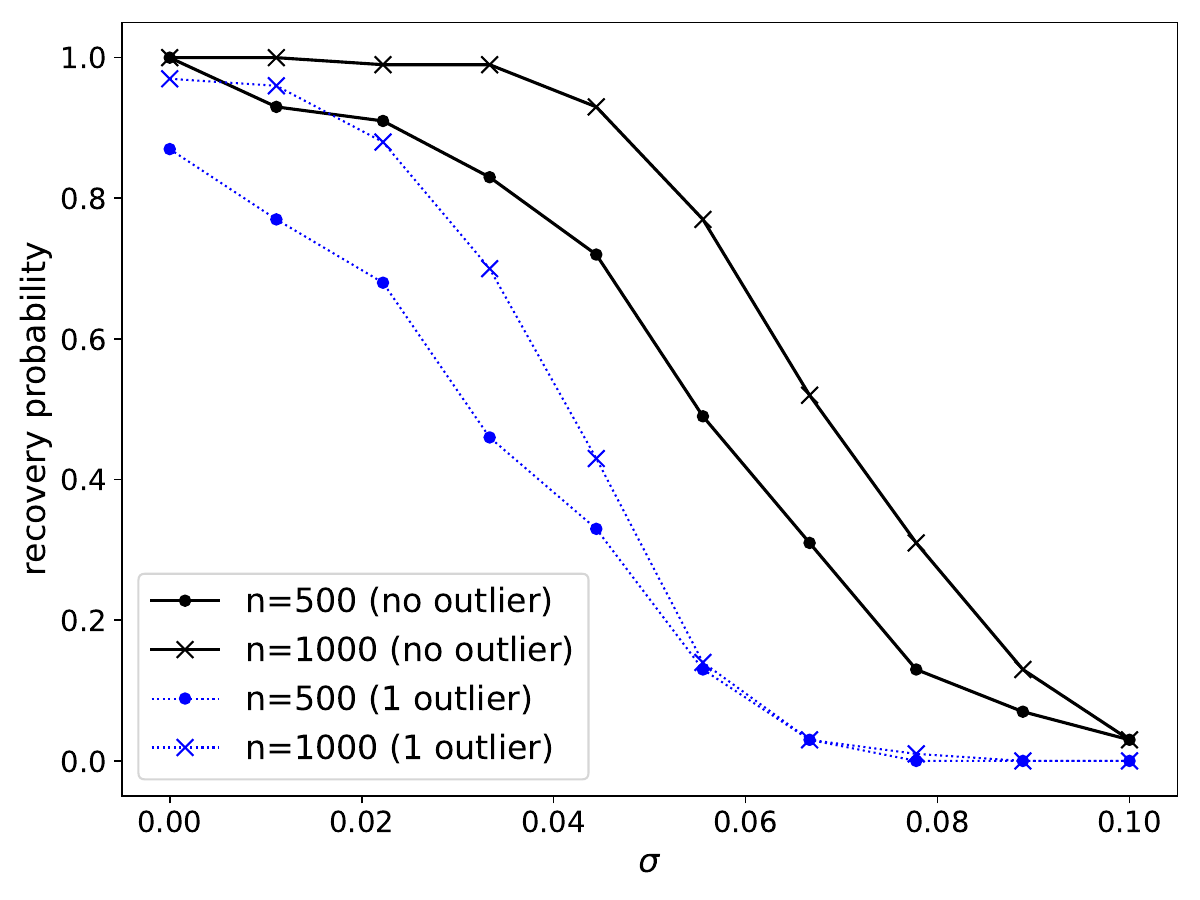}}
    \subfloat[Matching accuracy ($K = 10$)]{\includegraphics[trim=0.1cm 0.45cm 0.1cm 0.35cm, clip=true, width=0.42\textwidth]{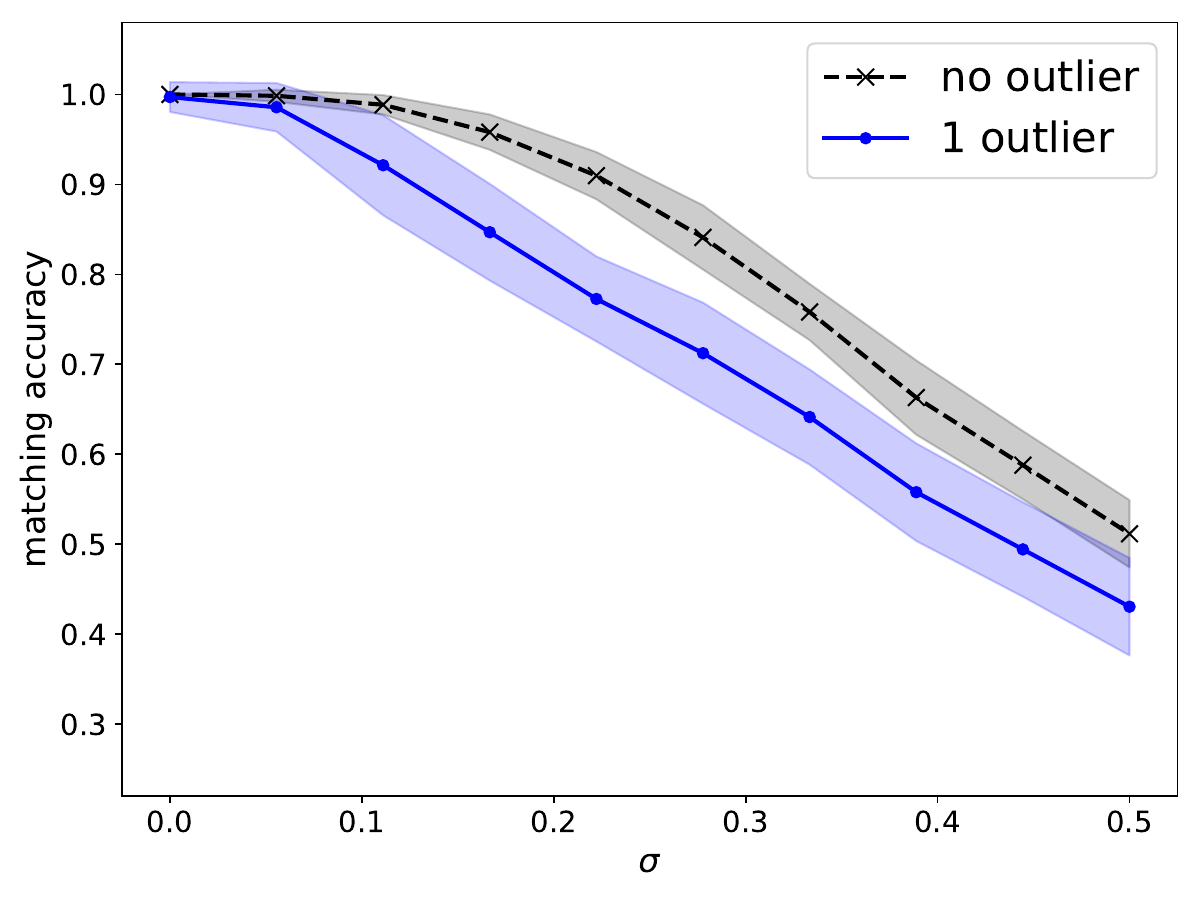}} \\
    \subfloat[Matching accuracy (varying $K$)]{\includegraphics[trim=0.1cm 0.45cm 0.1cm 0.35cm, clip=true, width=0.42\textwidth]{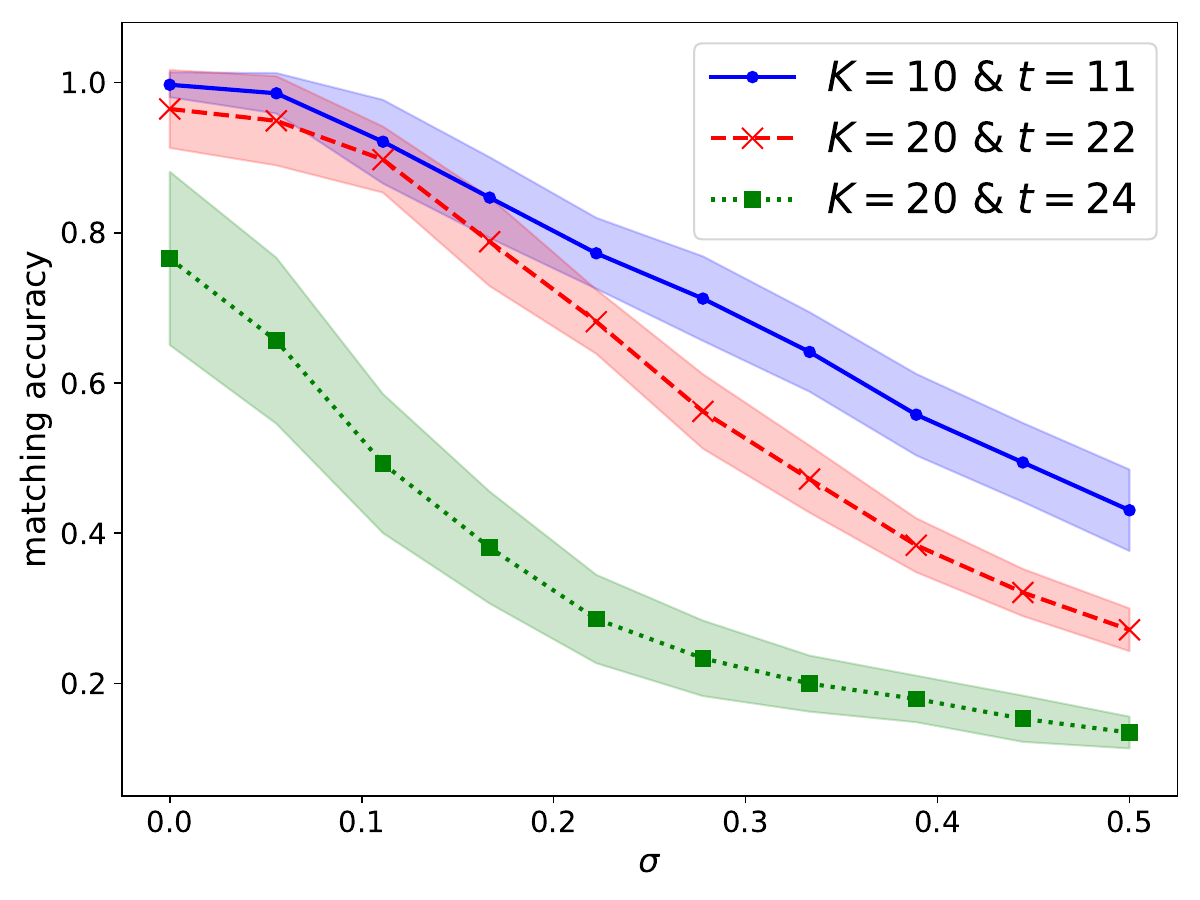}} 
    \subfloat[Matching accuracy (varying $d$)]{\includegraphics[trim=0.1cm 0.45cm 0.1cm 0.35cm, clip=true, width=0.42\textwidth]{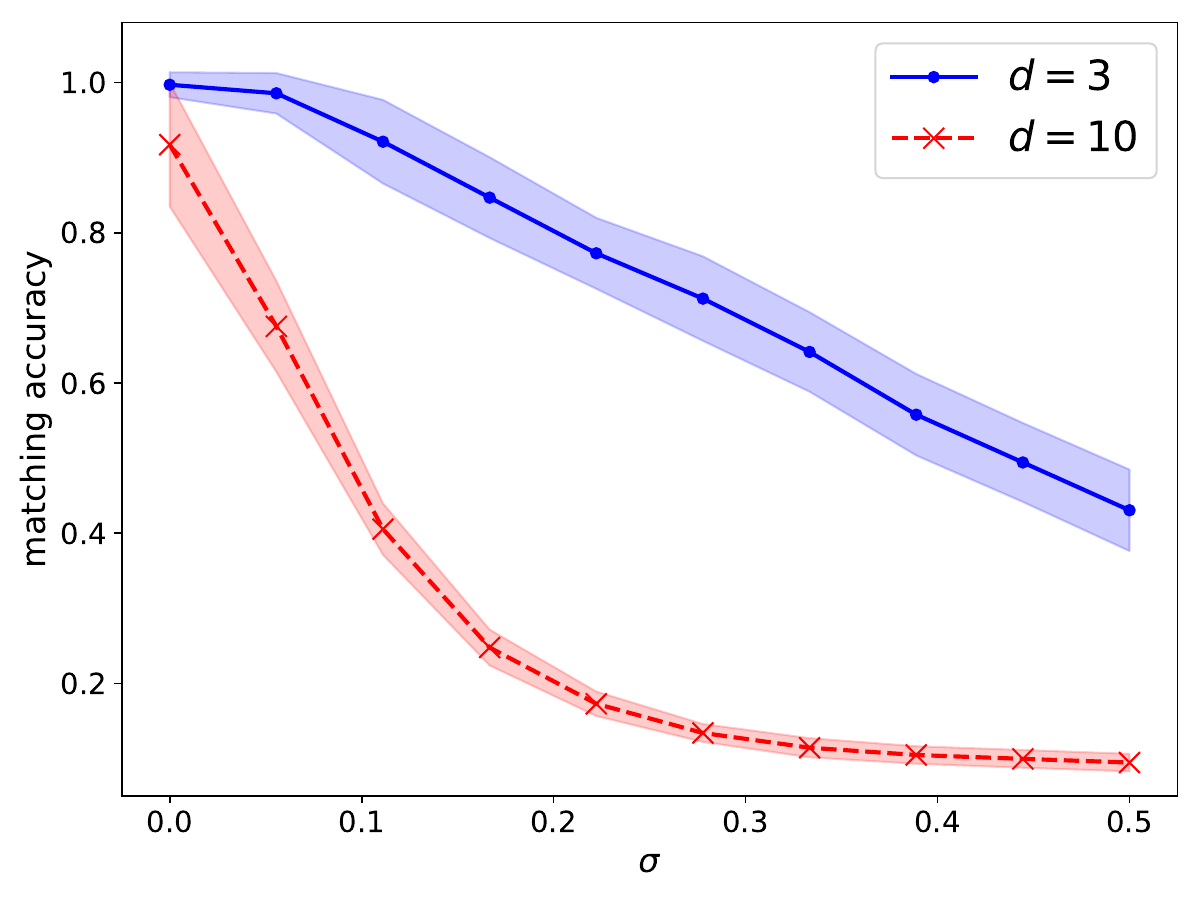}} 
    \caption{Gaussian mixture matching results with $p_k = q_k = \frac{1}{t}$, $\mu_k = N(\theta_k, \sigma^2 I_d)$, $\nu_k = N(\eta_k, \sigma^2 I_d)$, and varying $\sigma$. (a) shows the probability---the left-hand side of \eqref{eq:matching_guarantee}---with $d = 3$, $K = 10$, and the sample size $n = m \in \{500, 1000\}$. (b) shows the matching accuracy \eqref{eq:mixture_matching_accuracy} for the same setting as (a), while (c), (d) show the matching accuracy for different settings, varying $K, d$, respectively, with $n = m = 1000$ for (b)-(d). For all plots, the sampling from $\mu, \nu$ is repeated 100 times. For (a), the markers show the probability estimated from the 100 realizations; for (b)-(d), they show the average of the matching accuracy over 100 realizations, with the standard deviations shown as the shaded regions. The blue solid lines of (b)-(d) are the same, showing the case of $K = 10$ and $t = s = 11$.}
    \label{fig:gmm_matching}
\end{figure}

Although the analysis in Section \ref{sec:outlier_robustness} was based on the probability of the perfect matching, it is reasonable to expect matching errors in practice. Figure \ref{fig:gmm_matching}(b) shows the matching accuracy for the same setting as Figure \ref{fig:gmm_matching}(a) with $n = m = 1000$, where the matching accuracy is the proportion of correctly matched $X_i$'s from the first $K$ components:
\begin{equation}
    \label{eq:mixture_matching_accuracy}
    \frac{|\{i \in [n] : t(i) \in [K] ~~ \text{and} ~~ s(\pi(i)) = t(i)\}|}{|\{i \in [n] : t(i) \in [K]\}|},
\end{equation}
where $t, s$ are defined as in Theorem \ref{thm:outlier_robustness}. Again, we can see that the proposed method performs reasonably well. Figure \ref{fig:gmm_matching}(c) compares the case with one outlier component ($K = 10$ and $t = s = 11$) to the case with more outlier components but with the same proportion ($K = 20$ and $t = s = 22$) and to the case with more outlier components with larger proportions ($K = 20$ and $t = s = 24$). We can observe that the matching accuracy decreases as the proportion of outlier components increases, which is consistent with the theoretical results in Section \ref{sec:outlier_robustness}. Even when the proportion $(t - K) / K$ is the same, the matching accuracy is lower for larger $t$ as the larger effective dimensions of the underlying distributions $\mu, \nu$ lead to larger estimation errors given the same sample size. Figure \ref{fig:gmm_matching}(d) compares the case with $d = 3$ to the case with $d = 10$ for $K = 10$ and $t = s = 11$. We can see that the matching accuracy is lower and deteriorates much faster for larger $d$.

\subsubsection{Noise Stability}
Next, we provide simulation results demonstrating the noise stability discussed in Section \ref{sec:noise_stability}. We first obtain $n$ locations $\theta_1, \ldots, \theta_n \in \R^d$ and define the other locations $\eta_1, \ldots, \eta_n$ as in \eqref{eq:locations_eta} using some rigid transformation $T$ and a permutation $\pi^\ast$. Then, we generate the points according to \eqref{eq:noise_stability_X_Y} and obtain $\hat{\pi}$ by applying Algorithm \ref{alg:distance_profile_matching_assignment}; we estimate $\P(\hat{\pi} = \pi^\ast)$ by repeating this for 100 times. Figure \ref{fig:noise_stability} shows the results, where the noise variables $\xi_i$'s and $\zeta_i$'s follow $N(0, \sigma^2 I_d)$, and we vary $\sigma$ and compare with the probability $\P(\hat{\pi}^{\mathrm{OT}} = \pi^\ast)$ for $\hat{\pi}^{\mathrm{OT}}$ obtained by the procedure \eqref{eq:feature_matching_linear_assignment}. As shown in Figure \ref{fig:noise_stability}(a), when the rigid transformation is based on a rotation matrix that only affects the first two coordinates, both methods can recover the true permutation to some extent for small enough $\sigma$; in this case, as the rotation matrix is relatively close to the identity matrix, the standard optimal transport method \eqref{eq:feature_matching_linear_assignment} shows good performance. For a rotation matrix that changes all the coordinates, however, the standard optimal transport method fails to recover the true permutation as shown in Figure \ref{fig:noise_stability}(b), while the distance profile matching method performs almost as in Figure \ref{fig:noise_stability}(a).

Lastly, we briefly comment that the plots of Figure \ref{fig:gmm_matching}(a) and Figure \ref{fig:noise_stability} suggest the existence of a phase transition in the noise level, where the performance of the matching methods might change dramatically under certain settings. We leave this for future work.

\begin{figure}[!htb]
    \centering
    \subfloat[Changing two coordinates]{\includegraphics[trim=0.1cm 0.2cm 0.1cm 0.2cm, clip=true, width=0.44\textwidth]{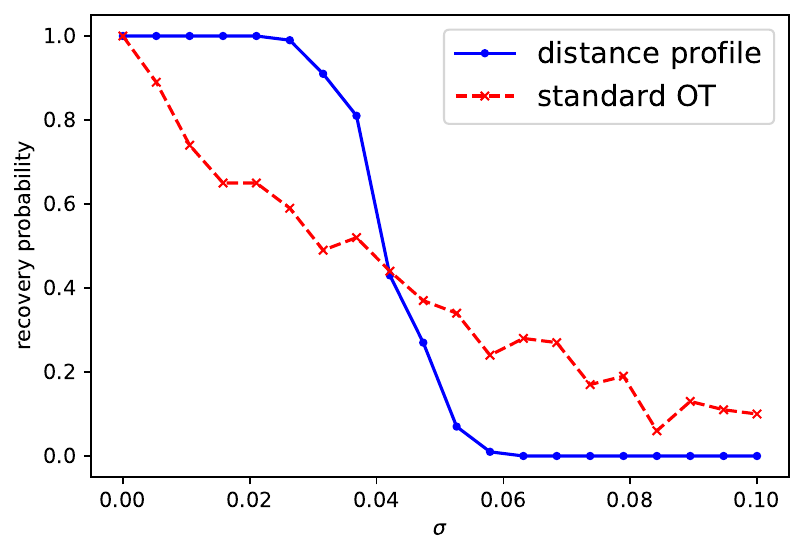}}
    \subfloat[Changing all coordinates]{\includegraphics[trim=0.1cm 0.2cm 0.1cm 0.2cm, clip=true, width=0.44\textwidth]{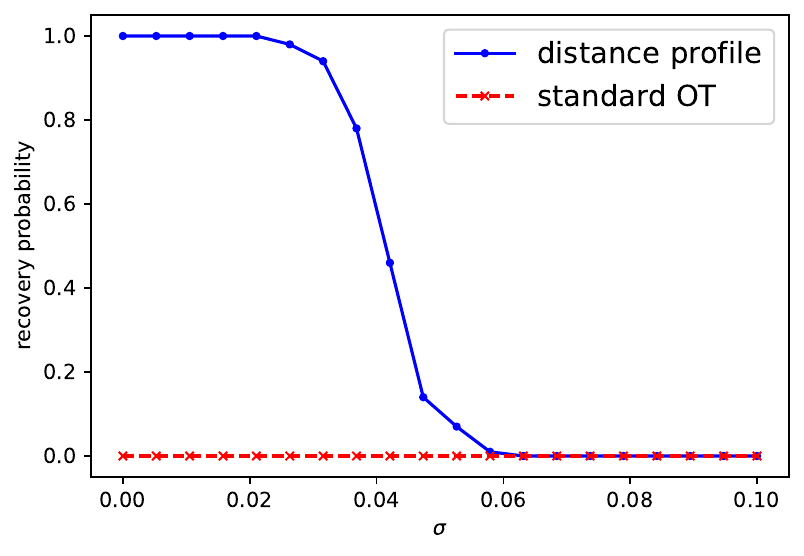}}
    \caption{Noise stability results under the setting of Theorem \ref{thm:noise_stability} with $\xi_i$'s and $\zeta_i$'s following $N(0, \sigma^2 I_d)$, where $d = 10$ and $n = 100$. The plots show the probability $\P(\hat{\pi} = \pi^\ast)$, where $\hat{\pi}$ corresponds to the output of Algorithm \ref{alg:distance_profile_matching_assignment} for the blue lines (distance profile) and to \eqref{eq:feature_matching_linear_assignment} for the red lines (standard OT); here, the probability is estimated by repeating each setting 100 times. For (a), the rigid transformation is based on a rotation matrix that only affects the first two coordinates, while (b) is based on a random rotation matrix changing all $d$ coordinates.}
    \label{fig:noise_stability}
\end{figure}

\subsection{Applications}
\label{sec:data_applications}

\subsubsection{Shape Correspondence}
We demonstrate how Algorithm \ref{alg:distance_profile_matching} in practice using stylized data examples. We apply it to the shape correspondence using CAPOD data \citep{papadakis2014canonically}. Figure \ref{fig:dolphin_matching}(a) and Figure \ref{fig:dolphin_matching}(d) show the two 3D images representing dolphins with different poses, from which we obtain point clouds $X_1, \ldots, X_n$ and $Y_1, \ldots, Y_m$ as shown in Figure \ref{fig:dolphin_matching}(b) and Figure \ref{fig:dolphin_matching}(e), respectively, where $n = m = 285$ in this example. We want to match the points corresponding to the same part of the dolphins, rather than finding a one-to-one correspondence. As explained earlier, it is reasonable to match similar points to the same point. For instance, there are many points concentrated around the nose---called the rostrum---of the source dolphin, which are expected to be matched to any points concentrated around the rostrum of the target dolphin. Hence, we first divide the dolphins into several components. Here, we use k-means clustering with k = $7$, which indeed cluster points based on the body parts. 

\begin{figure}[!htb]
    \centering
    \subfloat[Source object]{\includegraphics[trim=9cm 11cm 9cm 14cm, clip=true, width=0.33\textwidth]{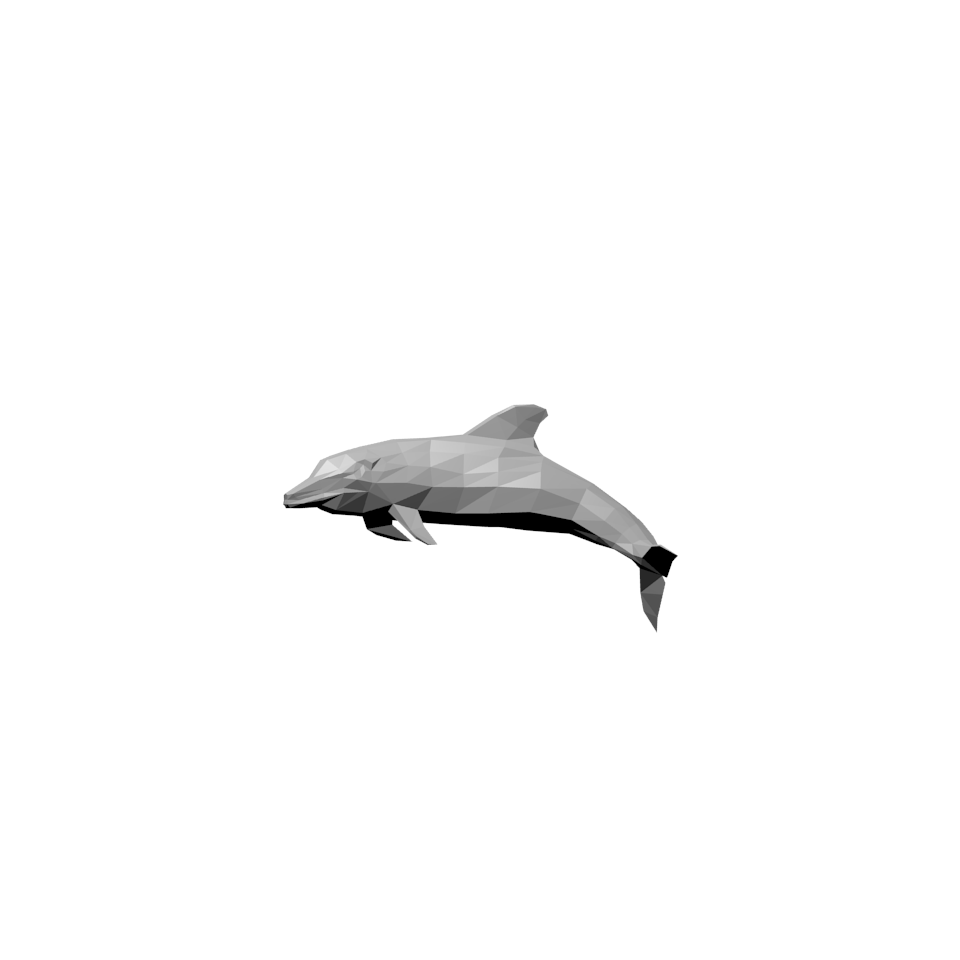}}
    \subfloat[Source points (color = $s(\pi(i))$)]{\includegraphics[trim=1.0cm 3.8cm .5cm 3.0cm, clip=true, width=0.33\textwidth]{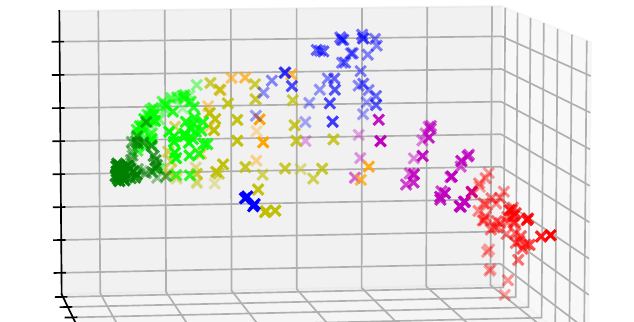}}
    \subfloat[Source points (color = $s(\pi(i))$)]{\includegraphics[trim=1.0cm 3.8cm .5cm 3.0cm, clip=true, width=0.33\textwidth]{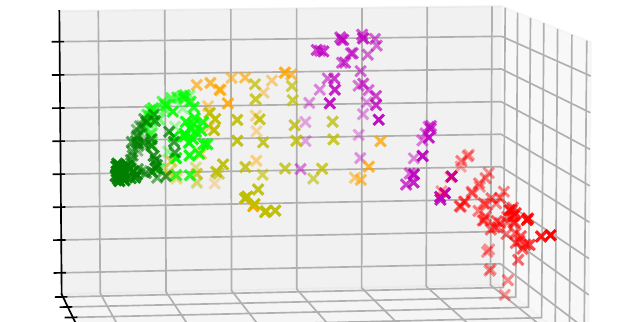}}    
    \qquad
    \subfloat[Target object]{\includegraphics[trim=9cm 12cm 8cm 14cm, clip=true, width=0.33\textwidth]{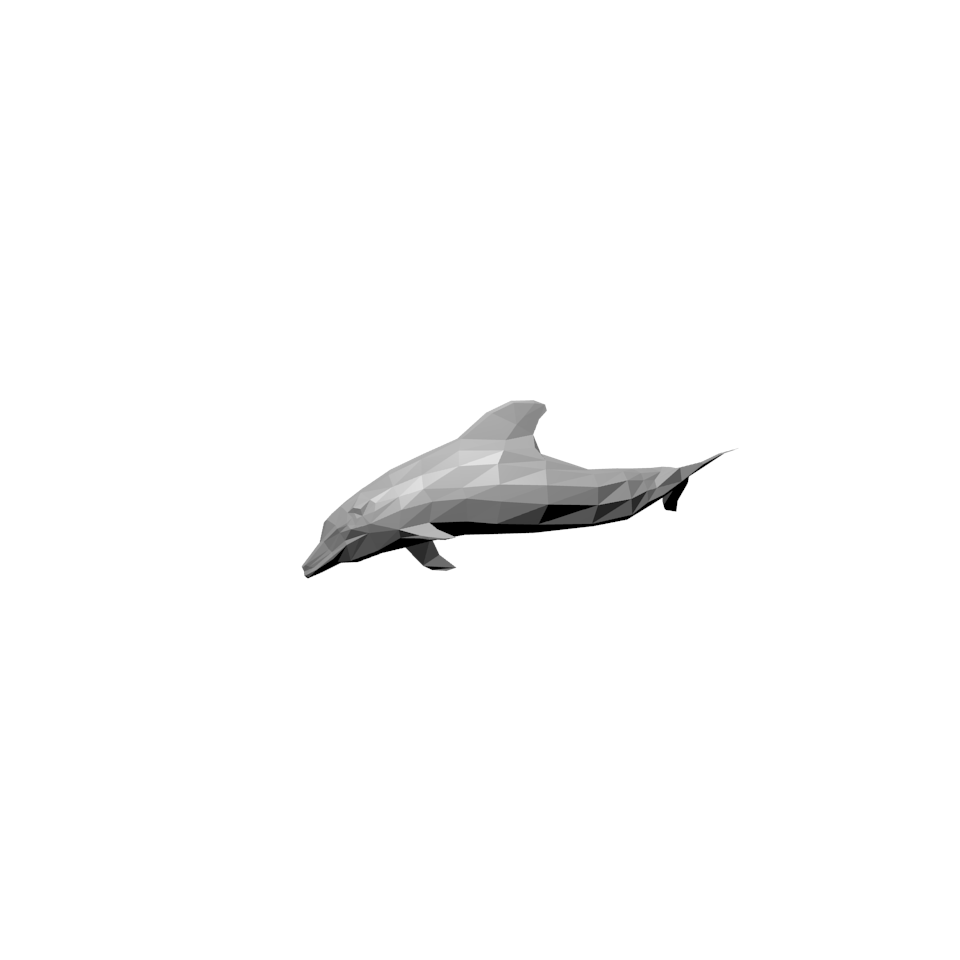}}
    \subfloat[Target points (color = $s(j)$)]{\includegraphics[trim=.3cm 3.5cm 1.5cm 4.0cm, clip=true, width=0.33\textwidth]{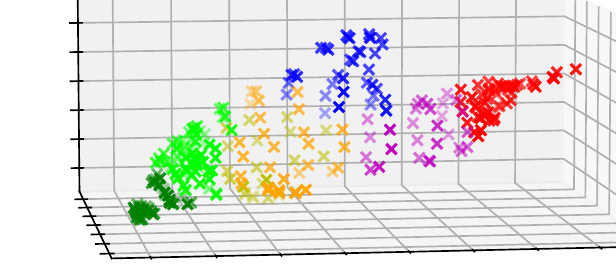}}
    \subfloat[Target points (color = $s(j)$)]{\includegraphics[trim=.3cm 3.5cm 1.5cm 4.0cm, clip=true, width=0.33\textwidth]{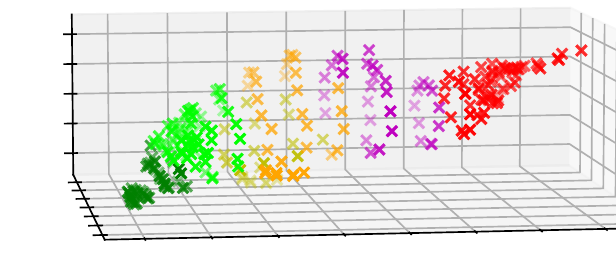}}
    \caption{3D images of dolphins and the corresponding point clouds. (a) and (d) are the source and target 3D images, from which we obtain point clouds $X_1, \ldots, X_n$ and $Y_1, \ldots, Y_m$, with $n = m = 285$, as shown in (b) and (e), respectively. We apply k-means clustering to both point clouds with k $= 7$ to cluster them based on the body parts. (e) shows the target points $Y_j$'s colored by $s(j)$'s, where $s(j) \in \{1, \ldots, 7\}$ denotes the index of the component that $Y_j$ belongs to. (b) shows the source points $X_i$'s colored by $s(\pi(i))$'s, where $\pi$ is obtained by applying Algorithm \ref{alg:distance_profile_matching}. (c) and (f) repeat (b) and (e), respectively, after removing 16 points around the fin of the target dolphin in (e) and clustering again with k $= 6$.}
    \label{fig:dolphin_matching}
\end{figure}

Figure \ref{fig:dolphin_matching}(e) shows the target points $Y_j$'s colored by $s(j)$'s, where $s(j) \in \{1, \ldots, 7\}$ is the index of the component that $Y_j$ belongs to. Now, we run Algorithm \ref{alg:distance_profile_matching} and obtain the output $\pi \colon [n] \to [m]$. For each $X_i$, we verify if the part that it belongs to coincides with the part that its match $Y_{\pi(i)}$ belongs to, which is indexed by $s(\pi(i))$. To this end, Figure \ref{fig:dolphin_matching}(b) shows the source points $X_i$'s colored by $s(\pi(i))$'s. We can observe that the proposed procedure recovers the corresponding regions mostly correctly. For instance, for $X_i$'s corresponding to the rostrum of the source dolphin, their colors $s(\pi(i))$'s (green) indicate that they are matched to the points corresponding to the rostrum of the target dolphin. Similarly, the points corresponding to the head (lime), tail (red), dorsal fin (blue), and peduncle (magenta) are mostly matched correctly, while there are some mismatches, for instance, among the points around the pectoral fin. Figure \ref{fig:dolphin_matching}(c) and Figure \ref{fig:dolphin_matching}(f) repeat the same procedure after removing 16 points around the fin of the target dolphin in Figure \ref{fig:dolphin_matching}(e) and clustering again with k = $6$. Figure \ref{fig:dolphin_matching}(c) shows that the matching result is similar to the previous one. In this case, though the points around the fin of the source dolphin do not have corresponding points in the target dolphin and thus we may ignore matching them, the proposed procedure still matches them to the points around the peduncle (magenta) of the target dolphin, which is a cluster adjacent to the removed fin. 

Figure \ref{fig:horse_matching} repeats the same procedure for the 3D images of horses. Again, we can verify that the proposed procedure matches the corresponding body parts mostly correctly. Figure \ref{fig:horse_matching}(c) and Figure \ref{fig:horse_matching}(f) shows the results after removing 62 points around the two legs on the target horse's left side in Figure \ref{fig:horse_matching}(e) and clustering again with k = $6$. We again observe that the points in the common parts are mostly matched correctly.

\begin{figure}[!htb]
    \centering
    \subfloat[Source object]{\includegraphics[trim=6cm 9cm 7cm 8cm, clip=true, width=0.33\textwidth]{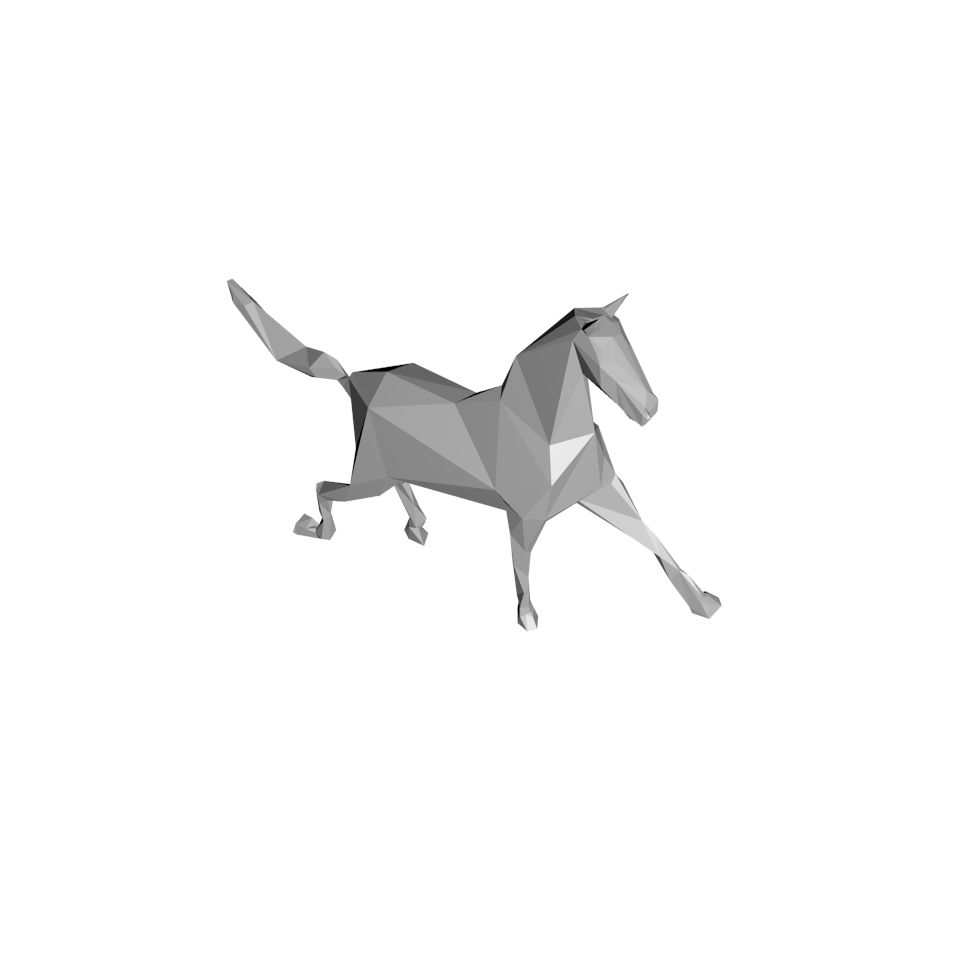}}
    \subfloat[Source points (color = $s(\pi(i))$)]{\includegraphics[trim=1.0cm 2.5cm .5cm 2.0cm, clip=true, width=0.33\textwidth]{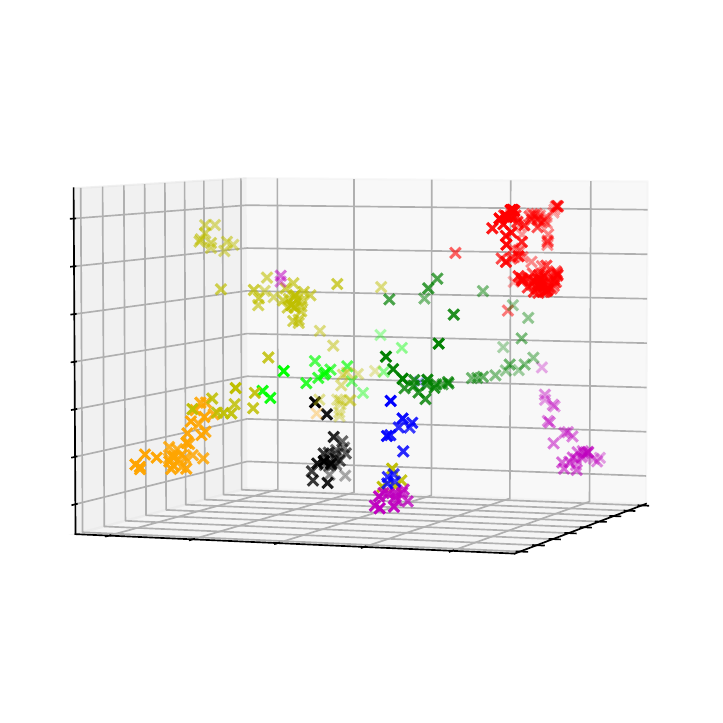}}
    \subfloat[Source points (color = $s(\pi(i))$)]{\includegraphics[trim=1.0cm 2.5cm .5cm 2.0cm, clip=true, width=0.33\textwidth]{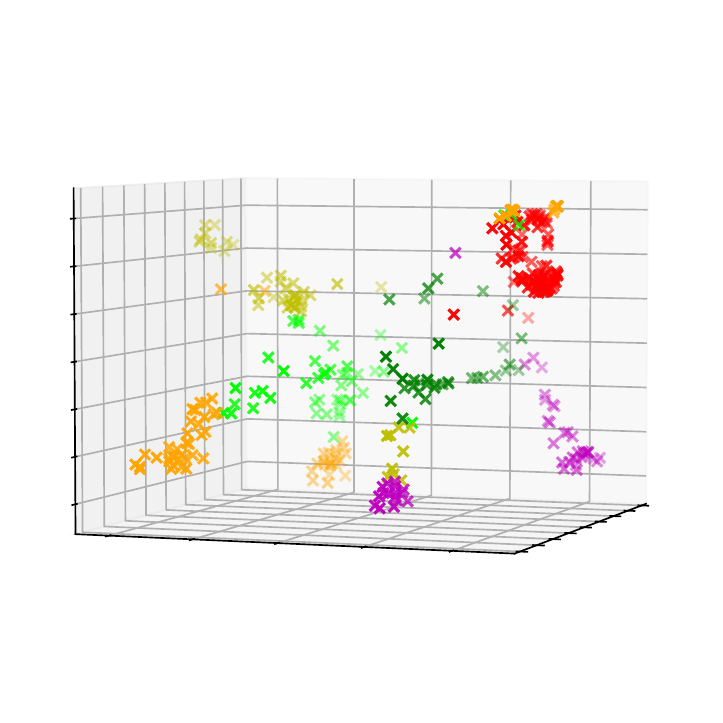}}    
    \qquad
    \subfloat[Target object]{\includegraphics[trim=6cm 9cm 7cm 8cm, clip=true, width=0.33\textwidth]{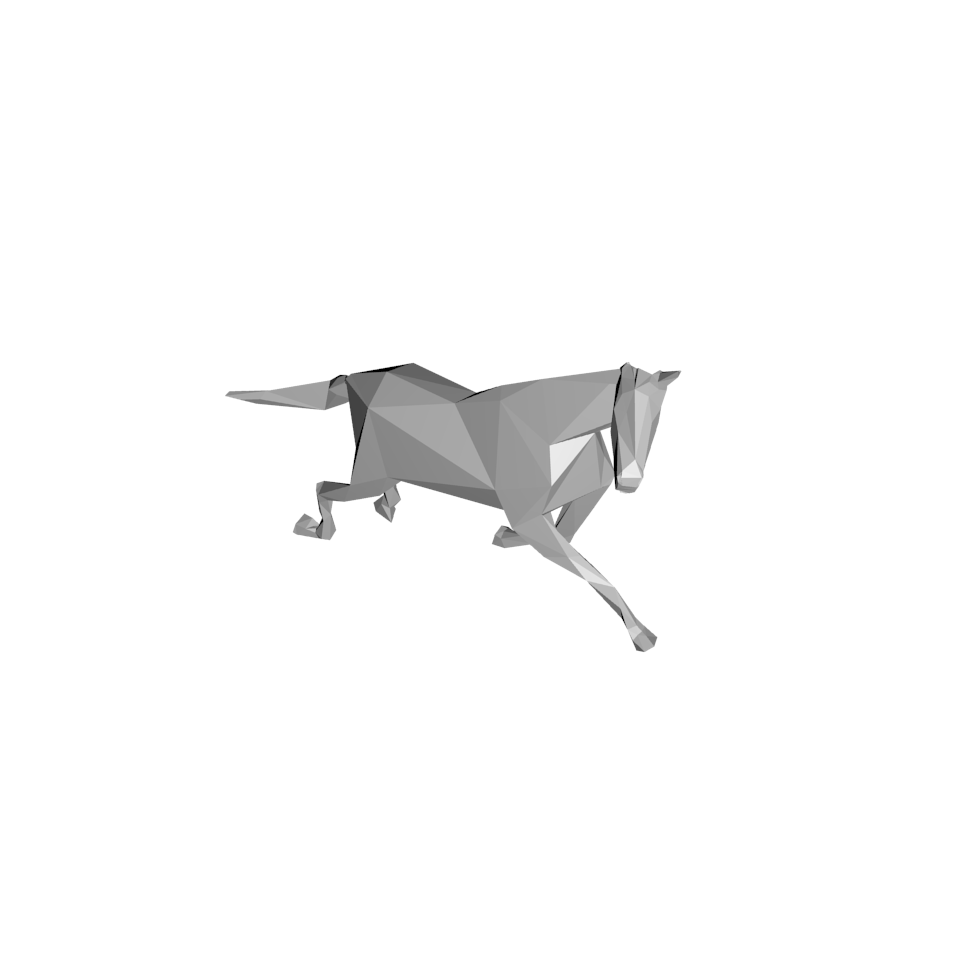}}
    \subfloat[Target points (color = $s(j)$)]{\includegraphics[trim=.3cm 2.5cm 1.5cm 2.0cm, clip=true, width=0.33\textwidth]{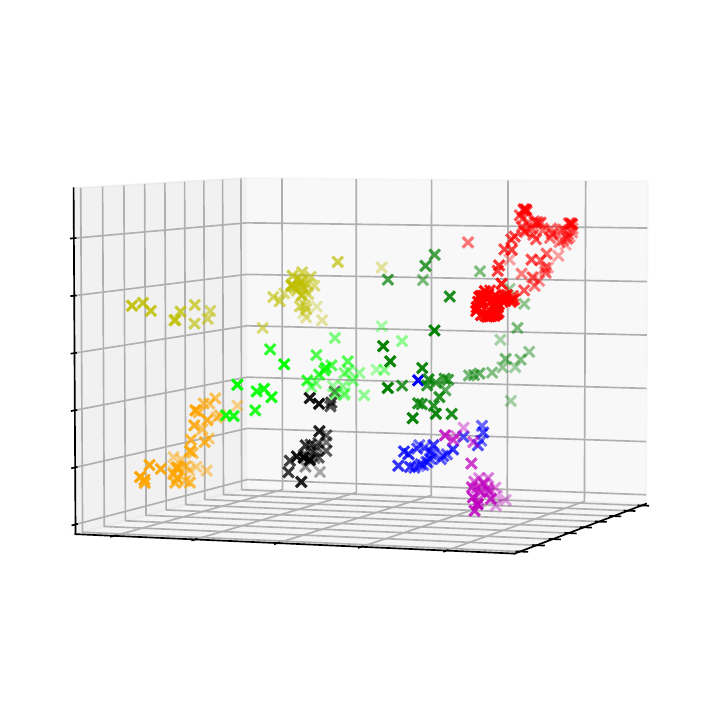}}
    \subfloat[Target points (color = $s(j)$)]{\includegraphics[trim=.3cm 2.5cm 1.5cm 2.0cm, clip=true, width=0.33\textwidth]{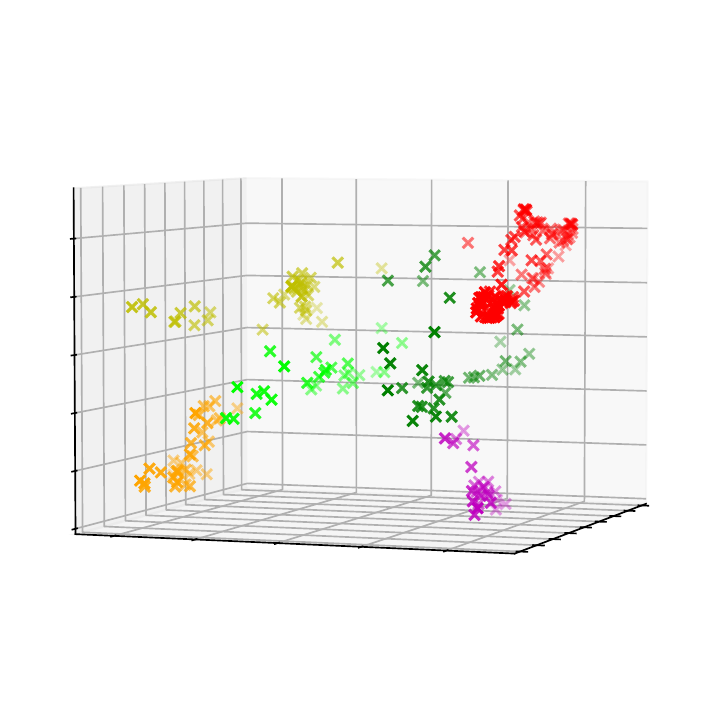}}
    \caption{3D images of horses and the corresponding point clouds. (a) and (d) are the source and target 3D images, from which we obtain point clouds $X_1, \ldots, X_n$ and $Y_1, \ldots, Y_m$, with $n = m = 346$, as shown in (b) and (e), respectively. We apply k-means clustering to both point clouds with k $= 8$ to cluster them based on the body parts. (e) shows the target points $Y_j$'s colored by $s(j)$'s, where $s(j) \in \{1, \ldots, 8\}$ denotes the index of the component that $Y_j$ belongs to. (b) shows the source points $X_i$'s colored by $s(\pi(i))$'s, where $\pi$ is obtained by applying Algorithm \ref{alg:distance_profile_matching}. (c) and (f) repeat (b) and (e), respectively, after removing 62 points around the two legs on the target horse's left side in (e) and clustering again with k $= 6$.}
    \label{fig:horse_matching}    
\end{figure}

\subsubsection{Partial Matching for Alignment of Cryo-EM Structures}
\label{sec:cryo-em}
We apply the developed partial matching procedure to align cryogenic electron microscopy (cryo-EM) structures of human $\gamma$-secretase. Cryo-EM has gained significant attention in the past decade as a powerful tool for unveiling the structures of biomolecules at near-atomic resolution. However, many computational challenges arise; see \cite{bendory2020single} for a review. Aligning cryo-EM structures is particularly challenging due to high noise levels and different conformations. Figure \ref{fig:gamma-secretase}(a) shows the 3D model of human $\gamma$-secretase, which is thought to be involved in the development of Alzheimer's disease. The cryo-EM data\footnote{Listed as 5A63 at the Protein Data Bank (PDB) \url{https://www.rcsb.org/structure/5A63} and as EMD-3061 at the Electron Microscopy Data Bank (EMDB) \url{https://www.ebi.ac.uk/emdb/EMD-3061}.} of human $\gamma$-secretase \citep{bai2015atomic} essentially represents a probability density function of the intensities of the object, stored as voxels on a 3D grid. Figures \ref{fig:gamma-secretase}(b)-(d) show the projections of the cryo-EM data along the $x$, $y$, and $z$ axes, respectively, which clearly demonstrate high noise levels.

\begin{figure}[!tb]
    \centering
    \subfloat[3D Model]{\includegraphics[width=0.25\textwidth]{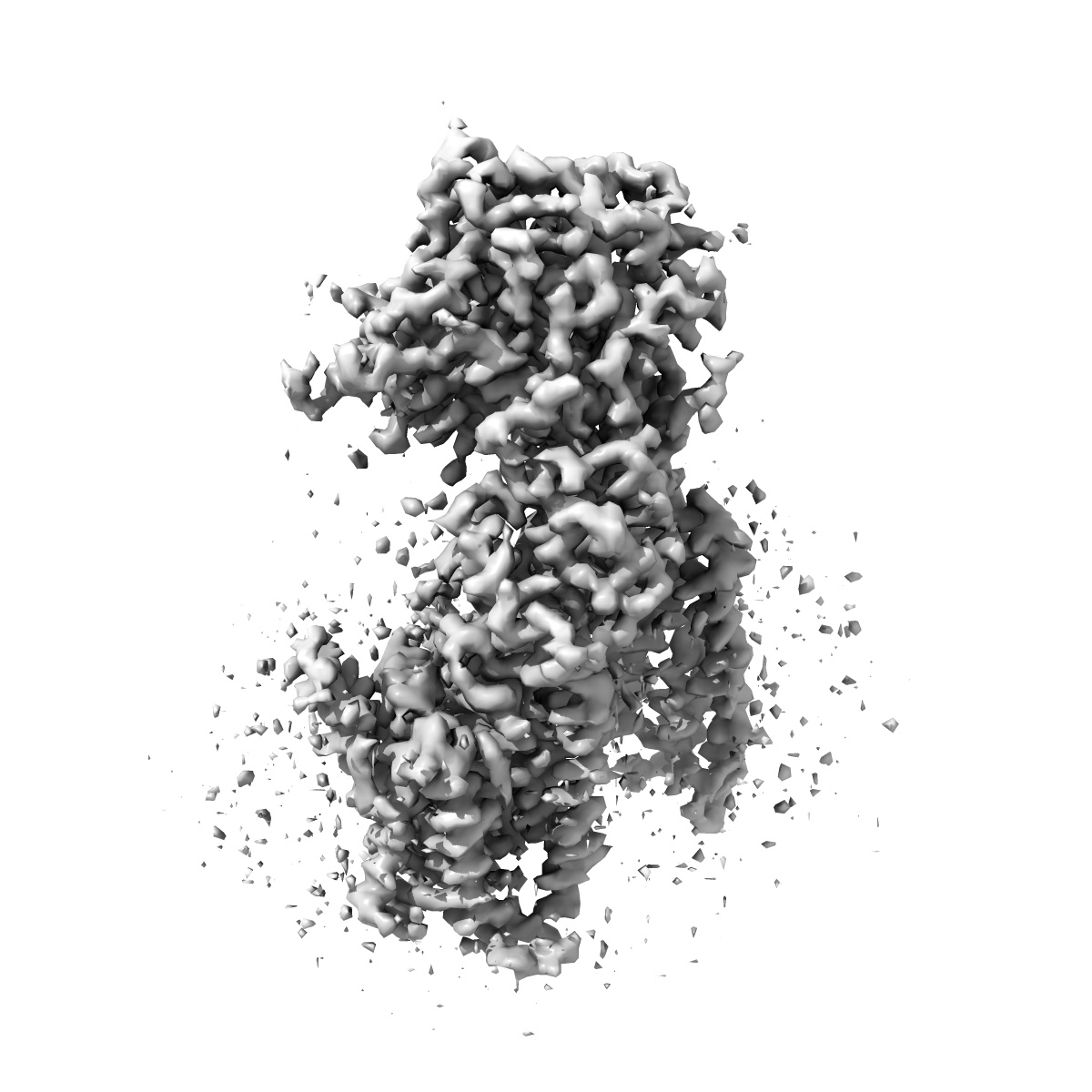}}
    \subfloat[Projection ($x$-axis)]{\includegraphics[trim=0.1cm 0.1cm 0.1cm 0.1cm, clip=true, width=0.25\textwidth]{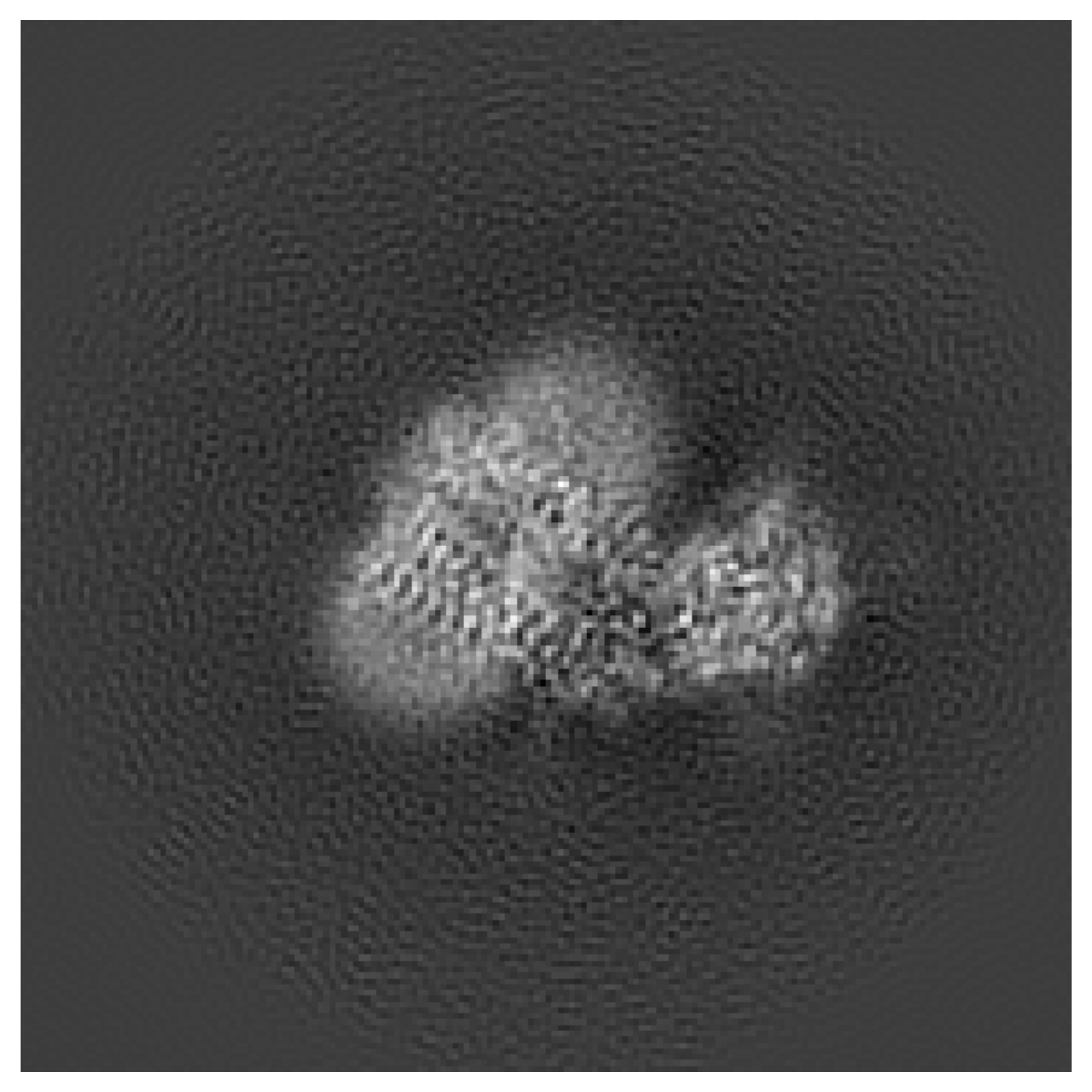}}
    \subfloat[Projection ($y$-axis)]{\includegraphics[trim=0.1cm 0.1cm 0.1cm 0.1cm, clip=true, width=0.25\textwidth]{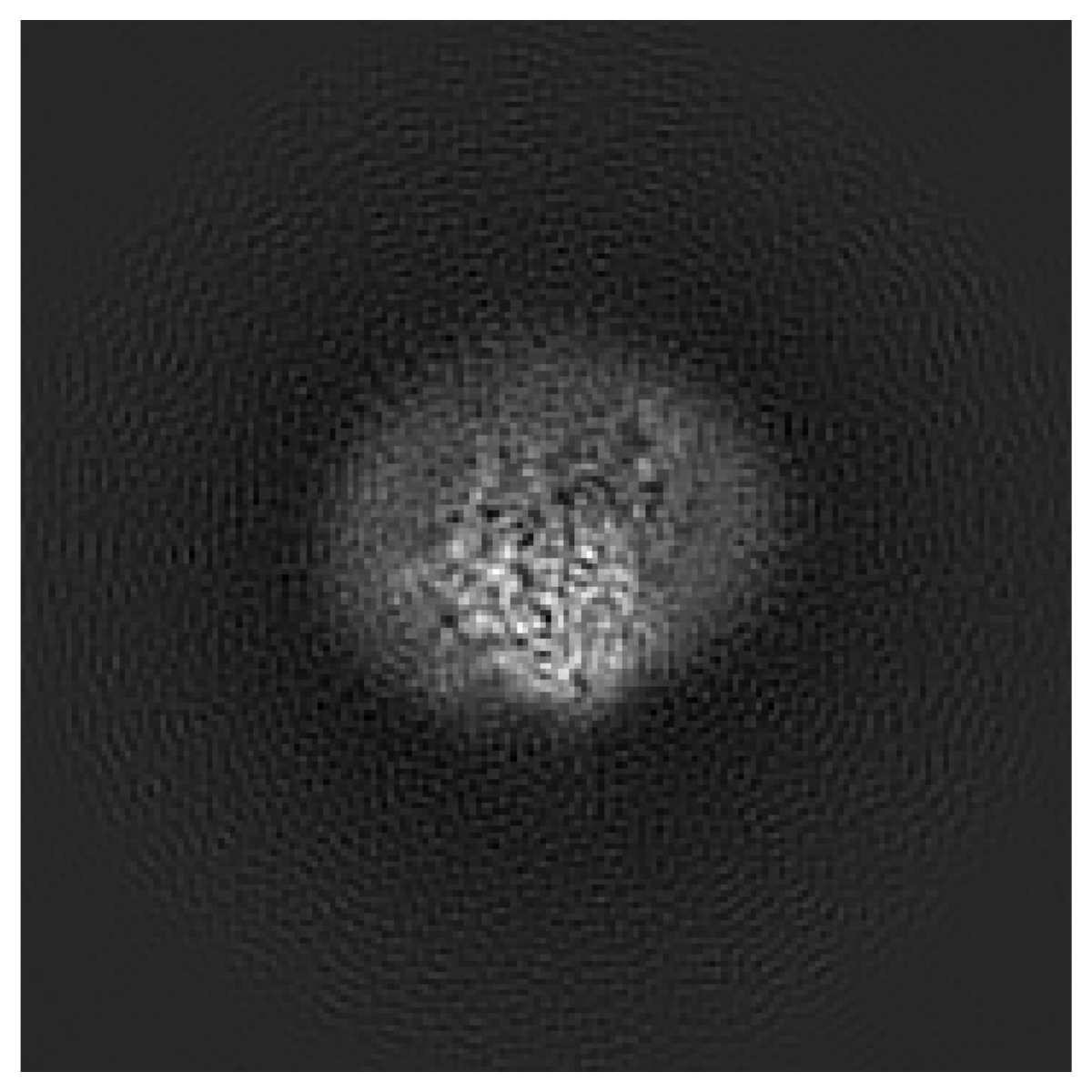}}
    \subfloat[Projection ($z$-axis)]{\includegraphics[trim=0.1cm 0.1cm 0.1cm 0.1cm, clip=true, width=0.25\textwidth]{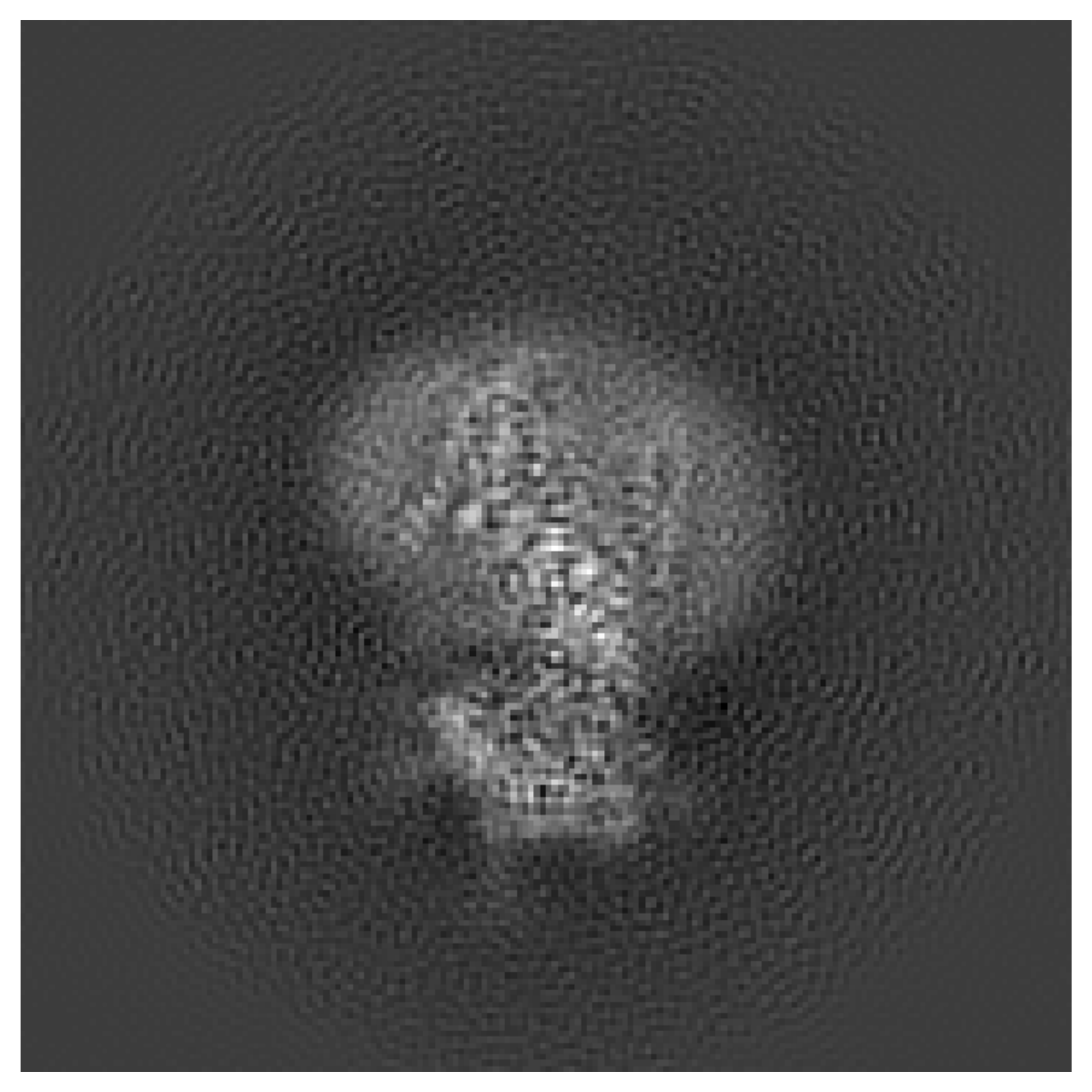}}
    \caption{Human $\gamma$-secretase. (a) shows the 3D model of human $\gamma$-secretase, and (b)-(d) show the projections of the cryo-EM data along the $x$, $y$, and $z$ axes, respectively.}
    \label{fig:gamma-secretase}
\end{figure}

Alignment problems in cryo-EM are often tackled by using the density functions \citep{singer2024alignment,harpaz2023three} with suitable transformations or by obtaining the point clouds from the density functions \citep{riahi2023alignot,tajmir2025alignment}, say, utilizing the topology representing network algorithm \citep{martinetz1994topology}. The proposed algorithm can be applied to the latter setting. Here, we sample two point clouds from the cryo-EM data, one from the original density function and the other from the density function after applying a rigid transformation; see Figure \ref{fig:cryo-em}(a). The goal is to recover the rigid transformation from the point clouds. 

The proposed partial matching---Algorithm \ref{alg:distance_profile_matching}---provides correspondences between the points. Once we have the correspondences, we can estimate the rigid transformation using the standard orthogonal Procrustes problem---after centering---between matrices $\mathbb{X}, \mathbb{Y} \in \R^{n \times d}$ whose rows are the point clouds $X_1, \ldots, X_n$ and $Y_{\pi(1)}, \ldots, Y_{\pi(n)}$, respectively, where $\pi$ is the output of Algorithm \ref{alg:distance_profile_matching}. Figure \ref{fig:cryo-em}(b) shows the aligned point clouds, which seems reasonable. However, for better results, it is beneficial to subset the good matches by using the threshold $\rho$ in Algorithm \ref{alg:distance_profile_matching}. The crux is that we only need to partially match the points for alignment, and focusing on good matches can improve the alignment results. Figures \ref{fig:cryo-em}(c) and \ref{fig:cryo-em}(d) apply Algorithm \ref{alg:distance_profile_matching} with $\rho$ being the $0.5$ and $0.1$ quantiles of $W(1, \pi(1)), \ldots, W(n, \pi(n))$ to use the best 50\% and 10\% of the correspondences, respectively. In these cases, we obtain the set of index $I$ from the output of Algorithm \ref{alg:distance_profile_matching} and estimate the rigid transformation using the submatrices $\mathbb{X}_I, \mathbb{Y}_I$ of $\mathbb{X}, \mathbb{Y}$, respectively, corresponding to the rows indexed by $I$. Figures \ref{fig:cryo-em}(c) and \ref{fig:cryo-em}(d) indeed look better than using the full correspondences, which we can quantify by computing the Wasserstein-2 distance between the aligned point cloud and the target point cloud. The Wasserstein-2 distance between the aligned point cloud and the target point cloud is $10.68$ for the full correspondences, $9.04$ for the best 50\%, and $8.31$ for the best 10\%, which confirms that the partial matching improves the alignment results.

\begin{figure}[!ht]
    \centering
    \subfloat[Point clouds]{\includegraphics[trim=1cm 2cm 1cm 3cm, clip=true, width=0.43\textwidth]{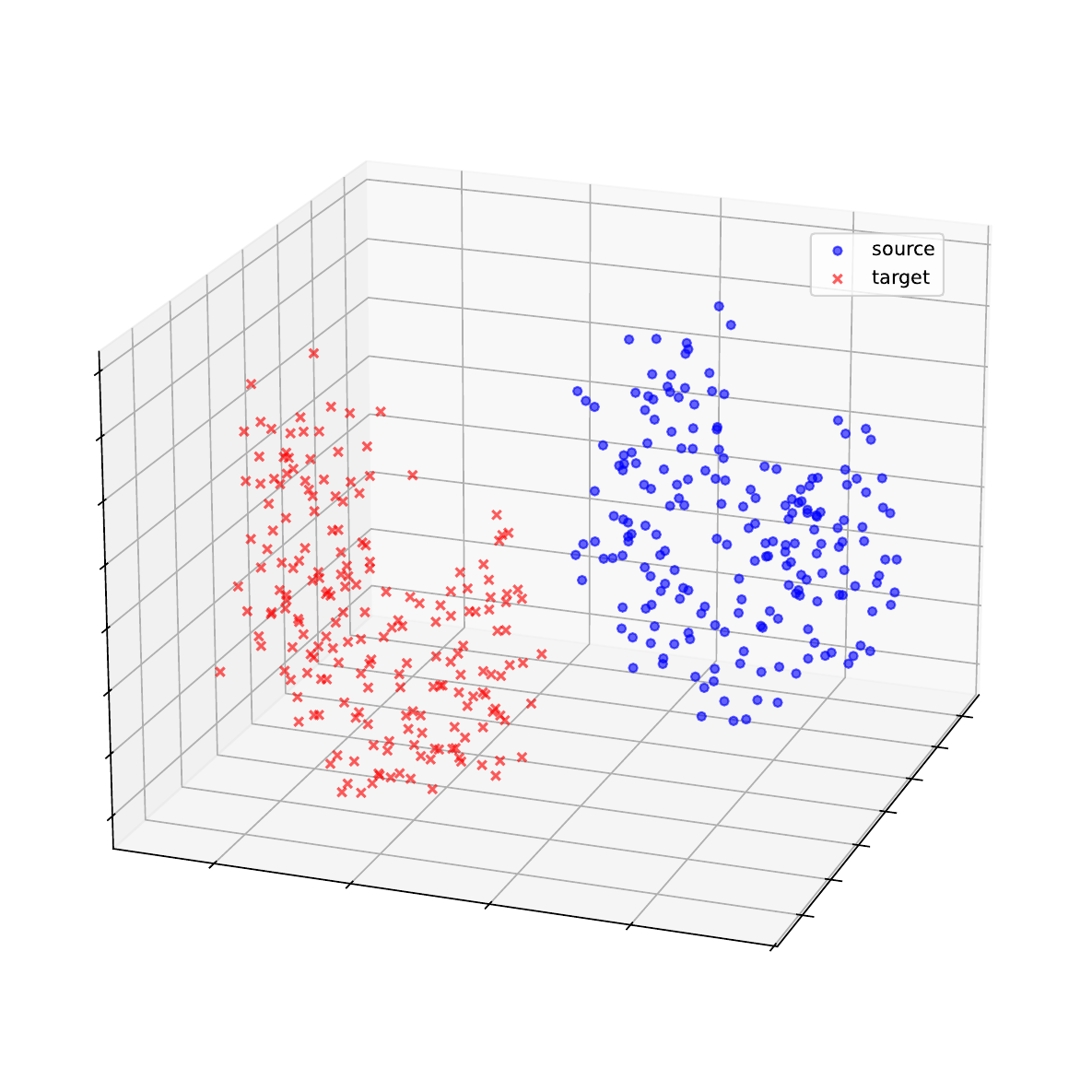}}
    \quad
    \subfloat[Aligned with full correspondences]{\includegraphics[trim=1cm 2cm 1cm 3cm, clip=true, width=0.43\textwidth]{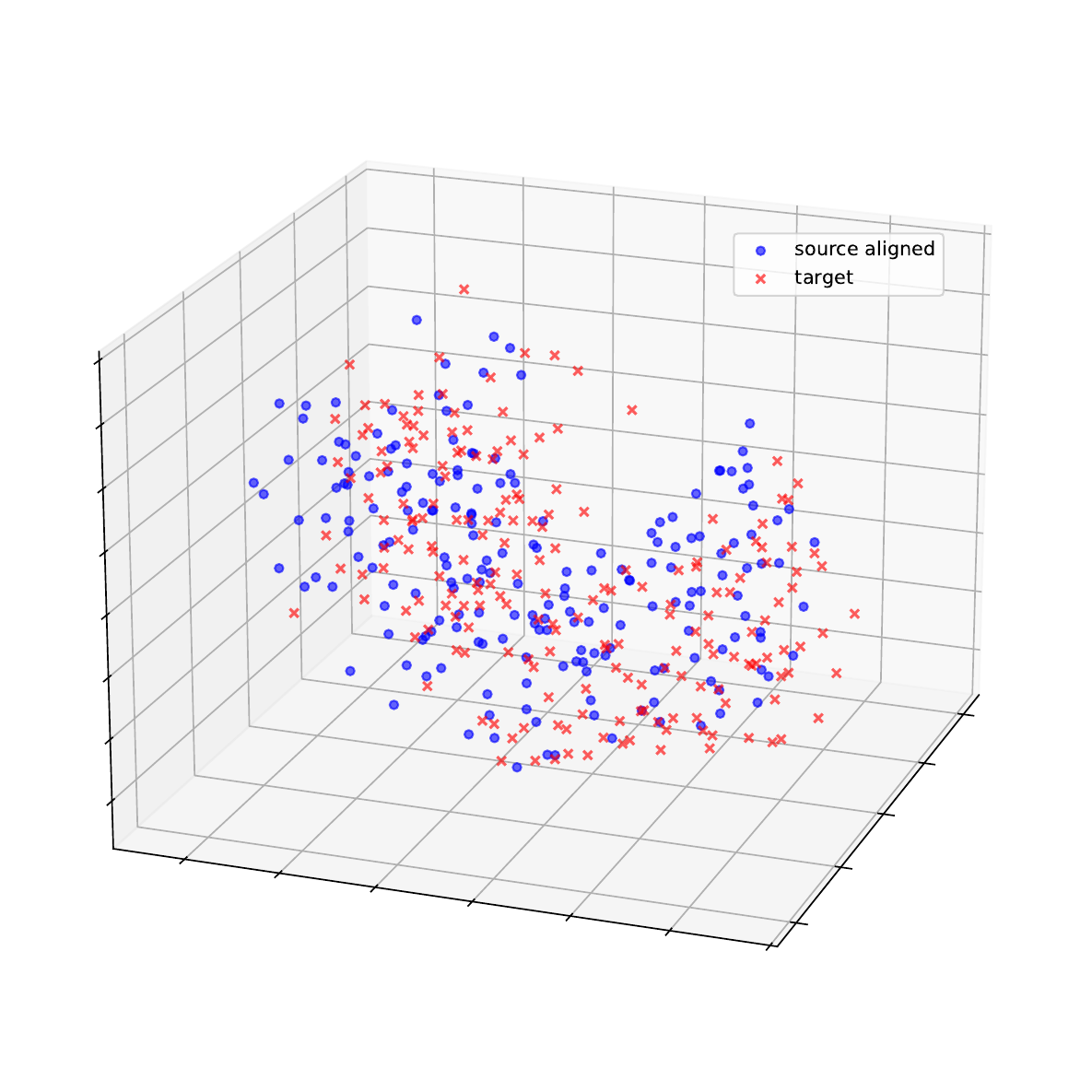}}
    \qquad
    \subfloat[Aligned with best 50\% correspondences]{\includegraphics[trim=1cm 2cm 1cm 3cm, clip=true, width=0.43\textwidth]{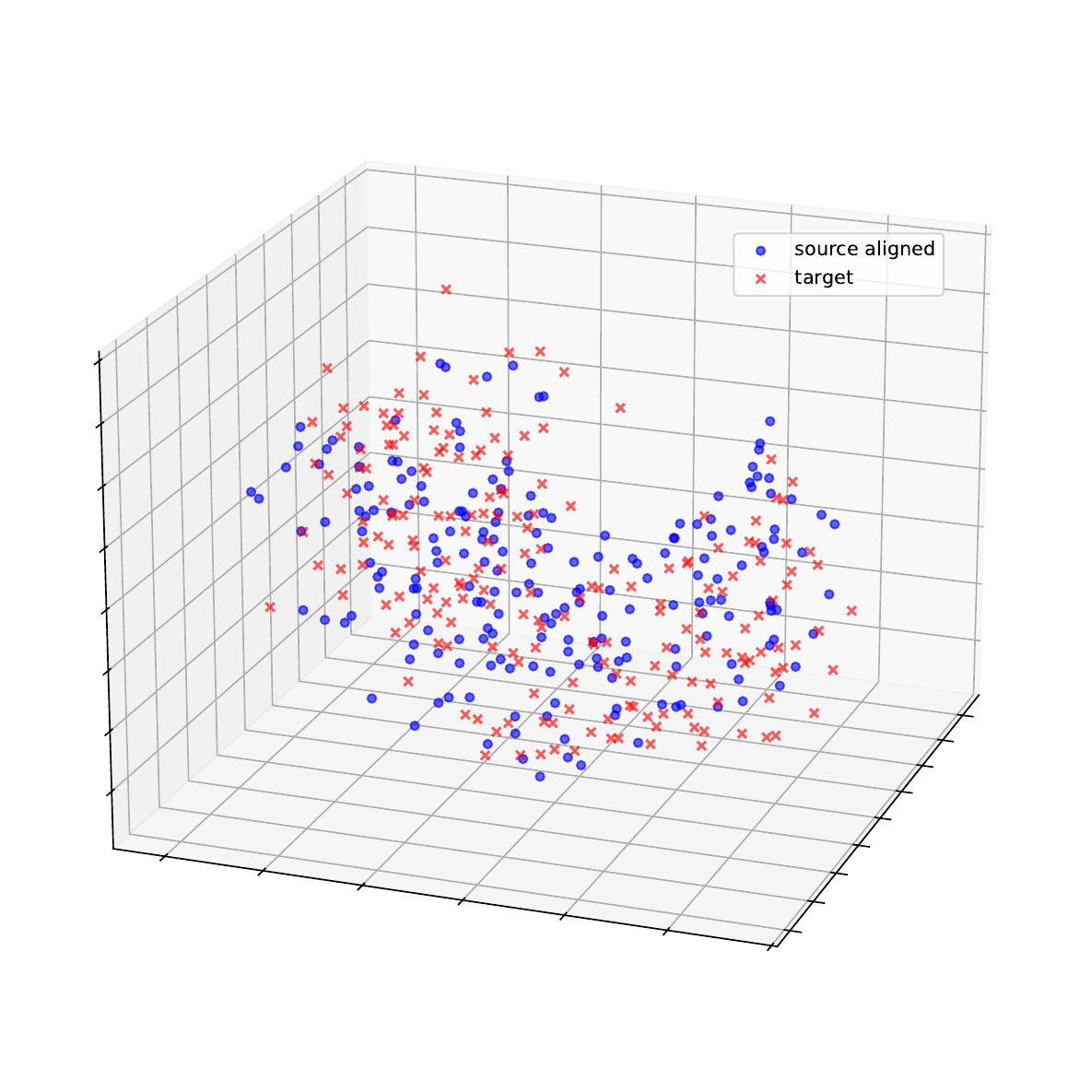}}
    \quad
    \subfloat[Aligned with best 10\% correspondences]{\includegraphics[trim=1cm 2cm 1cm 3cm, clip=true, width=0.43\textwidth]{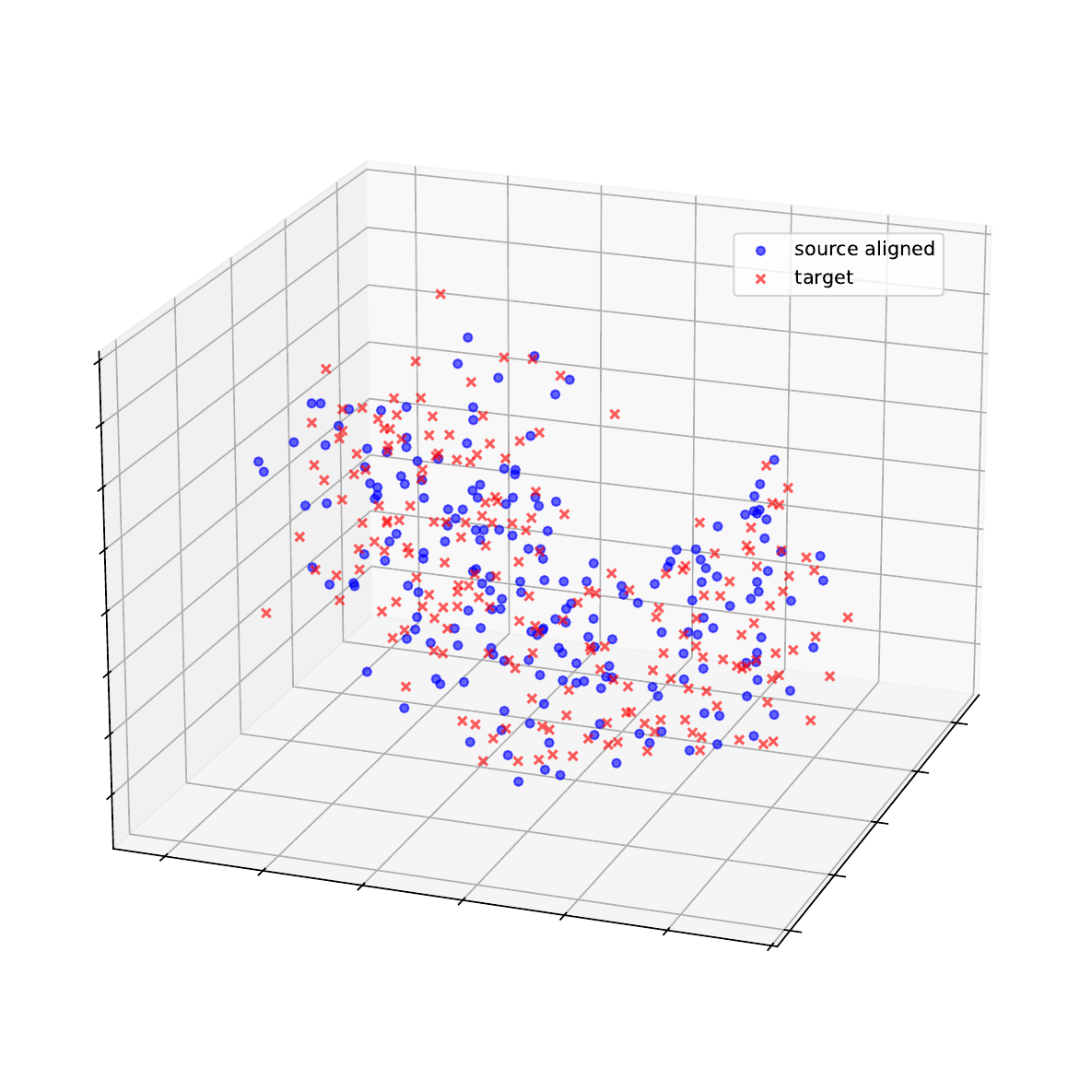}}
    \caption{Point clouds from the cryo-EM data of human $\gamma$-secretase and alignment results. (a) shows the point clouds sampled from the original density function and the density function after applying a rigid transformation, shown as source and target, respectively. For (b)-(c), we first apply Algorithm \ref{alg:distance_profile_matching} to obtain the correspondences and then estimate the rigid transformation using the orthogonal Procrustes problem. (b) shows the aligned point clouds using all the correspondences, while (c) and (d) use the best 50\% and 10\% of the correspondences, respectively, based on the threshold $\rho$ in Algorithm \ref{alg:distance_profile_matching}.}
    \label{fig:cryo-em}
\end{figure}

We repeat the same process for the cryo-EM data of the voltage-gated calcium channel ($\text{Ca}_{\text{v}}1.1$) of rabbit \citep{wu2016structure}, which plays a role as the sensor for excitation-contraction coupling of skeletal muscles. Figure \ref{fig:calcium} shows the 3D model and projections of the cryo-EM data. Figure \ref{fig:cryo-em2} shows the alignment results. Again, we obtain two point clouds from the cryo-EM data, one from the original density function and the other from the density function after applying a rigid transformation; see Figure \ref{fig:cryo-em2}(a). Figures \ref{fig:cryo-em2}(b)-(d) show the aligned point clouds using all the correspondences, the best 50\%, and the best 10\%, respectively. The Wasserstein-2 distance between the aligned point cloud and the target point cloud is $11.10$ for the full correspondences, $8.23$ for the best 50\%, and $9.17$ for the best 10\%. In this case, taking the best 50\% correspondences yields the best partial matching result among the three settings.

\begin{figure}[!tb]
    \centering
    \subfloat[3D Model]{\includegraphics[width=0.25\textwidth]{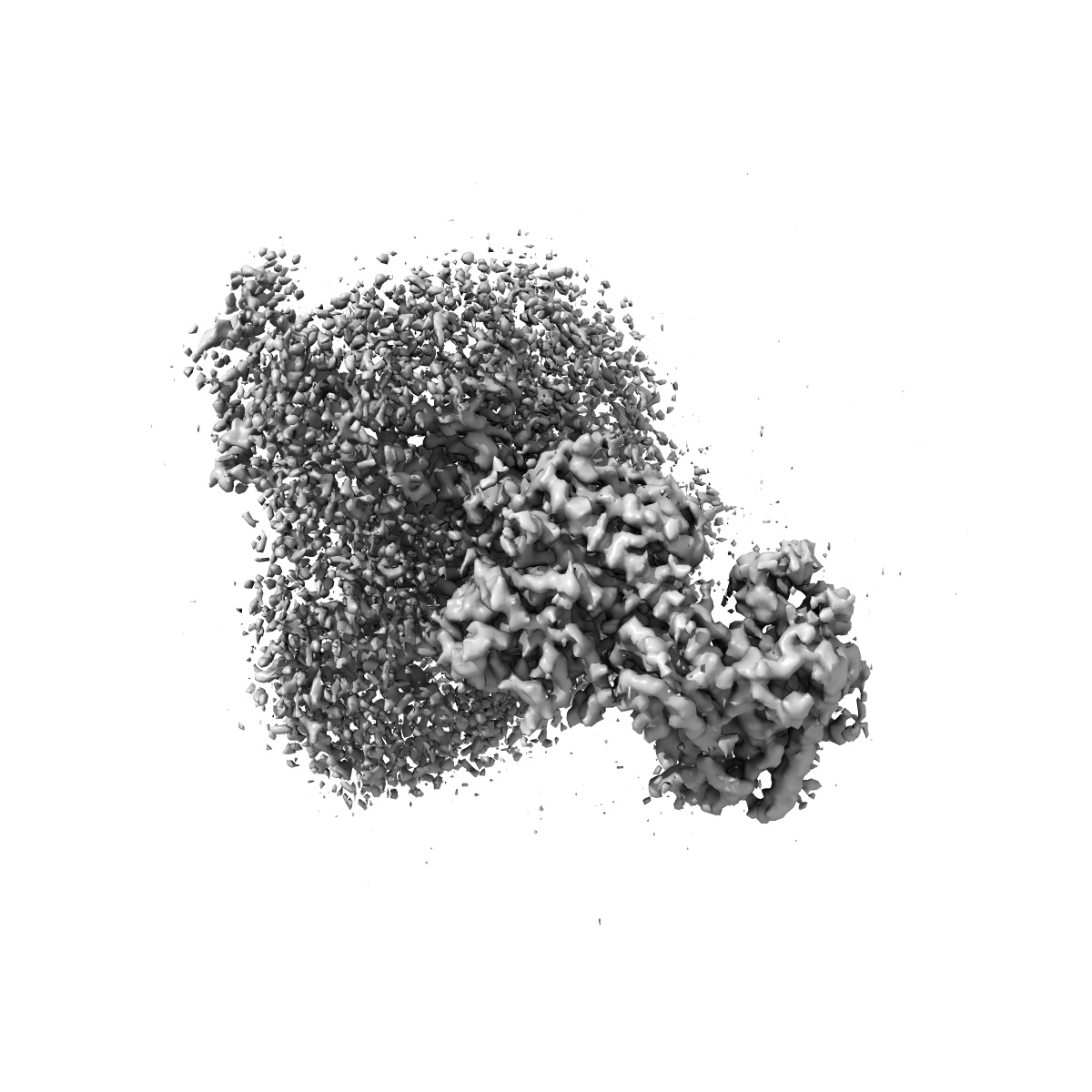}}
    \subfloat[Projection ($x$-axis)]{\includegraphics[trim=0.1cm 0.1cm 0.1cm 0.1cm, clip=true, width=0.25\textwidth]{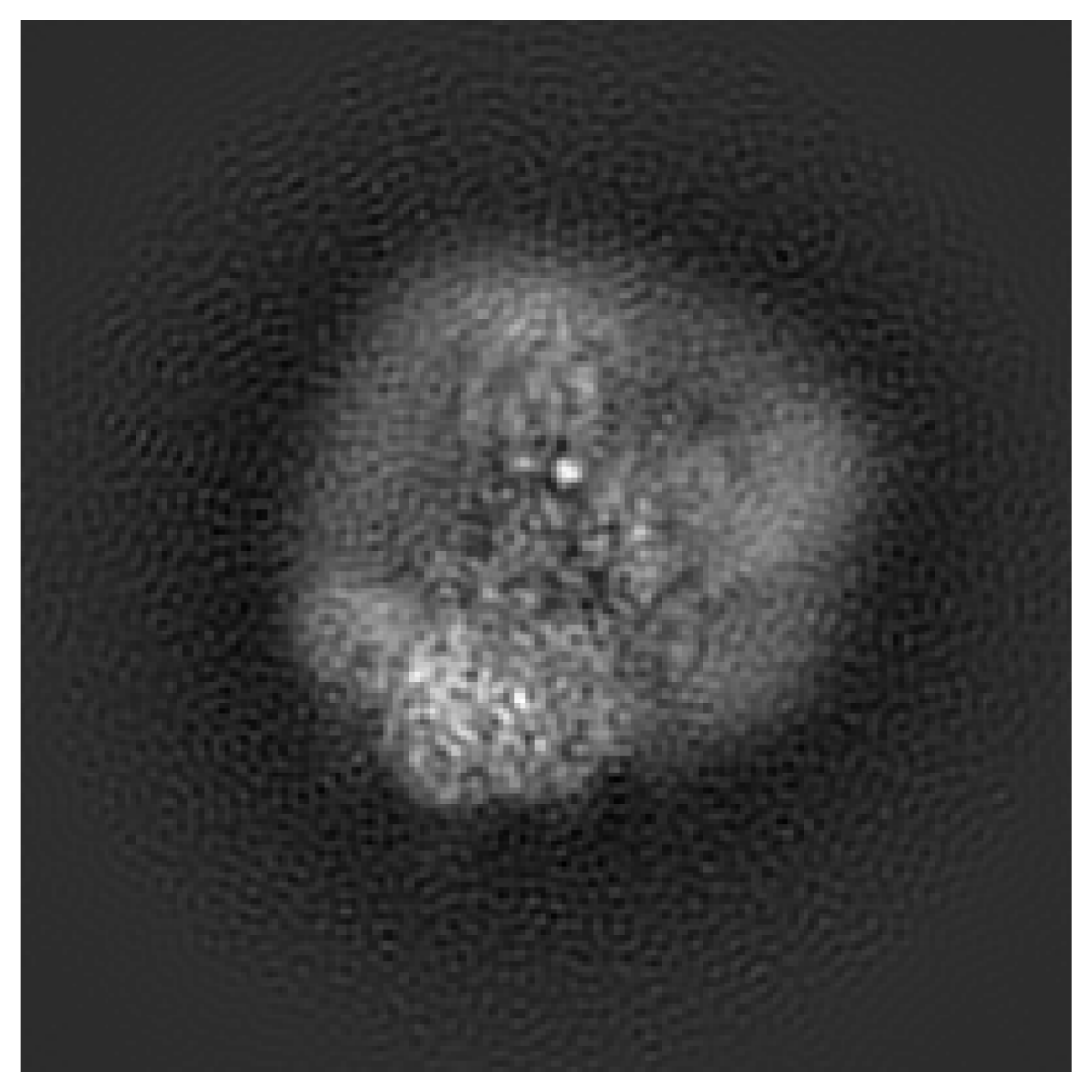}}
    \subfloat[Projection ($y$-axis)]{\includegraphics[trim=0.1cm 0.1cm 0.1cm 0.1cm, clip=true, width=0.25\textwidth]{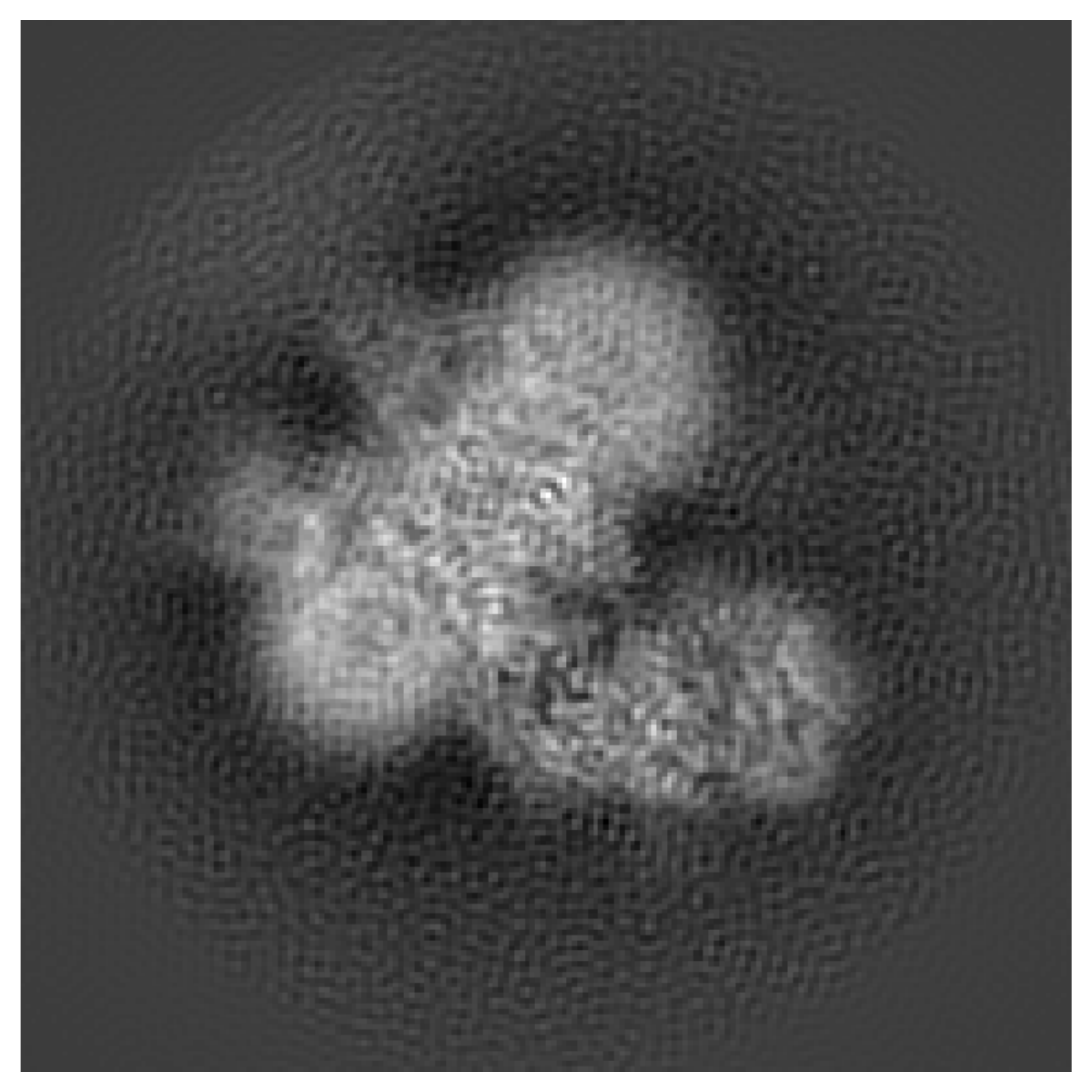}}
    \subfloat[Projection ($z$-axis)]{\includegraphics[trim=0.1cm 0.1cm 0.1cm 0.1cm, clip=true, width=0.25\textwidth]{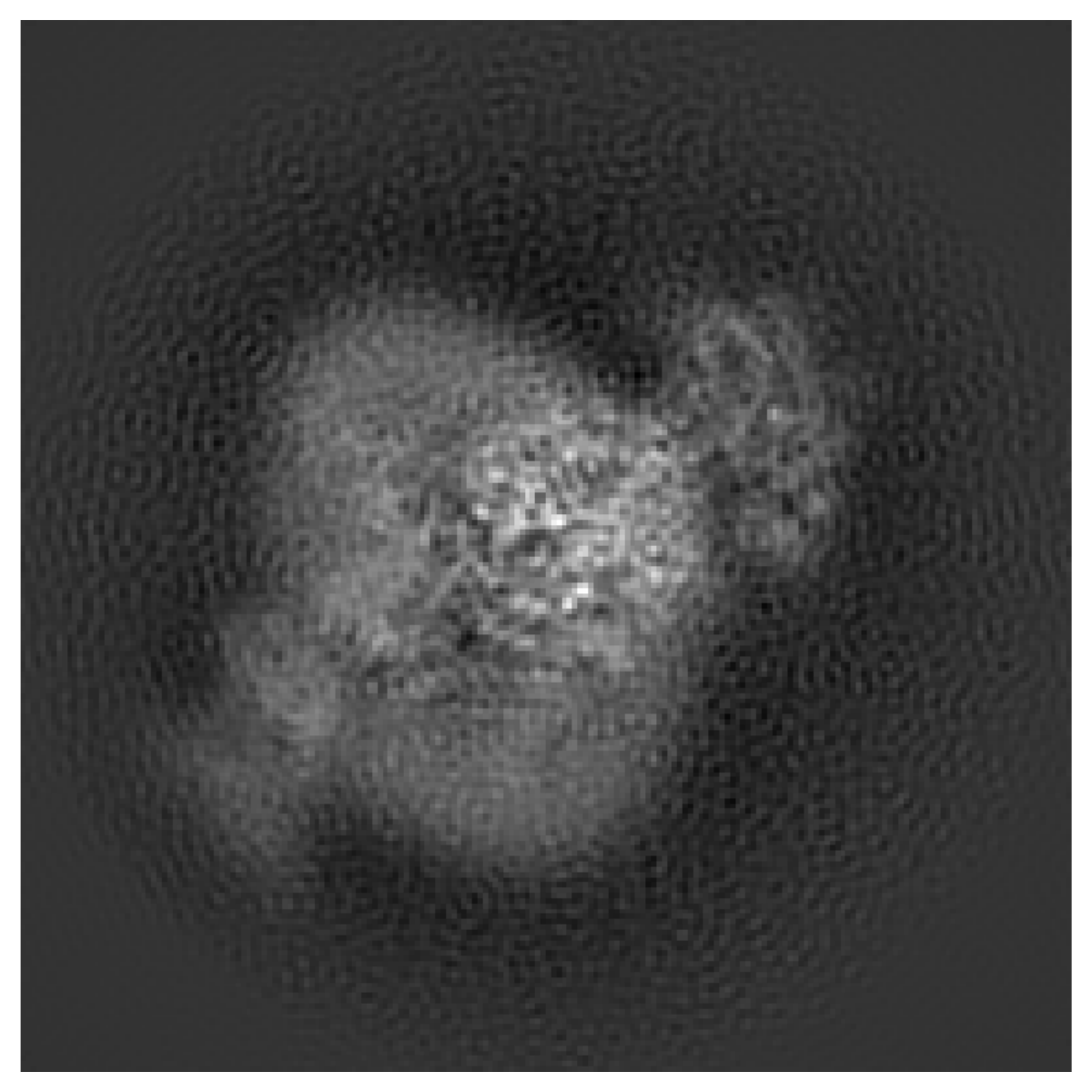}}
    \caption{Rabbit $\text{Ca}_{\text{v}}1.1$. (a) shows the 3D model of $\text{Ca}_{\text{v}}1.1$, and (b)-(d) show the projections of the cryo-EM data along the $x$, $y$, and $z$ axes, respectively.}
    \label{fig:calcium}
\end{figure}

\begin{figure}[!hb]
    \centering
    \subfloat[Point clouds]{\includegraphics[trim=1cm 2cm 1cm 3cm, clip=true, width=0.43\textwidth]{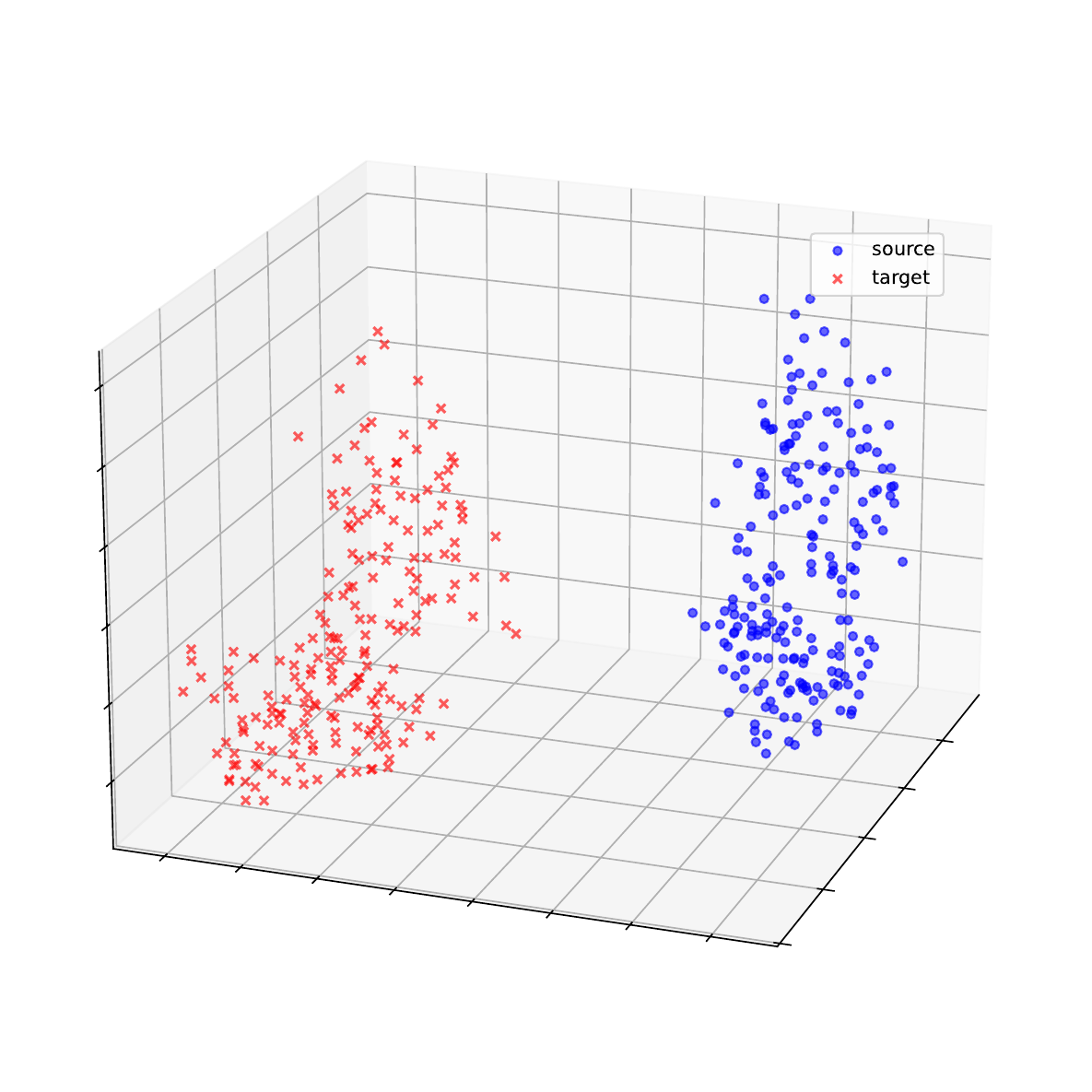}}
    \quad
    \subfloat[Aligned with full correspondences]{\includegraphics[trim=1cm 2cm 1cm 3cm, clip=true, width=0.43\textwidth]{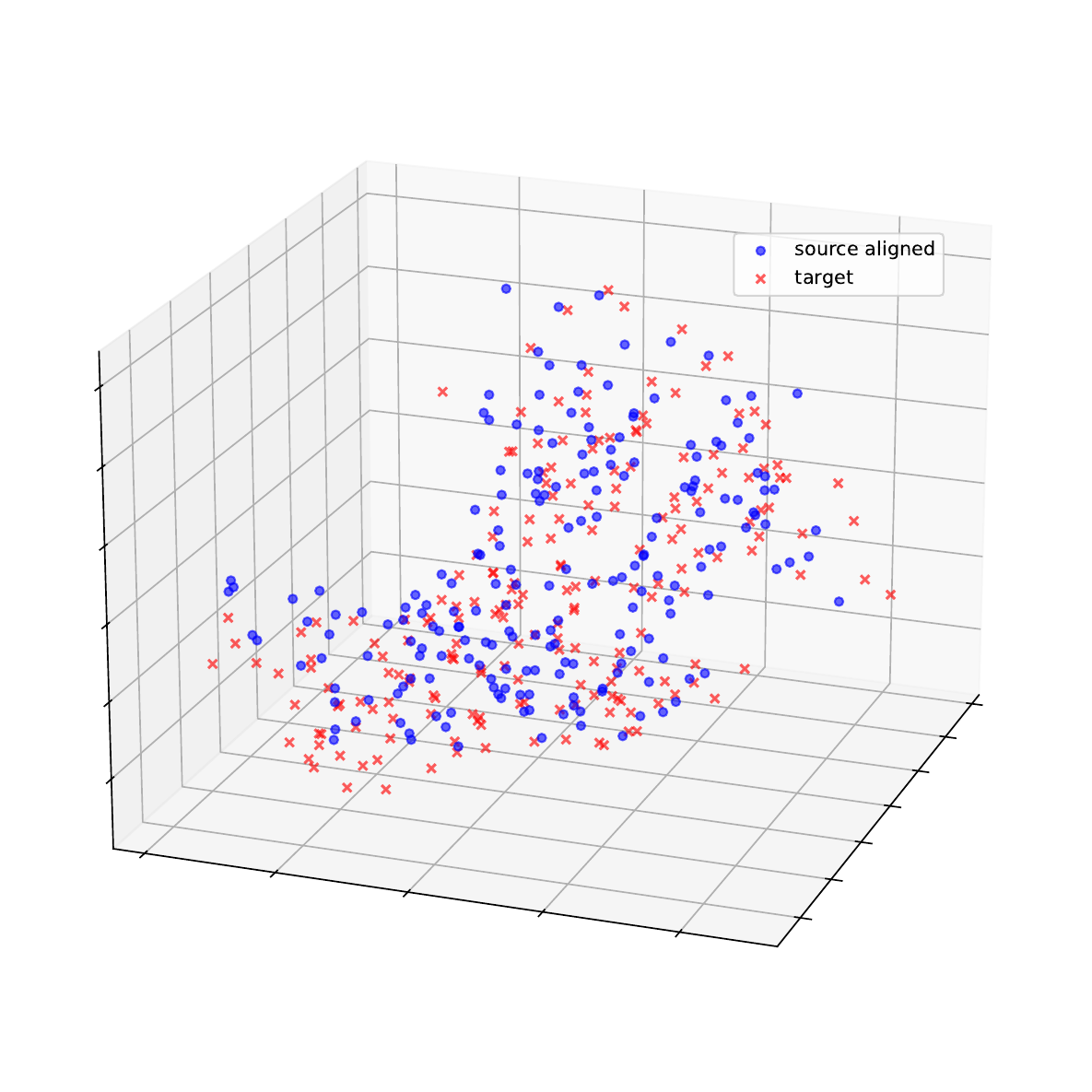}}
    \qquad
    \subfloat[Aligned with best 50\% correspondences]{\includegraphics[trim=1cm 2cm 1cm 3cm, clip=true, width=0.43\textwidth]{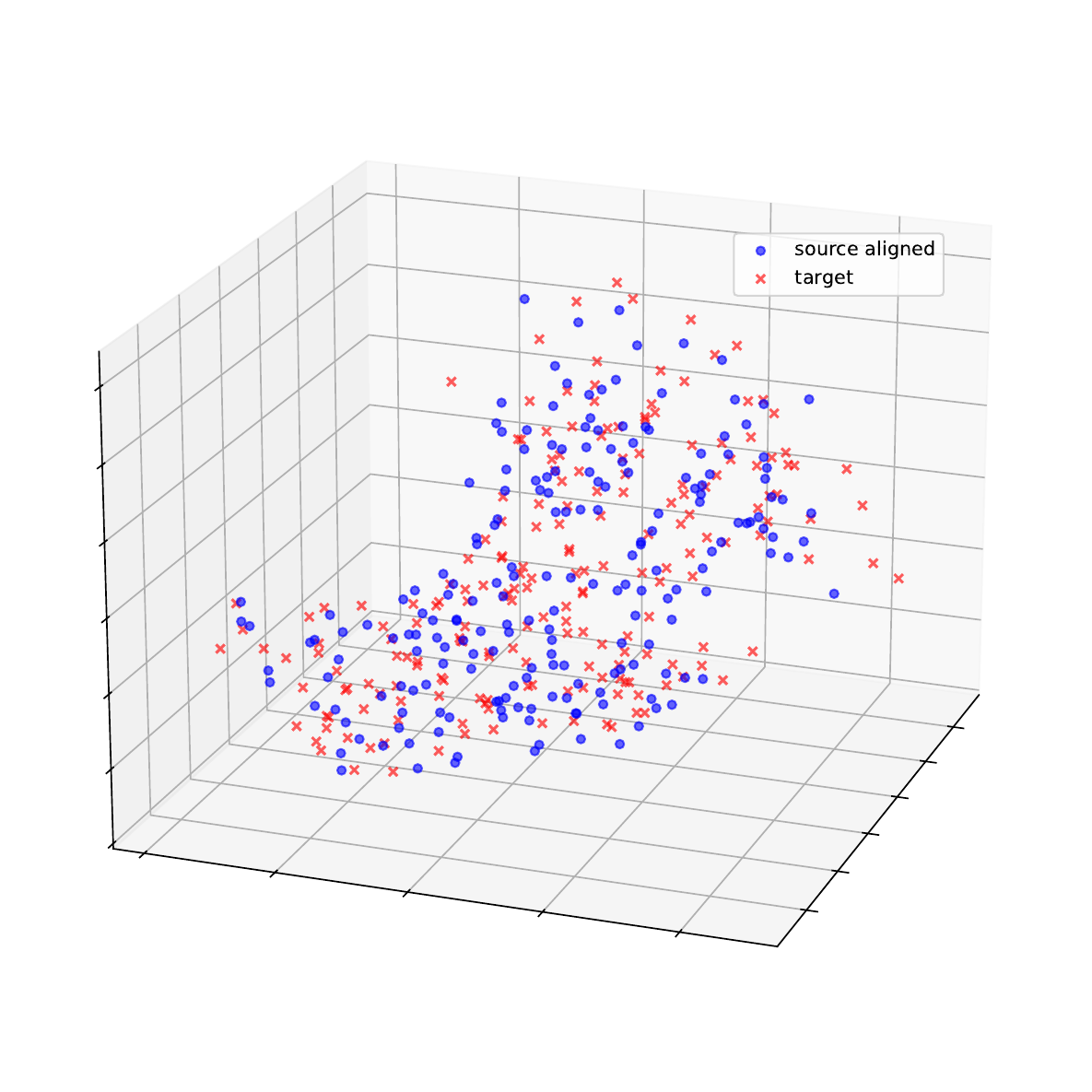}}
    \quad
    \subfloat[Aligned with best 10\% correspondences]{\includegraphics[trim=1cm 2cm 1cm 3cm, clip=true, width=0.43\textwidth]{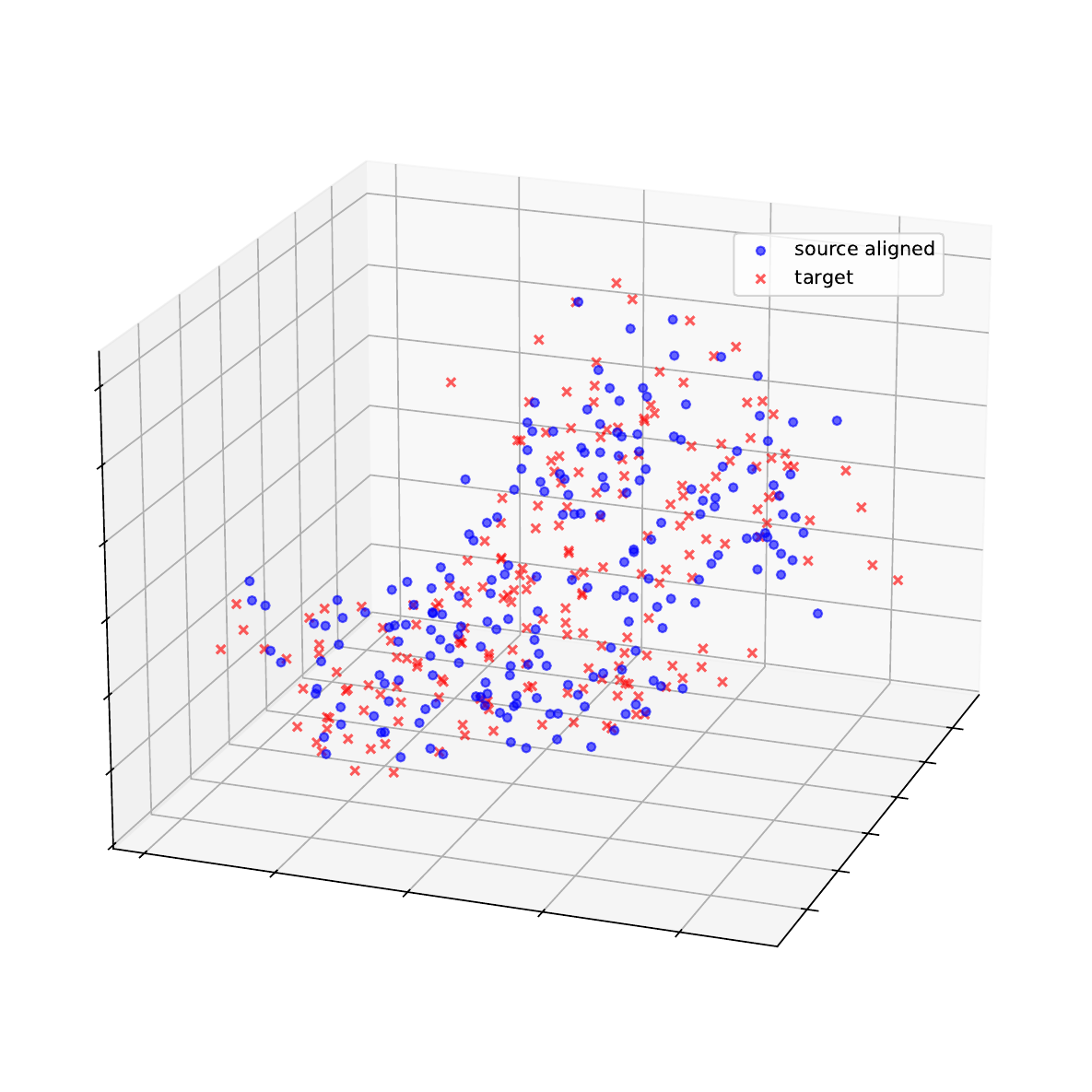}}
    \caption{Point clouds from the cryo-EM data of rabbit $\text{Ca}_{\text{v}}1.1$ and alignment results. (a) shows the point clouds sampled from the original density function and the density function after applying a rigid transformation, shown as source and target, respectively. For (b)-(c), we first apply Algorithm \ref{alg:distance_profile_matching} to obtain the correspondences and then estimate the rigid transformation using the orthogonal Procrustes problem. (b) shows the aligned point clouds using all the correspondences, while (c) and (d) use the best 50\% and 10\% of the correspondences, respectively, based on the threshold $\rho$ in Algorithm \ref{alg:distance_profile_matching}.}
    \label{fig:cryo-em2}
\end{figure}

\section{Conclusion}
We have proposed matching procedures based on distance profiles, which are easily implementable without having to solve nontrivial optimization problems over permutations. Most importantly, we have shown that the proposed method is provably robust to outlier components and noise under some model assumptions that we believe capture essential aspects of empirical applications. To the best of our knowledge, this is the first work that provides theoretical guarantees for the robustness of matching procedures based on distance profiles in the context of object matching. Also, we have derived the new statistical complexity result of the distance profile distance---the lower bound to the Gromov-Wasserstein distance---that has provided insights into how to use the distance profile matching procedures in practice. We have also demonstrated the empirical performance of the proposed methods through simulations and applications to real data.

We suggest several directions for future research. On the theoretical side, it would be interesting to consider other types of matching guarantees in Theorem \ref{thm:outlier_robustness}. As mentioned in Section \ref{sec:simulations}, the matching accuracy \eqref{eq:mixture_matching_accuracy} may be practically more useful. Similarly, for Theorem \ref{thm:noise_stability}, it would be interesting to consider the probability that $\hat{\pi}$ is close to $\pi^\ast$ under a certain discrepancy, such as the Hamming distance, instead of the probability of $\hat{\pi} = \pi^\ast$. Also, one can explore a generalization based on mixtures beyond the Euclidean space, say, a general metric or Banach space. On the practical side, it would be interesting to delve deeper into the procedure of finding inliers, that is, the points from the common part of two objects. For interested readers, we discuss in Appendix \ref{sec:inliers} how the threshold $\rho$ of Algorithm \ref{alg:distance_profile_matching} can be chosen to identify the points from the common part at least in theory. However, developing a practical method applicable to real data is very challenging in general and needs further investigation. We leave these directions for future research.

\section*{Acknowledgments}
The work of YK was partially supported by DMS-2339439, DE-SC0022232, and Sloan research fellowship. The work of YH was partially supported by DE-SC0022232.

\appendix
\section{Proofs}
\label{sec:proofs}

\subsection{Proof of Theorem \ref{thm:W1_bound_sup_x_euclidean}}
%Proof of Theorem \ref{thm:W1_bound_sup_x_euclidean}
\begin{proof}
    For $x \in \cX$ and $r \in \R$, let $\phi_{x, r}(\cdot) = 1\{d_\cX(x, \cdot) \le r\}$, which is a function from $\cX$ to $\{0, 1\}$. Let $F_x$ be the cumulative distribution function of $\mu_x$. Then, $F_x(r) = \mu \phi_{x, r}$ by definition. Similarly, letting $\hat{F}_x$ be the cumulative distribution function of $\hat{\mu}_x$, we have $\hat{F}_x(r) = \hat{\mu} \phi_{x, r}$. Recalling that $W_1$ between two probability measures on $\R$ is the $L_1$ distance between their cumulative distribution functions, we have
    \begin{equation*}
        W_1(\mu_x, \hat{\mu}_x) = \int_{0}^{\infty} |F_x(r) - \hat{F}_x(r)| \, \mathrm{d} r = \int_{0}^{\infty} |\mu \phi_{x, r} - \hat{\mu} \phi_{x, r}| \, \mathrm{d} r \le \sup_{\phi \in \cF_x} |\mu \phi - \hat{\mu} \phi|,
    \end{equation*}
    where $\cF_x := \{\phi_{x, r} : r \ge 0\}$ for any $x \in \cX$. Therefore,
    \begin{equation*}
        \sup_{x \in \cX} W_1(\mu_x, \hat{\mu}_x) \le \sup_{x \in \cX} \sup_{\phi \in \cF_x} |\mu \phi - \hat{\mu} \phi| = \sup_{\phi \in \cF_\cX} |\mu \phi - \hat{\mu} \phi|,
    \end{equation*}
    where $\cF_\cX := \bigcup_{x \in \cX} \cF_x = \{\phi_{x, r} : x \in \cX, r \ge 0\}$. By Theorem 8.23 of \cite{vershynin_2018}, 
    \begin{equation}
        \label{eq:sup_W1_VC_upper_bound}
        \E\left[\sup_{x \in \cX} W_1(\mu_x, \hat{\mu}_x)\right] \le \E\left[\sup_{\phi \in \cF_\cX} |\mu \phi - \hat{\mu} \phi|\right] \le C \sqrt{\frac{\mathrm{vc}(\cF_\cX)}{n}}.
    \end{equation}
    where $C$ is an absolute constant that is independent of $x$, $n$, $\cX$, and $\mu$. Here, $\mathrm{vc}(\cF_\cX)$ is the Vapnik-Chervonenkis (VC) dimension of the function class $\cF_\cX$. As $\cX \subset \R^d$, this is upper bounded by the VC dimension of the set of all closed Euclidean balls in $\R^d$, which is $d + 1$; see \cite{dudley1979balls}. Hence, we obtain \eqref{eq:sup_W1_euclidean}.
\end{proof}

\subsection{Proof of Proposition \ref{prop:analytic_properties}}
%Proof of Proposition \ref{prop:analytic_properties}
\begin{proof}
    Define a map $S \colon \cX \to \R \times \R$ by letting $S(\cdot) = (d_\cX(x, \cdot), d_\cX(x', \cdot))$. Then, one can verify that $S_\# \mu$ is a coupling of $\mu_x$ and $\mu_{x'}$. Therefore,
    \begin{equation*}
        \begin{split}
            W_1(\mu_x, \mu_{x'}) 
            & = \inf_{\rho \in \Pi(\mu_x, \mu_{x'})} \int_{\R \times \R} |z_1 - z_2| \, \mathrm{d} \rho(z_1, z_2) \\
            & \le \int_{\R \times \R} |z_1 - z_2| \, \mathrm{d} S_\# \mu(z_1, z_2) \\
            & = \int_{\cX} |d_\cX(x, \bar{x}) - d_\cX(x', \bar{x})| \, \mathrm{d} \mu(\bar{x}) \\
            & \le d_\cX(x, x'),
        \end{split}
    \end{equation*}
    where the last inequality follows from the triangle inequality of $d_\cX$. 
\end{proof}

\subsection{Proof of Proposition \ref{prop:W1_distance_profiles_mixture}}
%Proof of Proposition \ref{prop:W1_distance_profiles_mixture}
\begin{proof}
    By definition, one can check that $\mu_x$ is also given as the mixture of the distance profiles $(\mu_k)_x$'s, namely, $\mu_x = \sum_{k = 1}^{t} p_k (\mu_k)_x$. Then, 
    \begin{equation}
        \label{eq:Wasserstein_bound_mixture}
        W_1(\mu_x, \mu_{x'}) = W_1\left(\sum_{k = 1}^{t} p_k (\mu_k)_x, \sum_{k = 1}^{t} p_k (\mu_k)_{x'}\right) \le \sum_{k = 1}^{t} p_k W_1((\mu_k)_x, (\mu_k)_{x'}),
    \end{equation}
    where the inequality follows from the fact that if $\gamma_k \in \Pi((\mu_k)_x, (\mu_k)_{x'})$ is a coupling between $(\mu_k)_x$ and $(\mu_k)_{x'}$, the convex combination $\sum_{k = 1}^{t} p_k \gamma_k$ is a valid coupling between $\mu_x$ and $\mu_{x'}$. As $\mu_k$ is rotationally invariant with respect to the center $\theta_k$, we can deduce that $(\mu_k)_x = (\mu_k)_{\theta_k + \|x - \theta_k\|_2 u}$ for any $u$ on the unit sphere. Hence, 
    \begin{equation}
        \label{eq:W1_distance_profiles_rotationally_invariant}
        \begin{split}
            W_1((\mu_k)_x, (\mu_k)_{x'}) 
            & = W_1\left((\mu_k)_{\theta_k + \|x - \theta_k\|_2 u}, (\mu_k)_{\theta_k + \|x' - \theta_k\|_2 u}\right) \\
            & \le \left\|\theta_k + \|x - \theta_k\|_2 u - \theta_k - \|x' - \theta_k\|_2 u\right\|_2 \\
            & = \left|\|x - \theta_k\|_2 - \|x' - \theta_k\|_2\right|,            
        \end{split}
    \end{equation}
    where the inequality uses Proposition \ref{prop:analytic_properties}. Combining \eqref{eq:Wasserstein_bound_mixture} and \eqref{eq:W1_distance_profiles_rotationally_invariant}, we obtain \eqref{eq:W1_distance_profiles_mixture}.
\end{proof}

\subsection{Proof of Theorem \ref{thm:convergence_TLB_mixture}}
%Proof of Theorem \ref{thm:convergence_TLB_mixture}
\begin{proof}
    By \eqref{eq:TLB_decomposition}, \eqref{eq:TLB_decomposition_fast_part}, and Theorem \ref{thm:W1_bound_sup_x_euclidean}, it suffices to show that $\E(*) \le O(n^{-1 / t})$. Define $\Phi \colon \R^d \to \R^t$ by letting $\Phi(x) = \left(\sqrt{p_1} \|x - \theta_1\|_2, \ldots, \sqrt{p_t} \|x - \theta_t\|_2\right)$. From Proposition \ref{prop:W1_distance_profiles_mixture}, we have
    \begin{equation*}
        W_1(\mu_x, \mu_{x'}) \le \sum_{k = 1}^{t} p_k \left|\|x - \theta_k\|_2 - \|x' - \theta_k\|_2\right| \le \|\Phi(x) - \Phi(x')\|_2,
    \end{equation*}
    where the last inequality follows from the Cauchy-Schwarz inequality. For any coupling $\gamma \in \Pi(\mu, \hat{\mu})$, note that $(\Phi \times \Phi)_\# \gamma \in \Pi(\Phi_\# \mu, \Phi_\# \hat{\mu})$, where $\Phi \times \Phi \colon \R^d \times \R^d \to \R^t \times \R^t$ is a map defined by $(\Phi \times \Phi)(x, x') = (\Phi(x), \Phi(x'))$. Conversely, for any coupling $\rho \in \Pi(\Phi_\# \mu, \Phi_\# \hat{\mu})$, we can find a coupling $\gamma \in \Pi(\mu, \hat{\mu})$ such that $\rho = (\Phi \times \Phi)_\# \gamma$ by Lemma 3.12 of \cite{Ambrosio2013}. In summary, 
    \begin{equation}
        \label{eq:coupling_change_of_variables}
        \Pi(\Phi_\# \mu, \Phi_\# \hat{\mu}) = \{(\Phi \times \Phi)_\# \gamma : \gamma \in \Pi(\mu, \hat{\mu})\}.
    \end{equation}
    Therefore,
    \begin{equation*}
        \begin{split}
            \inf_{\gamma \in \Pi(\mu, \hat{\mu})} \int_{\R^d \times \R^d} W_1(\mu_x, \mu_{x'}) \, \mathrm{d} \gamma(x, x') 
            & \le \inf_{\gamma \in \Pi(\mu, \hat{\mu})} \int_{\R^d \times \R^d} \|\Phi(x) - \Phi(x')\|_2 \, \mathrm{d} \gamma(x, x') \\
            & = \inf_{\gamma \in \Pi(\mu, \hat{\mu})} \int_{\R^t \times \R^t} \|z - z'\|_2 \, \mathrm{d} (\Phi \times \Phi)_\# \gamma(z, z')\\
            & = \inf_{\rho \in \Pi(\Phi_\# \mu, \Phi_\# \hat{\mu})} \int_{\R^t \times \R^t} \|z - z'\|_2 \, \mathrm{d} \rho(z, z') \\
            & = W_1(\Phi_\# \mu, \Phi_\# \hat{\mu}),
        \end{split}
    \end{equation*}
    where the first equality follows from the change of variables, and the second equality is due to \eqref{eq:coupling_change_of_variables}. As $\Phi_\# \hat{\mu}$ is equivalent to the empirical measure of $\Phi_\# \mu$ which is defined on $\R^t$, we conclude from the standard results \citep{dudley1969speed,fournier2015rate} that $\E(*) \le W_1(\Phi_\# \mu, \Phi_\# \hat{\mu}) \le O(n^{-1 / t})$.
\end{proof}

\subsection{Proof of Theorem \ref{thm:outlier_robustness}}
\label{sec:proof_outlier_robustness}
%Proof of Theorem \ref{thm:outlier_robustness}
\begin{proof}
    We analyze the deviation of \eqref{eq:W(i, j)} from \eqref{eq:W1_distance_profiles_centers}, namely, for any pair $(i, j) \in [n] \times [m]$, we derive an upper bound on the deviation $|W(i, j) - \overline{W}(t(i), s(j))|$. To this end, we define an intermediate quantity $W'(i, j)$ as follows: for any $(i, j) \in [n] \times [m]$,
    \begin{equation*}
        \begin{split}
            W'(i, j) 
            & := W_1\bigg(\frac{1}{n} \sum_{\ell = 1}^{n} \delta_{\|\theta_{t(i)} - \theta_{t(\ell)}\|_2}, \frac{1}{m} \sum_{\ell = 1}^{m} \delta_{\|\eta_{s(j)} - \eta_{s(\ell)}\|_2}\bigg) \\
            & = W_1\bigg(\sum_{k = 1}^{t} \frac{n_k}{n} \delta_{\|\theta_{t(i)} - \theta_k\|_2}, \sum_{k = 1}^{s} \frac{m_k}{m} \delta_{\|\eta_{s(j)} - \eta_k\|_2}\bigg),
        \end{split}
    \end{equation*}
    where we define $n_k = \sum_{\ell = 1}^{n} 1\{t(\ell) = k\}$ and $m_k = \sum_{\ell = 1}^{m} 1\{s(\ell) = k\}$. We first derive an upper bound on $|W'(i, j) - \overline{W}(t(i), s(j))|$ for any $(i, j) \in [n] \times [m]$. By the triangle inequality of $W_1$, we have
    \begin{equation*}
        \begin{split}
            & |W'(i, j) - \overline{W}(t(i), s(j))| \\
            & \le \underbrace{W_1\bigg(\sum_{k = 1}^{t} \frac{n_k}{n} \delta_{\|\theta_{t(i)} - \theta_k\|_2}, \sum_{k = 1}^{t} p_k \delta_{\|\theta_{t(i)} - \theta_k\|_2}\bigg)}_{=: (*)} + \underbrace{W_1\bigg(\sum_{k = 1}^{s} \frac{m_k}{m} \delta_{\|\eta_{s(j)} - \eta_k\|_2}, \sum_{k = 1}^{s} q_k \delta_{\|\eta_{s(j)} - \eta_k\|_2}\bigg)}_{=: (**)}.
        \end{split}
    \end{equation*}    
    For each $\alpha \in [K]$, order $\theta_1, \ldots, \theta_t$ by the distance to $\theta_\alpha$, namely, find $k_\alpha \colon [t] \to [t]$ such that $\|\theta_\alpha - \theta_{k_\alpha(1)}\|_2 \le \cdots \le \|\theta_\alpha - \theta_{k_\alpha(t)}\|_2$. Then, by Lemma \ref{lem:W1_bound}, we have
    \begin{equation*}
        \begin{split}
            (*)
            & \le \max_{\alpha \in [K]} W_1\bigg(\sum_{k = 1}^{t} \frac{n_k}{n} \delta_{\|\theta_\alpha - \theta_k\|_2}, \sum_{k = 1}^{t} p_k \delta_{\|\theta_\alpha - \theta_k\|_2}\bigg) \\
            & \le \max_{\alpha \in [K]} \|\theta_\alpha - \theta_{k_\alpha(t)}\|_2 \cdot \max_{j \in [t - 1]} \left|\sum_{k = 1}^{j} \left(\frac{n_{k_\alpha(k)}}{n} - p_{k_\alpha(k)}\right)\right| \\
            & \le R \cdot \max_{\alpha \in [K]} \max_{j \in [t - 1]} \left|\sum_{k = 1}^{j} \left(\frac{n_{k_\alpha(k)}}{n} - p_{k_\alpha(k)}\right)\right|
            =: R \cdot Z_1.
        \end{split}
    \end{equation*}
    Similarly, for each $\beta \in [K]$, order $\eta_1, \ldots, \eta_s$ by the distance to $\eta_\beta$, namely, find $\ell_\beta \colon [s] \to [s]$ such that $\|\eta_\beta - \eta_{\ell_\beta(1)}\|_2 \le \cdots \le \|\eta_\beta - \eta_{\ell_\beta(s)}\|_2$. Then, by Lemma \ref{lem:W1_bound}, we have
    \begin{equation*}
        \begin{split}
            (**)
            & \le \max_{\beta \in [K]} W_1\bigg(\sum_{k = 1}^{s} \frac{m_k}{m} \delta_{\|\eta_\beta - \eta_k\|_2}, \sum_{k = 1}^{s} q_k \delta_{\|\eta_\beta - \eta_k\|_2}\bigg) \\
            & \le \max_{\beta \in [K]} \|\eta_\beta - \eta_{\ell_\beta(s)}\|_2 \cdot \max_{j \in [s - 1]} \left|\sum_{k = 1}^{j} \left(\frac{m_{\ell_\beta(k)}}{m} - q_{\ell_\beta(k)}\right)\right| \\
            & \le R \cdot \max_{\beta \in [K]} \max_{j \in [s - 1]} \left|\sum_{k = 1}^{j} \left(\frac{m_{\ell_\beta(k)}}{m} - q_{\ell_\beta(k)}\right)\right|
            =: R \cdot Z_2.
        \end{split}
    \end{equation*}
    Hence, $|W'(i, j) - \overline{W}(t(i), s(j))| \le R \cdot (Z_1 + Z_2)$. Next, we bound $|W(i, j) - W'(i, j)|$ for any $(i, j) \in [n] \times [m]$. By the triangle inequality again, we have
    \begin{equation*}
        \begin{split}
            & |W(i, j) - W'(i, j)| \\
            & \le W_1\bigg(\frac{1}{n} \sum_{\ell = 1}^{n} \delta_{\|X_i - X_\ell\|_2}, \frac{1}{n} \sum_{\ell = 1}^{n} \delta_{\|\theta_{t(i)} - \theta_{t(\ell)}\|_2}\bigg) + W_1\bigg(\frac{1}{m} \sum_{\ell = 1}^{m} \delta_{\|Y_j - Y_\ell\|_2}, \frac{1}{m} \sum_{\ell = 1}^{m} \delta_{\|\eta_{s(j)} - \eta_{s(\ell)}\|_2}\bigg). 
        \end{split}
    \end{equation*}
    Meanwhile, notice that $X_i = \theta_{t(i)} + \xi_i$ and $Y_j = \eta_{s(j)} + \zeta_j$ for any $(i, j) \in [n] \times [m]$, where $\xi_i \, | \, t(i)$ and $\zeta_j \, | \, s(j)$ are independent sub-Gaussian random vectors. Recall that the Wasserstein-1 distance between two sets of points in $\R$ with the same cardinality satisfies the following:
    \begin{equation*}
        W_1\bigg(\frac{1}{n} \sum_{\ell = 1}^{n} \delta_{\|X_i - X_\ell\|_2}, \frac{1}{n} \sum_{\ell = 1}^{n} \delta_{\|\theta_{t(i)} - \theta_{t(\ell)}\|_2}\bigg) \le \frac{1}{n} \sum_{\ell = 1}^{n}\left|\|X_i - X_\ell\|_2 - \|\theta_{t(i)} - \theta_{t(\ell)}\|_2\right|.
    \end{equation*}
    From $|\|X_i - X_\ell\|_2 - \|\theta_{t(i)} - \theta_{t(\ell)}\|_2 = |\|\theta_{t(i)} - \theta_{t(\ell)} + \xi_i - \xi_\ell\|_2 - \|\theta_{t(i)} - \theta_{t(\ell)}\|_2| \le \|\xi_i - \xi_\ell\|_2$, we have 
    \begin{equation*}
        W_1\bigg(\frac{1}{n} \sum_{\ell = 1}^{n} \delta_{\|X_i - X_\ell\|_2}, \frac{1}{n} \sum_{\ell = 1}^{n} \delta_{\|\theta_{t(i)} - \theta_{t(\ell)}\|_2}\bigg) \le \frac{1}{n} \sum_{\ell = 1}^{n} \|\xi_i - \xi_\ell\|_2 \le \max_{i, \ell \in [n]} \|\xi_i - \xi_\ell\|_2 =: U.
    \end{equation*}
    Similarly, 
    \begin{equation*}
        W_1\bigg(\frac{1}{m} \sum_{\ell = 1}^{m} \delta_{\|Y_j - Y_\ell\|_2}, \frac{1}{m} \sum_{\ell = 1}^{m} \delta_{\|\eta_{s(j)} - \eta_{s(\ell)}\|_2}\bigg) \le \max_{j, \ell \in [m]} \|\zeta_j - \zeta_\ell\|_2 =: V.
    \end{equation*}
    Hence, $|W(i, j) - W'(i, j)| \le U + V$ and $|W(i, j) - \overline{W}(t(i), s(j))| \le R (Z_1 + Z_2) + U + V$. Now, fix $i \in [n]$ such that $t(i) \in [K]$. Then,
    \begin{equation}
        \label{eq:gap_and_deviation}
        \begin{split}
            & \min_{\substack{j \in [m] \\ s(j) \neq t(i)}} W(i, j) - \min_{\substack{j \in [m] \\ s(j) = t(i)}} W(i, j) \\
            & \ge \min_{\substack{j \in [m] \\ s(j) \neq t(i)}} \overline{W}(t(i), s(j)) - \min_{\substack{j \in [m] \\ s(j) = t(i)}} \overline{W}(t(i), s(j)) - 2 R (Z_1 + Z_2) - 2 (U + V) \\
            & \ge \min_{\beta \in [s] \backslash \{t(i)\}} \overline{W}(t(i), \beta) - \overline{W}(t(i), t(i)) - 2 R (Z_1 + Z_2) - 2 (U + V) \\
            & \ge \underbrace{\min_{\substack{\alpha \in [K] \\ \beta \in [s] \backslash \{\alpha\}}} \overline{W}(\alpha, \beta) - R \cdot \bigg(1 - \sum_{k \in [K]} p_k \bigg)}_{=: \Omega} - 2 R (Z_1 + Z_2) - 2 (U + V),
        \end{split}
    \end{equation}
    where the last inequality follows from Lemma \ref{lem:Wasserstein_separation}. Therefore, $2 R (Z_1 + Z_2) + 2 (U + V) < \Omega$ implies that $s(\pi(i)) = t(i)$ for any $i \in [n]$ such that $t(i) \in [K]$. We show that $\P(E) \le \delta$, where $E$ is the event $(2 R (Z_1 + Z_2) + 2 (U + V) \ge \Omega)$. By Lemma \ref{lem:multinomial_concentration}, we have 
    \begin{align*}
        \P\left(Z_1 \ge \frac{\Omega}{8 R}\right) & \le 2 K (t - 1) \exp\left(-\frac{2 n \Omega^2}{(8 R)^2}\right), \\
        \P\left(Z_2 \ge \frac{\Omega}{8 R}\right) & \le 2 K (s - 1) \exp\left(-\frac{2 m \Omega^2}{(8 R)^2}\right).
    \end{align*}
    Meanwhile, for $i \neq \ell$, note that $\xi_i - \xi_\ell \, | \, t(i), t(\ell)$ is a sub-Gaussian random vector with variance proxy $\sigma_{t(i)}^2 + \sigma_{t(\ell)}^2$. By Lemma \ref{lem:sub_Gaussian}, we have
    \begin{equation*}
        \P\left(\max_{i, \ell \in [n]} \|\xi_i - \xi_\ell\|_2 \ge \frac{\Omega}{8}\right) \le \sum_{i \neq \ell} \P\left(\|\xi_i - \xi_\ell\|_2 \ge \frac{\Omega}{8}\right) \le n^2 e^{2 d} \exp\left(-\frac{(\Omega / 8)^2}{16 \sigma^2}\right).
    \end{equation*}
    Similarly, we have
    \begin{equation*}
        \P\left(\max_{j, \ell \in [m]} \|\zeta_j - \zeta_\ell\|_2 \ge \frac{\Omega}{8}\right) \le m^2 e^{2 d} \exp\left(-\frac{(\Omega / 8)^2}{16 \tau^2}\right). 
    \end{equation*}
    Putting all together, we have
    \begin{equation*}
        \begin{split}
            \P(E) 
            & \le \P\left(Z_1 \ge \frac{\Omega}{4 R}\right) + \P\left(Z_2 \ge \frac{\Omega}{4 R}\right) + \P\left(U \ge \frac{\Omega}{8}\right) + \P\left(V \ge \frac{\Omega}{8}\right) \\
            & \le 2 K (t + s - 2) \exp\left(-\frac{2 (n \wedge m) \Omega^2}{(8 R)^2}\right) + 2 (n \vee m)^2 e^{2 d} \exp\left(-\frac{(\Omega / 8)^2}{16 \Gamma^2}\right) \\
            & \le \frac{\delta}{2} + \frac{\delta}{2} = \delta,
        \end{split}        
    \end{equation*}
    where the last inequality follows from \eqref{eq:separation_condition}. Hence, we have \eqref{eq:matching_guarantee}.
\end{proof}

\begin{remark}
    From the proof of Theorem \ref{thm:outlier_robustness}, one may have noticed that we can replace the left-hand side of \eqref{eq:separation_condition} with 
    \begin{equation*}
        \Omega_o := \min_{\alpha \in [K]} \left(\min_{\beta \in [s] \backslash \{\alpha\}} \overline{W}(\alpha, \beta) - \overline{W}(\alpha, \alpha)\right).
    \end{equation*}
    To see this, it suffices to observe that we can replace $\Omega$ in the last line of \eqref{eq:gap_and_deviation} with $\Omega_o$, and the rest of the proof follows the same with $\Omega_o$ in place of $\Omega$. Since $\Omega_o \ge \Omega$, the condition \eqref{eq:separation_condition} is a stronger condition than requiring $\Omega_o$ to be at least the right-hand side of \eqref{eq:separation_condition}. We chose to present \eqref{eq:separation_condition} as it is written in terms of the outlier proportion which is more interpretable. 
\end{remark}

Lastly, we provide the lemmas used in the proof of Theorem \ref{thm:outlier_robustness}.
\begin{lemma}
    \label{lem:Wasserstein_separation}
    For the centers $\theta_1, \ldots, \theta_t$ and $\eta_1, \ldots, \eta_s$ in Assumption \ref{a:mixture}, define $\overline{W}$ as in \eqref{eq:W1_distance_profiles_centers}. Then, for any $k \in [K]$, we have
    \begin{equation}
        \label{eq:matching_center_upper_bound}
        \overline{W}(\alpha, \alpha) \le \bigg(1 - \sum_{k \in [K]} p_k\bigg) \cdot \max_{(\alpha, \beta) \in [t] \times [s]} \|\theta_\alpha - \eta_\beta\|_2.
    \end{equation}
\end{lemma}

%Proof of Lemma \ref{lem:Wasserstein_separation}
\begin{proof}
    The Wasserstein-1 distance on $\mathscr{P}(\R)$ can be written as the $L^1$ distance between cumulative distribution functions. For $\alpha \in [K]$, we have
    \begin{equation*}
        \begin{split}
            \overline{W}(\alpha, \alpha) 
            & = \int_\R \left|\sum_{k = 1}^{t} p_k 1\{\|\theta_{\alpha} - \theta_k\|_2 \le z\} - \sum_{k = 1}^{s} q_k 1\{\|\eta_{\alpha} - \eta_k\|_2 \le z\}\right| \, \mathrm{d} z \\           
            & = \int_\R \left|\sum_{k \in [t] \backslash [K]} p_k 1\{\|\theta_{\alpha} - \theta_k\|_2 \le z\} - \sum_{k \in [s] \backslash [K]} q_k 1\{\|\theta_{\alpha} - \eta_k\|_2 \le z\}\right| \, \mathrm{d} z \\
            & = (1 - \lambda) W_1\left(\sum_{k \in [t] \backslash [K]} \tfrac{p_k}{1 - \lambda} \delta_{\|\theta_{\alpha} - \theta_k\|_2}, \sum_{k \in [s] \backslash [K]} \tfrac{q_k}{1 - \lambda} \delta_{\|\theta_{\alpha} - \eta_k\|_2}\right),
        \end{split}
    \end{equation*}
    where $\lambda = \sum_{k \in [K]} p_k$. Since $|\|\theta_\alpha - \theta_k\|_2 - \|\theta_\alpha - \eta_\ell\|_2| \le \|\theta_k - \eta_\ell\|_2$ for any $k \in [t] \backslash [K]$ and $\ell \in [s] \backslash [K]$, one can deduce that \eqref{eq:matching_center_upper_bound} holds.
\end{proof}

\begin{lemma}
    \label{lem:W1_bound}
    Let $z_1 \le \cdots \le z_t$ be any real numbers. For any $(p_1, \ldots, p_t), (q_1, \ldots, q_t) \in \Delta_t$, we have 
    \begin{equation}
        \label{eq:W1_bound}
        W_1\bigg(\sum_{k = 1}^{t} p_k \delta_{z_k}, \sum_{k = 1}^{t} q_k \delta_{z_k}\bigg) \le (z_t - z_1) \cdot \max_{j \in [t - 1]} \bigg|\sum_{k = 1}^{j} (p_k - q_k)\bigg|.
    \end{equation}
\end{lemma}
\begin{proof}
    Notice that
    \begin{equation*}
        W_1\bigg(\sum_{k = 1}^{t} p_k \delta_{z_k}, \sum_{k = 1}^{t} q_k \delta_{z_k}\bigg) = \int_{\R} \bigg|\sum_{k = 1}^{t} (p_k - q_k) 1\{z_k \le z\}\bigg| \, \mathrm{d} z = \sum_{j = 1}^{t - 1} \bigg|\sum_{k = 1}^{j} (p_k - q_k)\bigg| \cdot (z_{j + 1} - z_j),
    \end{equation*}
    where the last equality follows from 
    \begin{equation*}
        \bigg|\sum_{k = 1}^{t} (p_k - q_k) 1\{z_k \le z\}\bigg| = \sum_{j = 1}^{t - 1} \bigg|\sum_{k = 1}^{j} (p_k - q_k)\bigg| \cdot 1\{z_j \le z < z_{j + 1}\} \quad \forall z \in \R.
    \end{equation*}
    Therefore, we have \eqref{eq:W1_bound}.
\end{proof}

\subsection{Proof of Theorem \ref{thm:noise_stability}}
%Proof of Theorem \ref{thm:noise_stability}
\begin{proof}
    For the proof, we may assume that $\pi^\ast$ is the identity without loss of generality. First, letting $Q(\pi) = \sum_{i = 1}^{n} W_1(\hat{\mu}_i, \hat{\nu}_{\pi(i)})$, observe that
    \begin{equation*}
        \left(\hat{\pi} \neq \pi^\ast\right) \subset \bigcup_{\pi \in \cS_n \backslash \{\pi^\ast\}} \left(Q(\pi^\ast) \ge Q(\pi)\right).
    \end{equation*}
    The event $Q(\pi^\ast) \ge Q(\pi)$ is equivalent to 
    \begin{equation*}
        \sum_{i = 1}^{n} \left(W_1(\hat{\mu}_i, \hat{\nu}_i) - W_1(\hat{\mu}_{i}, \hat{\nu}_{\pi(i)})\right) = \sum_{i : \pi(i) \neq i} \left(W_1(\hat{\mu}_i, \hat{\nu}_i) - W_1(\hat{\mu}_{i}, \hat{\nu}_{\pi(i)})\right) \ge 0,
    \end{equation*}
    which is contained in the following events:
    \begin{equation*}
        \bigcup_{i : \pi(i) \neq i} \left(W_1(\hat{\mu}_i, \hat{\nu}_i) \ge W_1(\hat{\mu}_{i}, \hat{\nu}_{\pi(i)})\right) \subset \bigcup_{i = 1}^{n} \bigcup_{j \neq i} \left(W_1(\hat{\mu}_i, \hat{\nu}_i) \ge W_1(\hat{\mu}_i, \hat{\nu}_j)\right) 
    \end{equation*}
    Therefore, 
    \begin{equation}
        \label{eq:assigment-based_matching_fail}
        \left(\hat{\pi} \neq \pi^\ast\right) \subset \bigcup_{i = 1}^{n} \bigcup_{j \neq i} \left(W_1(\hat{\mu}_i, \hat{\nu}_i) \ge W_1(\hat{\mu}_i, \hat{\nu}_j)\right).
    \end{equation}
    Now, recall that we can write a rigid transformation $T \colon \R^d \to \R^d$ as $T(\theta) = R \theta + b$ for some orthonormal matrix $R \in \R^{d \times d}$ and some vector $b \in \R^d$. Hence, we have 
    \begin{align*}
        \hat{\mu}_i & = \frac{1}{n} \sum_{\ell = 1}^{n} \delta_{\|X_i - X_\ell\|_2} = \frac{1}{n} \sum_{\ell = 1}^{n} \delta_{\|\theta_i - \theta_\ell + \xi_i - \xi_\ell\|_2} = \frac{1}{n} \sum_{\ell = 1}^{n} \delta_{\|\theta_i - \theta_\ell + \xi_{i \ell}\|_2}, \\
        \hat{\nu}_i & = \frac{1}{n} \sum_{\ell = 1}^{n} \delta_{\|Y_i - Y_\ell\|_2} = \frac{1}{n} \sum_{\ell = 1}^{n} \delta_{\|R(\theta_i - \theta_\ell) + \zeta_i - \zeta_\ell\|_2} = \frac{1}{n} \sum_{\ell = 1}^{n} \delta_{\|\theta_i - \theta_\ell + \zeta_{i \ell}\|_2},
    \end{align*}
    where we let $\xi_{i \ell} = \xi_i - \xi_\ell$ and $\zeta_{i \ell} = R^\top (\zeta_i - \zeta_\ell)$. Next, recall that the Wasserstein-1 distance between two sets of points in $\R$ with the same cardinality satisfies the following:
    \begin{equation*}
        W_1(\hat{\mu}_i, \hat{\nu}_j) = \min_{\pi \in \cS_n} \frac{1}{n} \sum_{\ell = 1}^{n} \left|\|\theta_i - \theta_\ell + \xi_{i \ell}\|_2 - \|\theta_j - \theta_{\pi(\ell)} + \zeta_{j \pi(\ell)}\|_2\right|.
    \end{equation*}
    Note that   
    \begin{equation*}
        \begin{split}
            & \left|\left|\|\theta_i - \theta_\ell + \xi_{i \ell}\|_2 - \|\theta_j - \theta_{\pi(\ell)} + \zeta_{j \pi(\ell)}\|_2\right| - \left|\|\theta_i - \theta_\ell\|_2 - \|\theta_j - \theta_{\pi(\ell)}\|_2\right|\right| \\
            & \le \left|\|\theta_i - \theta_\ell + \xi_{i \ell}\|_2 - \|\theta_i - \theta_\ell\|_2\right| + \left|\|\theta_j - \theta_{\pi(\ell)}\|_2 - \|\theta_j - \theta_{\pi(\ell)} + \zeta_{j \pi(\ell)}\|_2\right| \\
            & \le \|\xi_{i \ell}\|_2 + \|\zeta_{j \pi(\ell)}\|_2 \\
            & \le (U + V),
        \end{split}
    \end{equation*}
    where we define $U = \max_{i, \ell \in [n]} \|\xi_{i \ell}\|_2$ and $V = \max_{j, \ell \in [n]} \|\zeta_{j \ell}\|_2$. As 
    \begin{equation*}
        W_1(\mu_i, \mu_j) = \min_{\pi \in \cS_n} \frac{1}{n} \sum_{\ell = 1}^{n} \left|\|\theta_i - \theta_\ell\|_2^2 - \|\theta_j - \theta_{\pi(\ell)}\|_2^2\right|,
    \end{equation*}
    we conclude that 
    \begin{equation*}
        |W_1(\hat{\mu}_i, \hat{\nu}_j) - W_1(\mu_i, \mu_j)| \le (U + V) \quad \forall i, j \in [n].
    \end{equation*}
    Therefore, if we can find $i, j \in [n]$ such that $i \neq j$ and $W_1(\hat{\mu}_i, \hat{\nu}_i) \ge W_1(\hat{\mu}_i, \hat{\nu}_j)$, we must have 
    \begin{equation*}
        (U + V) \ge W_1(\hat{\mu}_i, \hat{\nu}_i) \ge W_1(\hat{\mu}_i, \hat{\nu}_j) \ge W_1(\mu_i, \mu_j) - (U + V) \ge \Phi - (U + V).
    \end{equation*}
    Accordingly, by \eqref{eq:assigment-based_matching_fail}, we have 
    \begin{equation*}
        \P(\hat{\pi} \neq \pi^\ast) \le \P(2 (U + V) \ge \Phi) \le \P\left(U \ge \frac{\Phi}{4}\right) + \P\left(V \ge \frac{\Phi}{4}\right).
    \end{equation*}
    Since $\xi_{i \ell} = \xi_i - \xi_\ell$ is a sub-Gaussian random vector with variance proxy $\sigma_i^2 + \sigma_\ell^2$ for $i \neq \ell$, one can deduce from Lemma \ref{lem:sub_Gaussian} and \eqref{eq:noise_level_condition} that 
    \begin{equation*}
        \P\left(\max_{i, \ell \in [n]} \|\xi_{i \ell}\|_2 \ge \frac{\Phi}{4}\right) \le \sum_{i \neq \ell} \P\left(\|\xi_{i \ell}\|_2 \ge \frac{\Phi}{4}\right) \le n^2 e^{2 d} \exp\left(-\frac{(\Phi / 4)^2}{16 \sigma^2}\right) \le \frac{\delta}{2},
    \end{equation*}
    where $\sigma = \max\{\sigma_1, \ldots, \sigma_n\}$. Similarly, we have
    \begin{equation*}
        \P\left(\max_{j, \ell \in [n]} \|\zeta_{j \ell}\|_2 \ge \frac{\Phi}{4}\right) \le \frac{\delta}{2}.
    \end{equation*}
    Therefore, we have $\P(\hat{\pi} \neq \pi^\ast) \le \delta$.
\end{proof}

\section{Lemmas}
\label{sec:lemmas}
We have used the following lemmas on the concentration of the multinomial distribution and the sub-Gaussian distribution in the proofs of Theorems \ref{thm:outlier_robustness} and \ref{thm:noise_stability}. The former is an application of Proposition A.6.6 of \cite{vaart_wellner_1996}, while the latter is a standard result that can be deduced from the Chernoff bound; we provide the proof for completeness.

\begin{lemma}
    \label{lem:multinomial_concentration}
    Suppose $(z_1, \ldots, z_t)$ is from the multinomial distribution with the number of trials $n$ and the probability vector $(p_1, \ldots, p_t) \in \Delta_t$. Let $\cA \subset 2^{[t]}$ be a collection of subsets of $[t]$. Then, for any $\epsilon \in (0, 1)$, we have
    \begin{equation*}
        \P\left(\max_{A \in \cA} \left|\sum_{k \in A} \left(\frac{z_k}{n} - p_k\right)\right| \ge \epsilon\right) \le 2 |\cA| \cdot e^{-2 n \epsilon^2}.
    \end{equation*}
\end{lemma}

\begin{lemma}
    \label{lem:sub_Gaussian}
    We say a probability measure $\mu \in \mathscr{P}(\R^d)$ is sub-Gaussian with variance proxy $\sigma^2 > 0$ if the following holds for $X \sim \mu$:
    \begin{equation*}
        \E e^{u^\top (X - \E X)} \le \exp\left(\frac{\sigma^2 \|u\|_2^2}{2}\right) \quad \forall u \in \R^d.
    \end{equation*}
    In this case, for any $t \ge 0$, we have 
    \begin{equation}
        \label{eq:sub_Gaussian_concentration}
        \P\left(\|X - \E X\|_2 \ge t\right) \le e^{2 d} \exp\left(-\frac{t^2}{8 \sigma^2}\right).
    \end{equation}
\end{lemma}

\begin{proof}
    Without loss of generality, we may assume $\E X = 0$. Let $B_r = \{x \in \R^d : \|x\|_2 \le r\}$ be the closed ball of radius $r$ centered at the origin, and recall that $\|x\|_2 = r \cdot \sup_{u \in B_r} u^\top x$ for any $x \in \R^d$. Now, for $\epsilon \in (0, 1)$, let $N_\epsilon$ be a minimal $\epsilon$-covering of $B_1$. For any $u \in B_1$, we can find $z_u \in N_\epsilon$ such that $\|u - z_u\|_2 \le \epsilon$. Hence, for any $u \in B_1$, we have
    \begin{equation*}
        u^\top x \le (u - z_u)^\top x + z_u^\top x \le \sup_{z \in B_\epsilon} z^\top x + \max_{z \in N_\epsilon} z^\top x \le \epsilon \|x\|_2 + \max_{z \in N_\epsilon} z^\top x,
    \end{equation*}
    which implies 
    \begin{equation*}
        \|x\|_2 \le \frac{1}{1 - \epsilon} \max_{z \in N_\epsilon} z^\top x.
    \end{equation*}
    For any $z \in N_\epsilon \subset B_1$, note that $z^\top X$ is a sub-Gaussian random variable with variance proxy $\sigma^2$. Hence, for any $t \ge 0$, we have 
    \begin{equation*}
        \P(\|X\|_2 \ge t) 
        \le \P(\max_{z \in N_\epsilon} z^\top X \ge (1 - \epsilon) t) \le |N_\epsilon| \cdot \exp\left(-\frac{(1 - \epsilon)^2 t^2}{2 \sigma^2}\right).
    \end{equation*}
    By taking $\epsilon = 1 / 2$ and using $|N_\epsilon| \le (1 + 2 / \epsilon)^d$, we have \eqref{eq:sub_Gaussian_concentration}; we use $|N_{1 / 2}| \le 5^d \le e^{2 d}$.
\end{proof}

% \begin{lemma}
%     \label{lem:chi_squared_concentration}
%     Let $\chi^2(d)$ be the chi-squared random variable with $d$ degrees of freedom. Then,
%     \begin{equation*}
%         \P\left(\chi^2(d) \ge d + t\right) \le \exp\left(-\frac{t (t \wedge d)}{8 d}\right) \quad \forall t \ge 0.
%     \end{equation*}
%     Accordingly, for any $t \ge 2 d$, 
%     \begin{equation*}
%         \P\left(\chi^2(d) \ge t\right) \le \exp\left(-\frac{t}{16}\right).
%     \end{equation*}
% \end{lemma}

\section{Extensions}
\label{sec:extensions}

\subsection{Extension of Theorem \ref{thm:W1_bound_sup_x_euclidean} to General Metric Measure Spaces}
\label{sec:W1_bound_sup_x_general}
From \eqref{eq:sup_W1_VC_upper_bound} in the proof of Theorem \ref{thm:W1_bound_sup_x_euclidean}, one can deduce that Theorem \ref{thm:W1_bound_sup_x_euclidean} can be extended to any metric measure space $(\cX, d_\cX, \mu)$ by replacing $d + 1$ in \eqref{eq:sup_W1_euclidean} with the VC dimension of $\cF_\cX$ defined in the proof, which is the VC dimension of the set of all balls in $\cX$. Unfortunately, controlling the VC dimension of the set of all balls in a general metric measure space is highly nontrivial. A more straightforward quantity is the $\epsilon$-covering number $\cN(\epsilon, \cX, d_\cX)$, namely, the smallest possible number of $\epsilon$-balls to cover $\cX$. The following theorem derives an upper bound using the covering number for a general bounded metric measure space $(\cX, d_\cX, \mu)$. 

\begin{theorem}
    \label{thm:W1_bound_sup_x_general}
    Let $(\cX, d_\cX, \mu)$ be a metric measure space and $\hat{\mu} = \frac{1}{n} \sum_{i = 1}^{n} \delta_{X_i}$ be the empirical measure based on $X_1, \ldots, X_n$ that are i.i.d.\ from $\mu$. Suppose $\Delta = \sup_{x, x' \in \cX} d_\cX(x, x') < \infty$ and the support of $\mu$ is exactly $\cX$. Let $\cN(\epsilon, \cX, d_\cX)$ be the $\epsilon$-covering number of $\cX$ under $d_\cX$. Then,
    \begin{equation}
        \label{eq:E_sup_W1_upper_bound_chaining}
        \E\left[\sup_{x \in \cX} W_1(\mu_x, \hat{\mu}_x)\right] \le \frac{C}{\sqrt{n}} \int_{0}^{\Delta / 4} \sqrt{\frac{\Delta}{\epsilon} + \log \cN(\epsilon, \cX, d_\cX)} \, \mathrm{d} \epsilon,
    \end{equation}
    where $C$ is an absolute constant that is independent of $n$, $\cX$, and $\mu$.
\end{theorem}

%Proof of Theorem \ref{thm:W1_bound_sup_x_general}
\begin{proof}
    By the Kantorovich-Rubinstein theorem, e.g., Theorem 1.14 of \cite{villani_2003}, we have 
    \begin{equation*}
        W_1(\mu_x, \hat{\mu}_x) = \sup_{f \in \mathrm{Lip}_1} \left(\mu_x f - \hat{\mu}_x f\right),
    \end{equation*}
    where $\mathrm{Lip}_1 = \left\{f \colon [0, \Delta] \to \R \, | \, f(0) = 0 ~ \text{and} ~ f ~ \text{is} ~ 1\text{-Lipschitz}\right\}$. Here, we have used the fact that the supports of $\mu_x$ and $\hat{\mu}_x$ are contained in the interval $[0, \Delta]$ for any $x \in \cX$. Hence, 
    \begin{equation*}
        \sup_{x \in \cX} W_1(\mu_x, \hat{\mu}_x) = \sup_{x \in \cX} \sup_{f \in \mathrm{Lip}_1} \left(\mu_x f - \hat{\mu}_x f\right) = \sup_{\phi \in \cF} \left(\mu \phi - \hat{\mu} \phi\right),
    \end{equation*}
    where $\cF = \{f \circ d_{\cX, x} : f \in \mathrm{Lip}_1, x \in \cX\}$ which is a function class from $\cX$ to $\R$. Here, $d_{\cX, x} \colon \cX \to \R$ is defined by $d_{\cX, x}(\cdot) = d_\cX(x, \cdot)$. Using symmetrization and Dudley's chaining argument (see Chapter 8 of \cite{vershynin_2018}), one can derive the following:
    \begin{equation}
        \label{eq:chaining_bound}
        \E\left[\sup_{x \in \cX} W_1(\mu_x, \hat{\mu}_x)\right] = \E\left[\sup_{\phi \in \cF} \left(\mu \phi - \hat{\mu} \phi\right)\right] \le \frac{C}{\sqrt{n}} \int_{0}^{\Delta / 2} \sqrt{\log \cN(\epsilon, \cF, \|\cdot\|_\cX)} \, \mathrm{d} \epsilon,
    \end{equation}
    where $\cN(\epsilon, \cF, \|\cdot\|_\cX)$ denotes the $\epsilon$-covering number of $\cF$ with respect to $\|\cdot\|_\cX$, namely, the sup norm defined by $\|\phi\|_\cX = \sup_{x \in \cX} |\phi(x)|$ for any $\phi \colon \cX \to \R$. Here, we are using the fact that $\sup_{\phi \in \cF} \|\phi\|_\cX \le \Delta$ since $\sup_{z \in [0, \Delta]} |f(z)| \le \Delta$ for any $f \in \mathrm{Lip}_1$ by Lipschitzness. Now, noticing that $\cF$ is a function class consisting of composite functions, we upper bound $\cN(\epsilon, \cF, \|\cdot\|_\cX)$ by the product of the $\epsilon$-covering number of $\mathrm{Lip}_1$ with respect to the sup norm $\|\cdot\|_\infty$ on $[0, \Delta]$, that is, $\|f\|_\infty = \sup_{z \in [0, \Delta]} |f(z)|$, and the $\epsilon$-covering number of $\cX$ with respect to $d_\cX$. Concretely, we claim
    \begin{equation}
        \label{eq:covering_ineq_product}
        \cN(2 \epsilon, \cF, \|\cdot\|_\cX) \le \cN(\epsilon, \mathrm{Lip}_1, \|\cdot\|_\infty) \times \cN(\epsilon, \cX, d_\cX).
    \end{equation}
    To see this, let $\mathrm{Lip}_1^\epsilon \subset \mathrm{Lip}_1$ and $\cX^\epsilon \subset \cX$ be $\epsilon$-coverings of $\mathrm{Lip}_1$ and $\cX$, respectively. For $f \in \mathrm{Lip}_1$ and $x \in \cX$, we can find $f' \in \mathrm{Lip}_1^\epsilon$ and $x' \in \cX^\epsilon$ such that $\|f - f'\|_\infty \le \epsilon$ and $d_\cX(x, x') \le \epsilon$. Then, 
    \begin{equation*}
        \begin{split}
            \|f \circ d_{\cX, x} - f' \circ d_{\cX, x'}\|_\cX 
            & \le \|f \circ d_{\cX, x} - f \circ d_{\cX, x'}\|_\cX + \|f \circ d_{\cX, x'} - f' \circ d_{\cX, x'}\|_\cX \\
            & \le \|d_{\cX, x} - d_{\cX, x'}\|_\cX + \|f - f'\|_\infty \\
            & \le d_\cX(x, x') + \epsilon \\
            & \le 2 \epsilon.
        \end{split}
    \end{equation*}
    Hence, we obtain \eqref{eq:covering_ineq_product}. The standard upper bound on the covering number of $\mathrm{Lip}_1$ (see Exercise 8.2.7 of \cite{vershynin_2018}) tells
    \begin{equation}
        \label{eq:Lipschitz_class_covering_bound}
        \log \cN(\epsilon, \mathrm{Lip}_1, \|\cdot\|_\infty) \le \frac{c \Delta}{\epsilon},
    \end{equation}
    where $c$ is some absolute constant. Combining \eqref{eq:chaining_bound}, \eqref{eq:covering_ineq_product}, and \eqref{eq:Lipschitz_class_covering_bound}, we obtain \eqref{eq:E_sup_W1_upper_bound_chaining}.
\end{proof}

\begin{remark}
    If one assumes that $\cX$ is a compact subset of $\R^d$ and $d_\cX$ is the standard Euclidean distance in Theorem \ref{thm:W1_bound_sup_x_general}, the bounds coincide with those in Theorem \ref{thm:W1_bound_sup_x_euclidean}. To see this, it suffices to notice that $\cN(\epsilon, \cX, d_\cX)$ scales as $(1 + 2 / \epsilon)^d$; see Section 4.2 of \cite{vershynin_2018}. The key difference, however, is that Theorem \ref{thm:W1_bound_sup_x_general} requires compactness of the support of $\mu$, while Theorem \ref{thm:W1_bound_sup_x_euclidean} does not.
\end{remark}

\subsection{Extension of Theorem \ref{thm:outlier_robustness}: Identifying Points from the Common Part}
\label{sec:inliers}
Section \ref{sec:outlier_robustness} focuses on the matching of the points from the common part, namely, the first $K$ components. The matching guarantee ensures that all $X_i$'s such that $t(i) \in [K]$ are correctly matched, while ignoring the points from the outlier components. If we assume a stronger condition by replacing the left-hand side of \eqref{eq:separation_condition} with the following quantity
\begin{equation}
    \label{eq:mixture_centers_separation_stronger}
    \min_{\substack{\alpha \in [K] \\ \beta \in [s] \backslash \{\alpha\}}} \overline{W}(\alpha, \beta) \wedge \min_{\substack{\alpha \in [t] \backslash [K] \\ \beta \in [s]}} \overline{W}(\alpha, \beta) - R \cdot \bigg(1 - \sum_{k \in [K]} p_k \bigg),
\end{equation}
we can show that
\begin{equation}
    \label{eq:threshold_range}
    \max_{\substack{i \in [n] \\ t(i) \in [K]}} W(i, \pi(i)) < \min_{\substack{i \in [n] \\ t(i) \notin [K]}} W(i, \pi(i))
\end{equation}
holds with probability at least $1 - \delta$. This means that if the threshold $\rho$ in Algorithm \ref{alg:distance_profile_matching} is chosen between the two quantities in \eqref{eq:threshold_range}, the output $I$ of the algorithm exactly coincides with the set $\{i \in [n] : t(i) \in [K]\}$, thereby identifying the points from the common part. 

Let us prove \eqref{eq:threshold_range}. To this end, first recall from the proof of Theorem \ref{thm:outlier_robustness} that the event $C = (2 R (Z_1 + Z_2) + 2 (U + V) < \Omega)$ implies $s(\pi(i)) = t(i)$ for any $i \in [n]$ such that $t(i) \in [K]$. Hence, on the event $C$, one can deduce that for such $i \in [n]$, we have
\begin{equation*}
    W(i, \pi(i)) \le \overline{W}(t(i), t(i)) + R (Z_1 + Z_2) + U + V,
\end{equation*}
which follows from $|W(i, j) - \overline{W}(t(i), s(j))| \le R (Z_1 + Z_2) + (U + V)$ for any $i, j \in [n]$. Therefore, invoking Lemma \ref{lem:Wasserstein_separation}, we have
\begin{equation*}
    \max_{\substack{i \in [n] \\ t(i) \in [K]}} W(i, \pi(i)) \le R \cdot \bigg(1 - \sum_{k \in [K]} p_k \bigg) + R (Z_1 + Z_2) + (U + V).
\end{equation*}
Meanwhile, 
\begin{equation*}
    \begin{split}
        \min_{\substack{i \in [n] \\ t(i) \notin [K]}} W(i, \pi(i)) 
        & \ge \min_{\substack{i \in [n] \\ t(i) \notin [K]}} \overline{W}(t(i), s(\pi(i))) - R (Z_1 + Z_2) - (U + V) \\
        & \ge \min_{\substack{\alpha \in [t] \backslash [K] \\ \beta \in [s]}} \overline{W}(\alpha, \beta) - R (Z_1 + Z_2) - (U + V).
    \end{split}
\end{equation*}
Now, let $C'$ be the event $(2 R (Z_1 + Z_2) + 2 (U + V) < \Omega')$, where $\Omega'$ represents the quantity in \eqref{eq:mixture_centers_separation_stronger}. Clearly, we have $\Omega' \le \Omega$ and thus $C' \subset C$. Therefore, on the event $C'$, we have 
\begin{equation*}
    \min_{\substack{i \in [n] \\ t(i) \notin [K]}} W(i, \pi(i)) - \max_{\substack{i \in [n] \\ t(i) \in [K]}} W(i, \pi(i)) \ge \Omega' - 2 R (Z_1 + Z_2) - 2 (U + V) > 0.
\end{equation*}
Now, we have $\P(C') \ge 1 - \delta$ by repeating the proof of Theorem \ref{thm:outlier_robustness} with the left-hand side of \eqref{eq:separation_condition}, which is $\Omega$, replaced by $\Omega'$. Therefore, we have \eqref{eq:threshold_range} with probability at least $1 - \delta$.

\bibliographystyle{abbrvnat}
\bibliography{ref}

\end{document}